\documentclass[unsortedaddress,superscriptaddress,pra,notitlepage]{revtex4}

\usepackage{amsmath,amssymb}
\usepackage{braket}
\usepackage{theorem}
\usepackage{color}
\usepackage{amsfonts}
\usepackage[mathscr]{eucal}
\usepackage[hidelinks]{hyperref}
\usepackage{graphicx}
\usepackage{xcolor}
\usepackage{tikz}
\usepackage{scalerel}
\usepackage{bm}

\usetikzlibrary{calc,decorations.pathreplacing}
\usetikzlibrary{arrows,shapes}

\newtheorem{definition}{Definition}[section]

\newtheorem{thm}[definition]{Theorem}
\newtheorem{prop}[definition]{Proposition}
\newtheorem{proposition}[definition]{Proposition}
\newtheorem{lemma}[definition]{Lemma}

\newtheorem{rem}[definition]{Remark}

\newtheorem{cor}[definition]{Corollary}

\newtheorem{assumption}[definition]{Assumption}

\newenvironment{proof}[1][Proof]{\begin{trivlist}
\item[\hskip \labelsep {\bfseries #1}]}{\hfill$\Box$\end{trivlist}}

\newcommand{\nn}{\nonumber \\}

\def\vfi{\varphi}
\def\theta{\vartheta}
\def\hil{{\mathcal H}}
\def\kil{{\mathcal K}}
\def\lil{{\mathcal L}}

\def\A{{\mathcal A}}

\def\B{{\mathcal B}}

\def\C{{\mathcal C}}
\def\cD{{\mathcal D}}
\def\E{{\mathcal E}}
\def\F{{\mathcal F}}

\def\I{{\mathcal I}}

\def\M{\mathcal{M}}

\def\R{{\mathcal R}}
\def\S{{\mathcal S}}

\def\X{{\mathcal X}}
\def\Y{{\mathcal Y}}

\def\half{\frac{1}{2}}
\def\imp{\implies}

\def\ep{\varepsilon}
\def\epsilon{\varepsilon}
\def\bN{\mathbb{N}}
\def\bC{\mathbb{C}}

\def\bR{\mathbb{R}}

\def\bP{\mathbb{P}}

\def\bz{\left(}
\def\jz{\right)}

\def\inv{^{-1}}
\def\tinv{^{-\mathrm{T}}}

\def\egy{\mathbf 1}

\def\cd{\mathcal{I}}

\def\sa{\mathrm{sa}}

\def\N{\mathcal{N}}

\def\what{\widehat}

\def\trans{^{\mathrm{T}}}

\def\rho{\varrho}

\def\nn{\nonumber}

\def\CP{\mathrm{CP}}
\def\cp{\mathrm{CP}}

\def\cptp{\mathrm{CPTP}}

\def\p{_{\ge 0}}
\def\pne{_{\gneq 0}}
\def\pp{_{>0}}

\def\valt{\cdot}

\def\simplex{\mathcal{S}}
\def\gm{C}
\def\ops{A}

\def\nh{\mathcal{N}}
\def\ah{\mathcal{A}}
\def\nw{^{*}}
\def\KA{\mathrm{KA}}
\def\type{\mathcal{T}}

\def\dli{\underline{\mathrm{d}}}
\def\dls{\overline{\mathrm{d}}}

\def\sci{\underline{\mathrm{sc}}}
\def\scs{\overline{\mathrm{sc}}}

\def\sh{\tilde}
\def\rate{\kappa}

\newcommand{\ki}[1]{\emph{#1}}

\newcommand{\s}{\mbox{ }}
\newcommand{\ds}{\mbox{ }\mbox{ }}

\newcommand{\norm}[1]{\left\| #1\right\|}
\newcommand{\snorm}[1]{\| #1 \|}
\newcommand{\inner}[2]{\left\langle #1 , #2\right\rangle}
\newcommand{\sinner}[2]{\langle #1 , #2\rangle}

\newcommand{\vecc}[1]{\underline{#1}}

\newcommand{\diad}[2]{\left|#1\right\rangle\!\left\langle #2\right|}
\newcommand{\sdiad}[2]{|#1\rangle\langle #2|}

\newcommand{\pr}[1]{\diad{#1}{#1}}
\newcommand{\wtilde}[1]{\widetilde{#1}}

\newcommand{\ceil}[1]{\left\lceil #1\right\rceil}

\newcommand{\persp}[1]{\mathcal{P}_{#1}}

\newcommand{\acc}[2]{#1_{#2,\mathrm{ac}}}

\renewcommand{\theenumi}{(\roman{enumi})}

\makeatletter
\renewcommand{\p@enumii}{}
\makeatother

\DeclareMathOperator{\id}{id}
\DeclareMathOperator{\Tr}{Tr}

\DeclareMathOperator{\supp}{supp}

\DeclareMathOperator{\ran}{ran}
\DeclareMathOperator{\dom}{\mathcal{D}}

\DeclareMathOperator{\choi}{C}
\DeclareMathOperator{\spann}{span}

\DeclareMathOperator{\spec}{spec}

\DeclareMathOperator{\Ad}{Ad}

\DeclareMathOperator{\ootimes}{\otimes\ldots\otimes}

\DeclareMathOperator{\logn}{\widehat\log}

\DeclareMathOperator{\co}{co}
\DeclareMathOperator{\direct}{\mathrm{d}}
\DeclareMathOperator{\sconv}{\mathrm{sc}}
\DeclareMathOperator{\Gm}{G}
\DeclareMathOperator{\Am}{\mathbb{A}}
\DeclareMathOperator{\smap}{\mathbb{S}}
\DeclareMathOperator{\st}{s}
\DeclareMathOperator{\ad}{ad}
\DeclareMathOperator{\para}{par}
\DeclareMathOperator{\gen}{gen}

\DeclareMathOperator*{\medoplus}{\scalerel*{\oplus}{\textstyle\sum}}
\DeclareMathOperator*{\medvee}{\scalerel*{\vee}{\textstyle\sum}}

\DeclareMathOperator{\overlap}{overlap}

\begin{document}

\newcounter{szamlalo}

\title{Error bounds for composite quantum hypothesis testing and a new characterization of the weighted Kubo-Ando geometric means}

\author{P\'eter E.~Frenkel}
\email{frenkelp265@gmail.com}

\affiliation{Department of Algebra and Number Theory, Institute of Mathematics, E\"otv\"os Lor\'and University, P\'azm\'any P\'eter s\'et\'any 1/C, Budapest, 1117 Hungary}

\affiliation{HUN-REN Alfr\'ed R\'enyi Institute of Mathematics, Re\'altanoda street 13-15, H-1053, Budapest}

\author{Mil\'an Mosonyi}
\email{milan.mosonyi@gmail.com}

\affiliation{Center for Quantum Technologies, National University of Singapore
Block S15, 3 Science Drive 2, Singapore 117543}

\affiliation{Department of Analysis and Operations Research, Institute of Mathematics,
Budapest University of Technology and Economics,
M\H uegyetem rkp.~3., H-1111 Budapest, Hungary}

\author{P\'eter Vrana}
\email{vranap@math.bme.hu}

\affiliation{Department of Algebra and Geometry, Institute of Mathematics,
Budapest University of Technology and Economics,
M\H uegyetem rkp.~3., H-1111 Budapest, Hungary}

\author{Mih\'aly Weiner}
\email{elektrubadur@gmail.com }

\affiliation{HUN-REN Alfr\'ed R\'enyi Institute of Mathematics, Re\'altanoda street 13-15, H-1053, Budapest}

\affiliation{Department of Analysis and Operations Research, Institute of Mathematics,
Budapest University of Technology and Economics,
M\H uegyetem rkp.~3., H-1111 Budapest, Hungary}

\begin{abstract}
\centerline{\textbf{Abstract}}
\vspace{.3cm}
The optimal error exponents of binary composite i.i.d.~state discrimination are trivially bounded by the worst-case pairwise exponents of discriminating individual elements of the sets representing the two hypotheses,
and in the finite-dimensional classical case, these bounds in fact give exact single-copy 
expressions for the error exponents. In contrast, in the non-commutative case, the optimal exponents
are only known to be expressible in terms of regularized divergences, resulting in formulas that, while conceptually relevant, are practically not very useful. In this paper, we develop further an approach initiated in 
[Mosonyi, Szil\'agyi, Weiner, IEEE Trans.~Inf.~Th.~68(2):1032--1067, 2022]  
to give improved single-copy bounds on the error exponents by comparing 
not only individual states from the two hypotheses, but also various 
unnormalized positive semi-definite operators associated to them. 
Here, we show a number of equivalent characterizations of such operators giving valid bounds, and show that 
in the commutative case, considering weighted geometric means of the states, 
and in the case of two states per hypothesis, considering weighted Kubo-Ando geometric means,
are optimal for this approach. 
As a result, we give a new characterization of the weighted Kubo-Ando geometric means as the only 
$2$-variable operator geometric means that are block additive, tensor multiplicative, 
continuous on monotone decreasing sequences,
and satisfy the 
arithmetic-geometric mean inequality. We also extend our results to composite quantum channel discrimination,
and show an analogous optimality property of the weighted Kubo-Ando geometric means of two quantum channels,
a notion that seems to be new. We extend this concept to defining the notion of superoperator 
perspective function and establish some of its basic properties, which may be of independent interest.

\vspace{.2cm}

\vspace{.2cm}


\end{abstract}

\maketitle

\tableofcontents

\section{Motivation and overview}
\label{sec:intro}

\subsection{Quantum state discrimination}
\label{sec:statedisc}

In the problem of composite i.i.d.~state discrimination, an experimenter is presented with several identical 
copies of a quantum system, all prepared in the same state, and they have to decide between two alternatives: that the 
state of the system is described by a density operator belonging to some 
subset $\nh$ of all density operators $\S(\hil)$ on the system's Hilbert space $\hil$, 
or by a density operator belonging to another subset $\ah$.
The most general decision scheme is given by a $2$-outcome measurement on the available copies of the system, described by a measurement operator $T\in\B(\hil^{\otimes n})_{[0,1]}$ 
(i.e., $T$ is an operator on $\hil^{\otimes n}$ satisfying $0\le T\le I$, where $n$ is the number of copies of the system): if the corresponding outcome occurs, the experimenter decides in favour of $\N$, whereas if the outcome corresponding to $I-T$ occurs, they choose $\A$.
The worst-case probabilities of the two possible errors that can occur (conventionally called type I and type II errors) are, respectively,
\begin{align*}
\alpha_n(T)&:=\sup_{\rho\in\nh}\Tr\rho^{\otimes n}(I-T),\\
\beta_n(T)&:=\sup_{\sigma\in\ah}\Tr\sigma^{\otimes n}T.
\end{align*}

In the asymptotic analysis of the problem, one is interested in the trade-off between these error probabilities in the limit $n\to+\infty$, which turns out to be quantifiable on the exponential scale by the following quantities, 
\begin{align}
\dli_r(\nh\|\ah):=\sup\left\{\liminf_{n\to+\infty}-\frac{1}{n}\log\alpha_n(T_n):\,
T_n\in\B(\hil^{\otimes n})_{[0,1]},\,n\in\bN,\,
\liminf_{n\to+\infty}-\frac{1}{n}\log\beta_n(T_n)>r\right\},\label{eq:direct exp def1}\\
\dls_r(\nh\|\ah):=\sup\left\{\limsup_{n\to+\infty}-\frac{1}{n}\log\alpha_n(T_n):\,
T_n\in\B(\hil^{\otimes n})_{[0,1]},\,n\in\bN,\,
\liminf_{n\to+\infty}-\frac{1}{n}\log\beta_n(T_n)> r\right\},\label{eq:direct exp def2}\\
\sci_r(\nh\|\ah):=\inf\left\{\liminf_{n\to+\infty}-\frac{1}{n}\log(1-\alpha_n(T_n)):\,
T_n\in\B(\hil^{\otimes n})_{[0,1]},\,n\in\bN,\,
\liminf_{n\to+\infty}-\frac{1}{n}\log\beta_n(T_n)> r\right\},\label{eq:sc exp de1}\\
\scs_r(\nh\|\ah):=\inf\left\{\limsup_{n\to+\infty}-\frac{1}{n}\log(1-\alpha_n(T_n)):\,
T_n\in\B(\hil^{\otimes n})_{[0,1]},\,n\in\bN,\,
\liminf_{n\to+\infty}-\frac{1}{n}\log\beta_n(T_n)> r\right\}.\label{eq:sc exp def2}
\end{align}
Here, $\dli_r(\nh\|\ah)$ and $\dls_r(\nh\|\ah)$ are called the \ki{direct exponents} corresponding to type II error rate $r$, and are relevant for $r$ below a threshold $r_0$, while
$\sci_r(\nh\|\ah)$ and $\scs_r(\nh\|\ah)$ are called the \ki{strong converse exponents}
corresponding to type II error rate $r$, and are relevant for $r$ above a (possibly different) threshold
$r_1\ge r_0$. Clearly, $\dli_r(\nh\|\ah)$ and $\sci_r(\nh\|\ah)$ cannot be both positive for a fixed $r$. 

In the simple case where $\nh=\{\rho\}$ and $\ah=\{\sigma\}$ consist of a single density operator each, the above exponents have been found to have single-copy expressions as 
\cite{ANSzV,Hayashicq,Nagaoka,MO}
\begin{align}
\dli_r(\{\rho\}\|\{\sigma\})=\dls_r(\{\rho\}\|\{\sigma\})=H_r(\rho\|\sigma):=\sup_{\alpha\in(0,1)}\frac{\alpha-1}{\alpha}\left[r-D_{\alpha}(\rho\|\sigma)\right],\label{eq:single-copy direct}\\
\sci_r(\{\rho\}\|\{\sigma\})=\scs_r(\{\rho\}\|\{\sigma\})
=
H\nw_r(\rho\|\sigma)
:=
\sup_{\alpha\ge 1}\frac{\alpha-1}{\alpha}\left[r-D\nw_{\alpha}(\rho\|\sigma)\right],\label{eq:single-copy sc}
\end{align}
where,
for any two positive semi-definite operators $A,B\in\B(\hil)\p$, 
\begin{align*}
D_{\alpha}(A\|B):=\frac{1}{\alpha-1}\log\Tr A^{\alpha}B^{1-\alpha}
\end{align*}
is their \ki{Petz-type R\'enyi $\alpha$-divergence} for $\alpha\in(0,1)$ \cite{P85},
and
\begin{align*}
D\nw_{\alpha}(A\|B):=\begin{cases}
\frac{1}{\alpha-1}\log\Tr \bz B^{\frac{1-\alpha}{2\alpha}}AB^{\frac{1-\alpha}{2\alpha}}\jz^\alpha,&\supp(A)\subseteq\supp(B),\\
+\infty,&\text{otherwise},
\end{cases}
\end{align*}
is their \ki{sandwiched R\'enyi $\alpha$-divergence} for $\alpha>1$ \cite{Renyi_new,WWY}.
Moreover,
$r<D(\rho\|\sigma)\imp \dli_r(\{\rho\}\|\{\sigma\})>0$, while
$r>D(\rho\|\sigma)\imp \sci_r(\{\rho\}\|\{\sigma\})>0$, where 
\begin{align*}
D(A\|B)
:=\begin{cases}
\Tr A(\log A-\log B),& \supp(A)\subseteq\supp(B),\\
+\infty,&\text{otherwise},
\end{cases}
\end{align*} 
is the Umegaki relative entropy of $A,B\in\B(\hil)\p$ \cite{Umegaki}.
These follow from \eqref{eq:single-copy direct}--\eqref{eq:single-copy sc} and the facts that 
$\sup_{\alpha\in(0,1)}D_{\alpha}(\rho\|\sigma)=D(\rho\|\sigma)=\inf_{\alpha>1}D_{\alpha}\nw(\rho\|\sigma)=:D_1\nw(\rho\|\sigma)$; see, e.g., 
\cite{MH,Renyi_new}.
The quantities $H_r(\rho\|\sigma)$ and $H_r\nw(\rho\|\sigma)$, defined for arbitrary PSD operators
$\rho,\sigma\in\B(\hil)\p$ as in \eqref{eq:single-copy direct}--\eqref{eq:single-copy sc},
are called the \ki{Hoeffding divergence/Hoeffding anti-divergence} of $\rho$ and $\sigma$ with parameter $r$, respectively.

In the known cases in the finite-dimensional classical setting, i.e., when all density operators in $\nh\cup\ah$ commute and $\hil$ is finite-dimensional,
the composite direct exponents still admit single-copy expressions as the worst-case pairwise exponents
\cite{AdvHyp,FuShen1996,FuShen1998,MWSz2020}
\begin{align*}
\dli_r(\nh\|\ah)=\dls_r(\nh\|\ah)=\inf_{\rho\in\nh,\,\sigma\in\ah}H_r(\rho\|\sigma).
\end{align*}
In the general composite case, however, we only have the inequalities
\begin{align}\label{eq:trivial bounds}
\dls_r(\nh\|\ah)\le\inf_{\rho\in\nh,\,\sigma\in\ah}H_r(\rho\|\sigma),
\ds\ds\ds
\sci_r(\nh\|\ah)\ge\sup_{\rho\in\nh,\,\sigma\in\ah}H\nw_r(\rho\|\sigma),
\end{align}
which may be strict in general \cite{BertaBrandaoHirche2017,MWSz2020},
and the only known universally valid expressions for these and related error exponents are in terms of regularized divergences involving arbitrarily many copies of the states from 
$\nh$ and $\ah$ \cite{BertaBrandaoHirche2017}, and therefore, while they are conceptually relevant, have limited practical value. 

A new approach to obtaining tighter single-copy bounds was put forward in \cite{MWSz2020} and elaborated further 
in \cite{BMV2024}.
On a very high level, the idea is to increase $\N$ and $\A$ so that taking the infimum/supremum over these increased sets instead of $\N$ and $\A$ in \eqref{eq:trivial bounds} still provides valid, and at least in some cases, strictly better bounds. 
To this end, for any $\R\subseteq\B(\hil)\p$, let 
\begin{align}
\C(\R):=\Big\{\gm\in\B(\hil)\p:&\,\forall n\in\bN\ds\exists\kappa_n\in(0,+\infty)\text{ s.t. }
\Tr \gm^{\otimes n}T\le\kappa_n\sup_{\omega\in\R}\Tr\omega^{\otimes n}T,\ds T\in\B(\hil^{\otimes n})_{[0,1]},\nn\\
&\text{and}\ds\limsup_{n\to+\infty}\frac{1}{n}\log\kappa_n\le 0\Big\}.\label{eq:CR def}
\end{align}
Then a straightforward computation verifies that 
any $\tilde C\in\C(\N)$ and $C\in\C(\A)$ yield single-copy bounds of the form
\begin{align}
\dli_r(\N\|\A)\le\dls_r(\N\|\A)&\le  
H_r(\tilde C\|C),
\label{direct upper}\\
\scs_r(\N\|\A)\ge\sci_r(\N\|\A)&\ge 
H_r\nw(\rho\|C),
\ds\ds\ds\rho\in\N.\label{sc lower}
\end{align}
We provide a detailed proof for readers' convenience in Appendix \ref{sec:improved bounds}.
The bounds in \eqref{direct upper}--\eqref{sc lower} may be interpreted as replacing the original composite i.i.d.~state discrimination problem
with various i.i.d.~discrimination problems that are simple, i.e., both the null and the alternative hypotheses 
are represented by single operators, but these operators are not necessarily normalized to be states anymore, and then relating the optimal exponents of these new simple discrimination problems to those of the original composite problem.

Tighter bounds on the error exponents than in \eqref{eq:trivial bounds} may be obtained 
from \eqref{direct upper}--\eqref{sc lower} by 
finding elements in $\C(\N)\setminus\N$ and in $\C(\A)\setminus\A$, and
one natural candidate here is to use operator geometric means that 
are tensor multiplicative and
satisfy the arithmetic-geometric mean 
inequality; in particular, the Kubo-Ando geometric mean \cite{KA} was used in 
\cite{MWSz2020} to this effect.
This motivates a systematic study of such operator geometric means, which seems to be a new direction
in matrix analysis, where operator geometric means are conventionally specified by other properties;
we will elaborate on this further in the next section. 
This approach is strongly underpinned by the findings of \cite{bunth2021equivariant}, where it was shown that in the case where 
$\ah$ is a closed and commutative subset of the positive definite elements in $\S(\hil)$ with a 
finite-dimensional Hilbert space $\hil$, 
optimizing over all weighted geometric means of the elements of $\ah$ in \eqref{sc lower} gives the exact strong converse exponent, i.e., 
\begin{align}
\sci_r(\N\|\A)=\scs_r(\N\|\A)&=
\sup_{\rho\in\N}\sup_{\nu\in\S(\ah)}\sup_{\alpha>1}\frac{\alpha-1}{\alpha}
\left[r-D_{\alpha}\nw\bz \rho\Big\| \exp\bz\int_{\ah}\log\sigma\,d\nu(\sigma)\jz\jz\right],
\label{eq:sc lower comm}
\end{align}
where $\S(\ah)$ stands for all probability measures on the Borel $\sigma$-algebra of $\ah$.

It is easy to see that the Petz-type R\'enyi $\alpha$-divergences are anti-monotone in both of their arguments for 
$\alpha\in(0,1)$, and the 
sandwiched R\'enyi $\alpha$-divergences are anti-monotone in their second argument for 
$\alpha>1$, i.e., for any $\rho_1\le\rho_2$ and $\sigma_1\le\sigma_2$,
\begin{align}\label{eq:Renyi anti-mon}
D_{\alpha}(\rho_1\|\sigma_1)\ge
D_{\alpha}(\rho_2\|\sigma_2),\ds\ds \alpha\in(0,1),\ds\ds\ds\ds
D\nw_{\alpha}(\rho_1\|\sigma_1)\ge
D\nw_{\alpha}(\rho_1\|\sigma_2),\ds\ds \alpha>1.
\end{align} 
(See, e.g., \cite{MH,Renyi_new} for details.)
For an arbitrary $\C\subseteq\B(\hil)\p$, let 
\begin{align*}
\max\C:=\left\{C\in\C:\,C'\in\C,\s C'\ge C\,\imp\,C'=C\right\}
\end{align*}
be the set of maximal elements in $\C$. Taking into account the anti-monotonicity of the R\'enyi divergences from \eqref{eq:Renyi anti-mon},
and also the trivial inclusion
\begin{align}\label{eq:trivial C}
\R\subseteq\C(\R),
\end{align}
valid for any $\R\subseteq\B(\hil)\p$, the bounds \eqref{direct upper}--\eqref{sc lower} lead to 
\begin{align}
\dls_r(\N\|\A)& 
\ds\le\ds  
\inf_{\tilde C\in\max\C(\N),C\in\max\C(\A)}H_r(\tilde C\|C)\nn\\
&\ds=\ds
\inf_{\tilde C\in\what\C(\N), C\in\what\C(\A)}H_r(\tilde C\|C)\nn\\
&\ds=\ds
\inf_{\tilde C\in\C(\N),C\in\C(\A)}H_r(\tilde C\|C)
\ds\le\ds
\inf_{\rho\in\N,\sigma\in\A}H_r(\rho\|\sigma),
\label{direct upper2}\\
\sci_r(\N\|\A)
&\ds\ge\ds 
\sup_{\rho\in\N,C\in\max\C(\A)}H_r\nw(\rho\|C)\nn\\
&\ds=\ds
\sup_{\rho\in\N,\what C\in\C(\A)}H_r\nw(\rho\|C)\nn\\
&\ds=\ds
\sup_{\rho\in\N, C\in\C(\A)}H_r\nw(\rho\|C)
\ds\ds\ds\ge\ds
\sup_{\rho\in\N,\sigma\in\A}H_r\nw(\rho\|\sigma),
\label{sc lower2}
\end{align}
where $\what\C(\N)$ and $\what\C(\A)$ are any sets satisfying
$\max\C(\N)\subseteq\what\C(\N)\subseteq\C(\N)$,
$\max\C(\A)\subseteq\what\C(\A)\subseteq\C(\A)$.

This leads to the following problem: 
\begin{align}\label{eq:holygrail}
\text{Given any $\R\subseteq\B(\hil)\p$, describe all elements of 
$\C(\R)$ / $\max\C(\R)$ / some $\what\C(\R)$.}
\end{align}

Of course, by this we mean a description that is efficiently verifiable in some sense; 
for instance, it would be desirable to have some criterion that decides whether a given 
operator $C$ belongs to $\C(\R)$ or $\max\C(\R)$ using only a single copy of $C$, as opposed to the definition 
given in \eqref{eq:CR def}, which uses infinitely many tensor powers of $C$, and, moreover, a finite set of criteria instead of the infinitely many inequalities required to hold for every test $T$.
The ultimate goal, however, is to find a way to actually construct 
a set $\what\C(\R)$ satisfying $\max\C(\R)\subseteq\what\C(\R)\subseteq\C(\R)$,
which is exactly what we achieve in the special cases when $\R$ is commutative, or it has only $2$ elements.

We will show in Section \ref{sec:classical} that when $\R$ is compact and commutative, a suitable
$\what\C(\R)$ is given by all the weighted geometric means of the elements of $\R$.
In Section \ref{sec:2-var}, we will show that when $\R$ consists of two incomparable positive definite operators then 
$\max\C(\R)$ is exactly the set of all weighted Kubo-Ando geometric means of the elements of $\R$.
In particular, the result for the commutative case combined with \eqref{eq:sc lower comm}
shows also that under the technical conditions imposed for \eqref{eq:sc lower comm}
(i.e., compactness of $\A$ and all elements of $\A$ being invertible), the first inequality in 
\eqref{sc lower2} holds as an equality.

Related to this problem, we note the simple alternative characterization of $\C(\R)$ as 
\begin{align}\label{eq:CR def2}
\C(\R)=\left\{\gm\in\B(\hil)\p:\,\Tr \gm^{\otimes n}T\le\sup_{\omega\in\R}\Tr\omega^{\otimes n}T,\ds T\in\B(\hil^{\otimes n})_{[0,1]},\ds n\in\bN\right\}.
\end{align}
Indeed, it is clear that if $C$ is in the set on the RHS of \eqref{eq:CR def2} then $C\in\C(\R)$. Vice versa, let 
$C\in\C(\R)$, and $(\kappa_n)_{n\in\bN}$ be a corresponding sequence as in \eqref{eq:CR def}.
Then for every $n,k\in\bN$, and any $T\in\B(\hil^{\otimes n})_{[0,1]}$,
\begin{align*}
(\Tr C^{\otimes n}T)^k=
\Tr C^{\otimes nk}T^{\otimes k}
\le
\kappa_{nk}\sup_{\omega\in\R}\Tr\omega^{\otimes nk}T^{\otimes k}
=
\kappa_{nk}\bz\sup_{\omega\in\R}\Tr\omega^{\otimes n}T\jz^k,
\end{align*}
and taking the $k$-th root and then the $\limsup$ in $k$ yield 
$\Tr C^{\otimes n}T\le\sup_{\omega\in\R}\Tr\omega^{\otimes n}T$, due to 
\begin{align*}
\limsup_{k\to+\infty}\frac{1}{k}\log\kappa_{nk}= 
n\limsup_{k\to+\infty}\frac{1}{nk}\log\kappa_{nk}\le 
n\limsup_{m\to+\infty}\frac{1}{m}\log\kappa_{m}\le 0.
\end{align*}
Thus, $C$ is in the set on the RHS of \eqref{eq:CR def2}.

We note that the above approach does not give an improvement over the trivial bounds \eqref{eq:trivial bounds}
when $\A$ is convex and compact (for the strong converse exponents)
or when $\A$ and $\N$ are both convex and compact (for the direct exponents). On the other hand, a strict improvement
can be attained already in the simplest possible case where $\N=\{\rho\}$ and $\A=\{\sigma_1,\sigma_2\}$; 
in the case of the strong converse exponent, a strict improvement is possible even for commuting states. 
We elaborate on these in Appendix \ref{sec:examples}.

A typical scenario for composite state discrimination with a finite null- or alternative hypothesis is when instead of exactly identifying one out of $r>2$ possible messages encoded into multiple copies of some quantum states
$\rho_1,\ldots,\rho_r$, the receiver only wants to know whether a specific message $k$
was sent or not, resulting in $\N=\{\rho_k\}$, $\A=\{\rho_i:\,i\in[r]\setminus\{k\}\}$. 
In a more realistic model, a small neighbourhood $B_i$ of states around $\rho_i$ may 
be assigned to each message, accounting for preparation errors on the side of the sender, yielding
$\N=B_k$, $\A=\cup_{i\in[r]\setminus\{k\}}B_i$. Note that even in this case, the alternative hypothesis will 
typically not be a convex set, and therefore the bounds 
\eqref{direct upper}--\eqref{sc lower} 
might provide an improvement over the trivial bounds \eqref{eq:trivial bounds}, 
and they indeed often do, as the examples in Appendix \ref{sec:examples} and in \cite{MWSz2020} demonstrate.

Finally, we leave open the problem of a general characterization of when 
the first inequality in 
\eqref{direct upper2} or in \eqref{sc lower2} holds as an equality.
In Appendix \ref{sec:disjoint}, we add to the known special cases by showing that both inequalities in 
\eqref{direct upper2} hold as equalities when 
both hypotheses are finite, and  the
states in the null hypotheses as well as the states in the alternative hypotheses have pairwise disjoint supports, i.e., 
$\supp(\rho)\cap\supp(\rho')=\{0\}$, $\rho,\rho'\in\N$, $\rho\ne\rho'$, and 
$\supp(\sigma)\cap\supp(\sigma')=\{0\}$, $\sigma,\sigma'\in\A$, $\sigma\ne\sigma'$.

\subsection{Operator geometric means}
\label{sec:operator geometric means}

For a measurable space $(\Y,\F)$, we call any function $\Gm$ satisfying 
the following properties a 
\ki{$\Y$-variable positive operator function:}
\begin{enumerate}
\renewcommand{\theenumi}{(\arabic{enumi})}
\item
$\Gm$ is defined on some non-empty $\dom(\Gm)\subseteq\cup_{d\in\bN}\B(\F,\bC^d)\p$, where
\begin{align*}
\B(\F,\bC^d)\p:=\{A:\,\Y\to\B(\bC^d)\p\text{ $\F$-measurable}\},
\end{align*}
with $\B(\bC^d)\p$ denoting the positive semi-definite linear operators on $\bC^d$, and 
\begin{align*}
A\in\dom_d(\Gm):=\dom(\Gm)\cap\B(\F,\bC^d)\p\ds\imp\ds \Gm(A)\in\B(\bC^d)\p.
\end{align*}
\item
$\Gm$ is covariant under unitaries, i.e., for any $A\in\dom_d(\Gm)$ and any 
unitary
$V$ on $\bC^{d}$, we have $(VA_yV^*)_{y\in\Y}\in\dom(\Gm)$ and 
\begin{align*}
\Gm\bz(VA_yV^*)_{y\in\Y}\jz=V\Gm\bz(A_y)_{y\in\Y}\jz V^*.
\end{align*}
\setcounter{szamlalo}{\value{enumi}}
\end{enumerate} 
Due to the unitary covariance, any such operator function can be uniquely extended to families 
of positive semi-definite operators $(A_y)_{y\in\Y}$ on a finite-dimensional Hilbert space $\hil$
such that there exists a unitary $V:\,\hil\to\bC^d$ with 
$(VA_yV^*)_{y\in\Y}\in\dom(\Gm)$, by defining 
$\Gm(A):=V^*\Gm\bz(VA_yV^*)_{y\in\Y}\jz V$.
We will use the notation $\dom_{\hil}(\Gm)$ for the set of all such $(A_y)_{y\in\Y}$.

We say that a $\Y$-variable positive operator function $\Gm$ is a \ki{$\nu$-weighted operator geometric mean} for some probability measure 
$\nu$ on $\F$, if the following holds:
\begin{enumerate}
\renewcommand{\theenumi}{(\arabic{enumi})}
\setcounter{enumi}{\value{szamlalo}}
\item
For any commuting family $A\in\B(\F,\bC^d)\p$ such that 
$A_y$ is positive definite for every $y\in\Y$ and 
$\Y\ni y\mapsto A_y$ is bounded and measurable,
we have $A\in\dom(\Gm)$ and
\begin{align*}
\Gm(A)=\exp\bz\int_{\Y}\log A_y\,d\nu(y)\jz.
\end{align*}
\setcounter{szamlalo}{\value{enumi}}
\end{enumerate} 
\renewcommand{\theenumi}{(\roman{enumi})}

For the purposes of quantum state discrimination, the relevant cases are where 
$\Y$ is finite, or it is a compact Hausdorff topological space and $\F$ is its Borel $\sigma$-algebra. In the latter case, it may also be natural to restrict the domain of a geometric mean to continuous functions 
on $\Y$.
Indeed, the typical setting for applications would be where $\Y$ is simply a closed subset of the state space of a finite-dimensional Hilbert space, and the map $\Y\ni y\mapsto A_y:=y$ is the identity.
It is also natural to consider this notion more generally for operators on infinite-dimensional Hilbert spaces. In this case, one needs to consider a number of technical issues that are not present in the finite-dimensional case. For instance, from the purely mathematical point of view, operator geometric means 
may be considered for families of bounded positive (semi-)definite operators, while 
the applications in quantum state discrimination call for their study for families of density operators
(positive trace-class operators with unit trace), and the focus on one of these or some other setting determines, e.g., the topology used in various definitions. We will only consider some 
questions in this level of generality for weak notions of means in Section \ref{sec:weak bounds}, where we 
give more precise definitions. 

For the applications in composite quantum state discrimination, we will be interested in 
positive operator functions satisfying some additional properties:
\begin{enumerate}
\setcounter{enumi}{\value{szamlalo}}
\renewcommand{\theenumi}{(\arabic{enumi})}
\item\label{tensor supermult}
\ki{Tensor multiplicativity:} For any finite-dimensional Hilbert space $\hil$, any
$A\in\dom_{\hil}(\Gm)$, and any $n\in\bN$, 
we have $A^{\otimes n}:=(A_y^{\otimes n})_{y\in\Y}\in \dom_{\hil^{\otimes n}}(\Gm)$ and 
\begin{align*}
\Gm(A^{\otimes n})= \Gm(A)^{\otimes n}.
\end{align*}

\item\label{tensor mult}
\ki{Strong tensor multiplicativity:} For any finite-dimensional Hilbert spaces $\hil^{(k)}$ and any
$A^{(k)}\in\dom_{\hil^{(k)}}(\Gm)$, $k=1,2$,
we have $A^{(1)}\otimes A^{(2)}:=(A^{(1)}_y\otimes A^{(2)}_y)_{y\in\Y}
\in \dom_{\hil^{(1)}\otimes\hil^{(2)}}(\Gm)$ and 
\begin{align*}
\Gm(A^{(1)}\otimes A^{(2)})= \Gm(A^{(1)})\otimes\Gm(A^{(2)}).
\end{align*}

\item\label{AM bound}
\ki{Arithmetic mean (AM) boundedness:} 
There exists a probability measure $\mu$ on $\F$ such that
\begin{align}\label{eq:AM bound}
\Gm(A)\le\int_{\Y}A_y\,d\mu(y),\ds\ds\ds A\in\dom(\Gm).
\end{align}
\setcounter{szamlalo}{\value{enumi}}
\end{enumerate}
When $\Gm=\Gm_{\nu}$ is a $\nu$-weighted operator geometric mean, we will say that it satisfies the 
\ki{arithmetic-geometric mean (AM-GM) inequality}, if \eqref{eq:AM bound} holds with $\mu=\nu$. 
In fact, instead of tensor multiplicativity and arithmetic mean boundedness, it will be sufficient to require that $\Gm$ satisfies the following weaker property:
\begin{enumerate}
\renewcommand{\theenumi}{(\arabic{enumi}')}
\setcounter{enumi}{\value{szamlalo}-1}
\item\label{as AM bound}
\ki{Asymptotic individual arithmetic mean boundedness:} For any $A\in\dom(\Gm)$ and any $n\in\bN$,
there exists a positive measure 
$\mu_n$ on $\F$ such that 
\begin{align*}
\Gm(A)^{\otimes n}\le\int_{\Y}A_y^{\otimes n}\,d\mu_n(y),\ds\ds\ds\ds\text{and}\ds\ds\ds\ds
\limsup_{n\to+\infty}\frac{1}{n}\log\mu_n(\Y)\le 0.
\end{align*}
\setcounter{enumi}{\value{szamlalo}}
\end{enumerate}
Note that here the measure $\mu_n$ may depend on $A$, hence the term ``individual''. Clearly, \ref{tensor supermult} and \ref{AM bound} imply \ref{as AM bound}, and it is equivalent to the existence of a probability measure $\bar\mu_n$ and an at most subexponentially growing prefactor $\kappa_n$ such that the approximate AM bound 
\begin{align*}
\Gm(A)^{\otimes n}\le\kappa_n\int_{\Y}A_y^{\otimes n}\,d\bar\mu_n(y)
\end{align*}
holds; if $\mu_n\ne 0$ then one can simply choose $\kappa_n:=\mu_n(\Y)$ and $\bar\mu_n:=\mu_n/\mu_n(\Y)$.
\renewcommand{\theenumi}{(\roman{enumi})}

Consider now a composite state discrimination problem specified by $\nh,\ah\subseteq\S(\hil)$ as in 
Section \ref{sec:statedisc}.
Clearly, if $G$ is a positive operator function with 
$\Y=\ah$, 
$\dom(\Gm)\ni (A_{\sigma}=\sigma)_{\sigma\in\ah}\equiv\ah$ 
and $\tilde G$ is a positive operator function with 
$\tilde\Y=\nh$, 
$\dom(\tilde\Gm)\ni(A_{\rho}=\rho)_{\rho\in\nh}\equiv\nh$
that both have the asymptotic individual arithmetic mean boundedness property, then
\begin{align}\label{eq:mean bound1}
\Tr\Gm\bz\A\jz^{\otimes n}T\le
\int_{\ah}\Tr\sigma^{\otimes n}T\,d\mu_n(\sigma)\le\mu_n(\Y)\sup_{\sigma\in\ah}\Tr\sigma^{\otimes n}T,\ds\ds\ds T\in\B(\hil^{\otimes n})_{[0,1]},\,n\in\bN,
\end{align} 
and similarly,
\begin{align}\label{eq:mean bound2}
\Tr\tilde\Gm(\nh)^{\otimes n}T\le
\int_{\nh}\Tr\rho^{\otimes n}T\,d\tilde\mu_n(\rho)\le\tilde\mu_n(\Y)\sup_{\rho\in\nh}\Tr\rho^{\otimes n}T,\ds\ds\ds T\in\B(\hil^{\otimes n})_{[0,1]},\,n\in\bN,
\end{align} 
leading to the bounds \eqref{direct upper}--\eqref{sc lower}
with $C:=\Gm(\ah)$ and $\tilde C:=\tilde\Gm(\nh)$.

The notion of operator geometric mean has been extensively studied in matrix analysis in the past few decades;
however, the focus in the usual approach is on other types of properties, like monotonicity of the mean in its arguments (in the positive semi-definite order) or joint concavity. 
It has been shown in the seminal work \cite{KA} that, essentially under these two conditions and the normalization $\Gm_{\nu}(X,X)=X$, 
for every probability distribution $\nu$ on a $2$-point set ($\Y=\{1,2\}$ for simplicity),
there exists a unique 
$\nu$-weighted operator geometric mean, given for positive definite operators as 
\begin{align*}
G_{\nu}^{\KA}(A_1,A_2)=A_2\#_{\nu(1)}A_1:=A_2^{1/2}\bz A_2^{-1/2}A_1A_2^{-1/2}\jz^{\nu(1)} A_2^{1/2};
\end{align*}
this is called the 
$\nu$-weighted Kubo-Ando geometric mean. It is easy to see that this mean is also tensor multiplicative and satisfies the AM-GM inequality.
Extensions of this mean to more variables 
may be obtained by iteration as $(\ldots(A_1\#_{t_1}A_2)\#_{t_2}\ldots )\#_{t_n}A_{n+1})$, or 
by various limiting, and other procedures \cite{AndoLiMathias2004,BhatiaHolbrook2006,BMP,Moakher_matrixmean,Lim_Palfia2012},
all of which satisfy tensor multiplicativity and the AM-GM inequality by construction; 
in particular, they are all suitable for the application in composite state discrimination described above. 
(See also \cite{AC2011,BGJ2019,BhatiaHolbrook2006,BJL2018,HiaiLim2020,KimLee2015,
Lawson_Lim2011,LawsonLim2014,mosonyi2022geometric,Petz_Temesi2005,PV_Hellinger} for further references on multivariate matrix geometric means).

For two variables, it is easy to construct 
operator geometric means in the sense of the above definition with properties (1)--(3) that are also manifestly tensor multiplicative (and also block additive; see below)
and are different from the Kubo-Ando means. For instance,
for any $t\in(0,1)$ and $z\in(0,+\infty)$, 
\begin{align}
G_{t,z}(A_1,A_2)
&:=
\bz A_1^{\frac{t}{2z}}A_2^{\frac{1-t}{z}}A_1^{\frac{t}{2z}}\jz^z,
\label{altmean1}\\
\wtilde G_{t,z}(A_1,A_2)
&:=
\bz A_2^{\frac{1-t}{2z}}A_1^{\frac{t}{z}}A_2^{\frac{1-t}{2z}}\jz^z,
\label{altmean2}\\
\what G_{t,z}(A_1,A_2)
&:=
\bz A_2^{\frac{1}{z}}\#_{t}A_1^{\frac{1}{z}}\jz^z,
\label{altmean3}\\
\what G_{t,+\infty}(A_1,A_2)
&:=\lim_{z\to+\infty}\what G_{t,z}(A_1,A_2)
=
e^{t \log A_1+(1-t)\log A_2}
\label{altmean4}
\end{align}
define $\nu=(t,1-t)$-weighted tensor multiplicative geometric means on pairs of positive definite operators, that are all different from $G_{\nu}^{\KA}$, with the exception of $\what G_{t,1}$. ($G_{t,z}$ and $\wtilde G_{t,z}$ are related to the family of quantum R\'enyi divergences considered in \cite{AD}, $\what G_{t,z}$ was considered in \cite{HP-GT} in connection to the Golden-Thompson inequality, and $\what G_{t,+\infty}$ is called the log-Euclidean geometric mean.)
It is natural to ask then if any of these means satisfy the AM-GM inequality, which has been answered in the negative very recently in \cite{NoAM-GM} (with the obvious exception of $\what G_{t,1}$) by constructing explicit counterexamples. 

In Section \ref{sec:2-var}, we give a complete
characterization of the weighted Kubo-Ando geometric means as the 
maximal $2$-variable operator geometric means satisfying 
asymptotic individual arithmetic mean boundedness
together with the following natural property:
\begin{enumerate}
\renewcommand{\theenumi}{(\arabic{enumi})}
\setcounter{enumi}{\value{szamlalo}}
\item
\ki{Block superadditivity:} 
If $(A_i,B_i)\in\dom_{\hil_i}(\Gm)$, $i\in[r]$,  
then $(\oplus_iA_i,\oplus_iB_i)\in \dom_{\oplus_i\hil_i}(\Gm)$, and 
\begin{align}\label{eq:block superadd def}
\Gm\bz\oplus_{i=1}^rA_i,\oplus_{i=1}^rB_i\jz
\ge
\oplus_{i=1}^r\Gm(A_i,B_i).
\end{align}
$\Gm$ is called \ki{block additive} if the above holds with an equality in \eqref{eq:block superadd def}.
\end{enumerate}
\renewcommand{\theenumi}{(\roman{enumi})}

\subsection{Overview of results}

As explained in Section \ref{sec:statedisc}, the state discrimination motivated mathematical problem that we want to solve is finding an easily verifiable characterization of the (maximal) elements of $\C(\R)$ for a given set $\R\subseteq\B(\hil)\p$. 
The membership of an operator $C$ in $\C(\R)$ is characterized by the infinite set of inequalities 
$\Tr C^{\otimes n}T\le\sup_{\omega\in\R}\Tr\omega^{\otimes n}T$, 
$T\in\B(\hil^{\otimes n})\p$, $n\in\bN$,
which can obviously be guaranteed to hold by imposing stronger inequalities. For instance, for any $n\in\bN$, 
the above inequality clearly follows if 
an arithmetic mean type bound
$\Tr C^{\otimes n}T\le\int_{\R}\Tr\omega^{\otimes n}T\,d\mu_n(\omega)$ holds for all $T$ with some probability measure $\mu_n$ that may depend on $n$ (see \eqref{eq:mean bound1} with $C:=\Gm(\ah)$). This in turn is further implied by a weak geometric mean type bound
$\Tr C^{\otimes n}T\le\exp\bz\int_{\R}\log\Tr\omega^{\otimes n}T\,d\nu_n(\omega)\jz$ for all $T$.
Somewhat surprisingly, these bounds become equivalent when required to hold for every $n\in\bN$, which is the main result of Section \ref{sec:weak bounds}.
In particular, this gives alternative characterizations of membership in $\C(\R)$ that 
may be easier to analyze, and indeed, it will be the geometric mean type bound that we will use 
to characterize the maximal elements of $\C(\R)$ for a $2$-element set $\R$ in 
Section \ref{sec:2-var}.

The proof of the above alternative characterizations builds on analogous results for the special case of commuting $\R$, which we prove in Section \ref{sec:classical}.
In particular, in this case, we can replace the above weak geometric mean type bound by a strong geometric mean bound, and as a result, we obtain that 
all the maximal elements of $\C(\R)$ are given by weighted geometric means of the elements of $\R$ when $\hil$ is finite dimensional and $\R$ is compact;
see Theorem \ref{thm:classical maxCR}.

In Section \ref{sec:2-var}, we 
use the characterization of $\C(\R)$ via weak geometric mean type bounds 
from Section \ref{sec:weak bounds}, together with Schur-Weyl duality, 
to characterize $\max\C(\R)$ for a two-element set $\R=\{\ops_1,\ops_2\}$,
solving the problem posed in \eqref{eq:holygrail} in the strongest possible sense. 
In particular, we show that if $\ops_1,\ops_2$ are positive definite and incomparable then 
\begin{align*}
\max\C(\{\ops_1,\ops_2\})=\{\ops_2\#_t\ops_1:\,t\in[0,1]\}.
\end{align*}
In the context of operator geometric means, our result provides
a new characterization of the weighted Kubo-Ando geometric means as the only 
$2$-variable operator geometric means 
that are block additive, strongly tensor multiplicative, continuous on monotone decreasing sequences,
and satisfy the AM-GM inequality. 
More precisely, we have the following:
\begin{thm}
Let $\Gm$ be a not constant zero $2$-variable positive operator function 
defined on $\cD(\Gm)=\cup_{d\in\bN}(\B(\bC^d)\p\times\B(\bC^d)\p)$ 
that is block additive, strongly tensor multiplicative,  and
asymptotically individually AM bounded.
Then there exists a $t\in[0,1]$ such that for any two positive definite $(\ops_1,\ops_2)\in\dom(\Gm)\pp$, 
\begin{align}\label{eq:KA equality2 intro}
\Gm(\ops_1,\ops_2)=\ops_2\#_t\ops_1.
\end{align}
If, moreover, $\Gm$ is continuous on monotone decreasing sequences, then \eqref{eq:KA equality2 intro} holds for every $(\ops_1,\ops_2)\in\dom(\Gm)$.
\end{thm}
We give a slightly more general statement in Theorem \ref{thm:KA char}, which we prove
from another new characterization that we give, namely, that the 
weighted Kubo-Ando geometric means are maximal among the operator geometric means that 
are block superadditive, 
tensor supermultiplicative, and individually AM bounded.
These are conceptually rather different from the usual characterization \cite{KA}, the main focus of which,
apart from tensor multiplicativity, is  
monotonicity in the arguments (and extra technical conditions like monotone continuity and the transformer inequality).
We note that we also use these properties of the Kubo-Ando means in our proof, but do not impose 
them on the map $(\ops_1,\ops_2)\mapsto C$.

In Section \ref{sec:channel disc}, we extend some of our results to composite channel discrimination.
We show an analogous way to give single-copy bounds on the error probabilities 
in terms of a single completely positive map $\E$ (taking the role of $C$ above), 
both for parallel and for adaptive discrimination strategies. 
We introduce the notion of the weighted Kubo-Ando geometric means for completely positive maps, and show 
their maximality for this approach  
in the $2$-variable case, analogously to the case of state discrimination. 
Related to this, we introduce the more general notion of superoperator perspective in 
Appendix \ref{sec:supop persp} and work out some
of its fundamental properties, as it may be of independent interest. This is supported by Appendix \ref{sec:oppersp}, where we collect some relevant facts about the notion of operator perspective function.

General definitions and technical tools are collected in Section \ref{sec:prelim}.

\section{Preliminaries}
\label{sec:prelim}

By $\log$ we denote the natural logarithm, and consider two natural extensions to 
$[0,+\infty]$ as
\begin{align*}
\log x:=\begin{cases}
-\infty,&x=0,\\
\log x,& x\in(0,+\infty),\\
+\infty,&x=+\infty,
\end{cases},\ds\ds\ds\ds\ds
\logn x:=\begin{cases}
0,&x=0,\\
\log x,& x\in(0,+\infty),\\
+\infty,&x=+\infty.
\end{cases}
\end{align*}
Correspondingly, we define $e^{-\infty}:=0$, $e^{+\infty}:=+\infty$, so that 
$e^{\log x}=x$ and $\log e^y=y$ hold for all $x\in[0,+\infty]$, $y\in[-\infty,+\infty]$. Moreover, we define
\begin{align}\label{eq:0on0}
0^0:=e^{0\log 0}=e^{\lim_{x\searrow 0}x\log x}=1.
\end{align}
Throughout the paper we use the convention
\begin{align}\label{zero times infty}
0\cdot(\pm\infty):=0.
\end{align}

We will follow the convention that the set of natural numbers $\bN$ does not include $0$.
For a natural number $r\in\bN$, $[r]:=\{1,\ldots,r\}$.

By a Hilbert space we always mean a complex Hilbert space.
We will denote the inner product on a Hilbert space by $\inner{\valt}{\valt}$ and follow the convention that 
it is linear in its second and conjugate linear in its first variable. We will also use the Dirac notation:
for any vectors $x,y$ in a Hilbert space $\hil$, the operator $\diad{y}{x}$ is defined by $\diad{y}{x}z:=\inner{x}{z}y$, $z\in\hil$.

For a Hilbert space $\hil$, let $\B(\hil)$ denote the set of all bounded linear operators on $\hil$, 
$\B(\hil)_{\sa}$ the set of self-adjoint operators, and 
for an interval $J\subseteq\bR$, let 
$\B(\hil)_{J}:=\{A\in\B(\hil)_{\sa}:\,\spec(A)\subseteq J\}$, i.e., 
the set of self-adjoint operators on $\hil$ with their spectra in $J$.
We will use the shorthand notations $\B(\hil)_{\ge 0}:=\B(\hil)_{[0,+\infty)}$
for the set of positive semi-definite (PSD) operators on $\hil$, and 
$\B(\hil)_{>0}:=\B(\hil)_{(0,+\infty)}$ for the set of 
positive definite operators. An inequality $A\le B$ between operators $A,B\in\B(\hil)$ is always interpreted in the L\"owner (or PSD) order, meaning $B-A\in\B(\hil)\p$. Elements of the set
\begin{align*}
\B(\hil)_{[0,1]}:=\left\{T\in\B(\hil)_{\sa}:\,0\le T\le I\right\}
\end{align*}
are called \ki{tests} on $\hil$.
The set of (orthogonal) projections on $\hil$ will be denoted by 
$\bP(\hil):=\B(\hil)_{\{0,1\}}=\{P\in\B(\hil)_{\sa}:\,P^2=P\}$.

The trace functional on trace-class operators and on PSD operators will be denoted by $\Tr$; the \ki{trace-norm} is then  
$\norm{X}_1:=\Tr|X|=\Tr\sqrt{X^*X}$, where $X$ is any trace-class operator.
The positive semi-definite trace-class operators with unit trace are called
\ki{density operators}, or \ki{states}, and  
$\S(\hil)$ denotes the set of all density operators on a Hilbert space $\hil$.

For a self-adjoint operator $A$ on a finite-dimensional Hilbert space, let $P^A_a:=\egy_{\{a\}}(A)$ denote the spectral projection of $A$
corresponding to the singleton $\{a\}\subset\bR$. 
(Here and henceforth $\egy_H$ stands for the characteristic (or indicator) function
of a set $H$.)
The projection onto the support of $A$ is $\sum_{a\ne 0}P^A_a$; in particular, if $A$ is 
positive semi-definite, it is equal to 
\begin{align}\label{eq:supppr}
A^0:=\lim_{t\searrow 0}A^t.\
\end{align}
(We note that this notation is not completely consistent with the convention \eqref{eq:0on0} if 
$A^0$ was to be interpreted in terms of functional calculus, but \eqref{eq:supppr} gives some justification for it, and we use it for convenience.)
In general, we follow the convention that real powers of a positive semi-definite operator $A$ are taken 
only on its support, i.e., for any $x\in\bR$, $A^x:=\sum_{a>0}a^x P^A_a$.
In particular, $A\inv:=\sum_{a>0}a\inv P^A_a$ stands for the generalized inverse of $A$,
and $A\inv A=AA\inv=A^0$. 

For PSD operators $A,B\in\B(\hil)\p$ on a finite-dimensional Hilbert space $\hil$, the 
\ki{absolutely continuous part} of $A$ with respect to $B$ is defined as 
\cite{AndoLebesgue1976,Anderson1971}
$\acc{A}{B}:=\max\{X\in\B(\hil)\p:\,X\le A,\,X^0\le B^0\}$. 
It is immediate from the definition that 
$(\acc{A}{B})^0=A^0\wedge B^0$, and that
$\acc{A}{B}=\acc{A}{P}$ for any projection $P$ such that $A^0\wedge B^0\le P\le B^0$.
For any such projection, we have the explicit expression \cite{Anderson1971}
\begin{align}\label{eq:ac explicit1}
\acc{A}{B}&=PAP-PA((I-P)A(I-P))\inv AP.
\end{align}
(Note that $\acc{A}{P}$ is also called the operator $A$ \ki{shorted} by the projection $P$.)
Moreover, 
\begin{align}\label{eq:ac explicit2}
\acc{A}{B}&=\bz(A^0\wedge B^0)A\inv (A^0\wedge B^0)\jz\inv,
\end{align}
where the inverses are in the above generalized sense.
This can be obtained from \eqref{eq:ac explicit1} using the well-known formula for the inverse of 
a $2\times 2$ block matrix;
see also \cite[page 436]{Krein1947SelfAdjointExtensionsI},
\cite{MasseyStojanoff2004}, or \cite[Appendix A]{mosonyi2026regularized} for alternative proofs.

For a continuous function $f:\,(0,+\infty)\to\bR$, the corresponding 
\ki{operator perspective function} 
$\persp{f}$ is defined on pairs 
of positive definite operators $A,B\in\B(\hil)\pp$ as 
\begin{align}\label{eq:oppersp def1}
\persp{f}(A,B):=B^{1/2}f\bz B^{-1/2}AB^{-1/2}\jz B^{1/2},
\end{align}
and extended to $A,B\in\B(\hil)\p$ as 
\begin{align}\label{eq:oppersp def2}
\persp{f}(A,B):=\lim_{\ep\searrow 0}\persp{f}(A+\ep I,B+\ep I),
\end{align}
whenever the limit exists \cite{Effros,ENG,HiaiMosonyi2017}, 
in which case we say that $\persp{f}(A,B)$ is defined.
It is easy to see that the limit in \eqref{eq:oppersp def2} always exists if $f$ is a non-negative operator monotone function.
It is also easy to see that for the \ki{transpose function} $\tilde f(x):=xf(1/x)$, $x>0$, we have
\begin{align}\label{persp trans}
\persp{f}(A,B)=\persp{\tilde f}(B,A),
\end{align}
whenever both sides are well-defined. 
If $\hil$ is finite dimensional and $f$ is a non-negative operator monotone function with 
$f(0^+)=0=\tilde f(0^+)$ then we have 
\begin{align}\label{eq:oppersp def3}
\persp{f}(A,B)=
\acc{B}{P}^{1/2}f\bz \acc{B}{P}^{-1/2}\acc{A}{P}\acc{B}{P}^{-1/2}\jz \acc{B}{P}^{1/2},
\end{align}
for any $A,B\in\B(\hil)\p$, where $P:=A^0\wedge B^0$, according to \cite[Theorem 17]{Kosaki_ac}.
In particular, for any $t\in[0,1]$, the 
\ki{$t$-weighted Kubo-Ando geometric mean} \cite{KA} is given by 
\begin{align}
B\#_{t}A:=P_{\id_{(0,+\infty)}^t}(A,B)
&:=
\lim_{\ep\searrow 0}(B+\ep I)^{1/2}\bz (B+\ep I)^{-1/2}(A+\ep I)(B+\ep I)^{-1/2}\jz^t(B+\ep I)^{1/2}\nn\\
&=
\begin{cases}
B,&t=0,\\
\acc{B}{P}^{1/2}\bz \acc{B}{P}^{-1/2}\acc{A}{P}\acc{B}{P}^{-1/2}\jz^t \acc{B}{P}^{1/2},& t\in(0,1),\\
A,&t=1,
\end{cases}\label{eq:KA cases}
\end{align}
where again $P:=A^0\wedge B^0$,
and it satisfies 
\begin{align*}
B\#_{t}A=A\#_{1-t}B,
\end{align*}
for any $A,B\in\B(\hil)\p$.
Moreover,
the $t$-weighted Kubo-Ando geometric mean satisfies the \ki{transformer identity} \cite{KA}:
\begin{align}\label{eq:transformer}
(XBX)\#_{t}(XAX)=X(B\#_{t}A)X
\end{align}
for any $A,B\in\B(\hil)\p$ and positive definite $X\in\B(\hil)\pp$.
It is straightforward from \eqref{eq:KA cases} that for fixed $A,B$,
\begin{align}\label{eq:KA cont1}
t\mapsto B\#_{t}A\ds\ds\text{is continuous on}\ds\ds (0,1),  
\end{align}
but not necessarily at $0$ and $1$, since
\begin{align}\label{eq:KA cont2}
\lim_{t\searrow 0} B\#_{t}A=\acc{B}{A}\le B=B\#_0A,\ds\ds\ds
\lim_{t\nearrow 1} B\#_{t}A=\acc{A}{B}\le A=B\#_1A,
\end{align}
and the first inequality is not an equality if $B^0\nleq A^0$,
and the second inequality is not an equality if $A^0\nleq B^0$.

For a Hausdorff topological space $\Y$, $\simplex(\Y)$ will denote the set of probability measures on the Borel 
$\sigma$-algebra of $\Y$. We will always consider it equipped with the weak$^*$-topology induced by the bounded 
continuous functions on $\Y$, i.e., a net $(\nu_i)_{i\in\I}$ in $\S(\Y)$ converges to some $\nu\in\S(\Y)$ if 
$\int_{\Y}f(y)\,d\nu_i(y)$ converges to $\int_{\Y}f(y)\,d\nu(y)$ for any bounded continuous function 
$f:\,\Y\to\bR$.
When $\Y$ is finite, the Borel $\sigma$-algebra is just the full power set of $\Y$, $\S(\Y)$ is naturally identified with the set of probability density functions on $\Y$, or equivalently, with the simplex spanned by the canonical basis vectors in $\bR^{\Y}$, and the weak$^*$-topology coincides with the 
topology induced by any norm on $\bR^{\Y}$.
For an arbitrary non-empty set 
$\X$, $\simplex_f(\X)$ will denote the set of finitely supported probability density functions on $\X$. 

The following statement from \cite[Theorem 5.2]{FarkasRevesz2006} is an extension 
of the minimax theorems due to Kneser \cite{Kneser} and Fan \cite{Fan}
to the case where $f$ can take the value $+\infty$. 

\begin{lemma}\label{lemma:KF+ minimax}
Let $X$ be a compact convex set in a topological vector space $V$ and $Y$ be a convex
subset of a vector space $W$. Let $h:\,X\times Y\to\bR\cup\{+\infty\}$ be such that
\smallskip

\s(i) $h(x,\valt)$ is concave on $Y$ for each $x\in X$, and
\smallskip

(ii) $h(\valt,y)$ is convex and lower semi-continuous  on $X$ for each $y\in Y$.
\smallskip

\noindent Then 
\begin{align}\label{minimax statement}
\inf_{x\in X}\sup_{y\in Y}h(x,y)=
\sup_{y\in Y}\inf_{x\in X}h(x,y),
\end{align}
and the infima in \eqref{minimax statement} can be replaced by minima.
\end{lemma}

The following simple observation will be useful in some applications of the above minimax theorem.

\begin{lemma}\label{lemma:usc}
Let $\A$ be a finite set, $\Y$ be a compact Hausdorff topological space, and 
$h:\,\A\times\Y\to[0,+\infty)$ be a function such that 
for every $a\in\A$, $h(a,\valt)$ is continuous. 
Then 
\begin{align}\label{eq:joint usc}
\simplex(\A)\times\simplex(\Y)\ni(\mu,\nu)\mapsto\sum_{a\in\A}\mu(a)\int_{\Y}\log h(a,y)\,d\nu(y)
=\int_{\Y}\sum_{a\in\A}\mu(a)\log h(a,y)\,d\nu(y)=:F(\mu,\nu)
\end{align}
is (jointly) upper semi-continuous in the product of 
any norm topology on $\S(\A)$ and the weak$^*$-topology on $\S(\Y)$
induced by the continuous functions on $\Y$.
\end{lemma}
\begin{proof}
For any $\ep>0$, 
\begin{align}
F_{\ep}(\mu,\nu):=\sum_{a\in\A}\mu(a)\int_{\Y}\log (h(a,y)+\ep)\,d\nu(y)
=\int_{\Y}\sum_{a\in\A}\mu(a)\log (h(a,y)+\ep)\,d\nu(y)
\end{align}
is easily seen to be jointly continuous on $\simplex(\A)\times\simplex(\Y)$. 
Since $F_{\ep}(\mu,\nu)<+\infty$ for every $\ep>0$, the monotone convergence theorem implies that 
$F(\mu,\nu)=\inf_{\ep>0}F_{\ep}(\mu,\nu)$, and hence $F$, as the infimum of continuous functions, is upper semicontinuous.
\end{proof}

\section{Classical case}
\label{sec:classical}

In this section we consider the problems outlined in Section \ref{sec:intro} in the setting where all involved operators 
commute. In this case, the $C^*$-algebra they generate in $\B(\hil)$ is isomorphic to the space of continuous functions on a compact Hausdorff space, and therefore we will formulate our results for functions, 
replacing the operator $C$ with a function $f$, and the 
elements of $\R$ (alternatively, the
operators $(A_y)_{y\in\Y}$) with a collection of functions $(g_y)_{y\in\Y}$, on some set $\X$. 
In fact, we consider a slightly more general setting in most of the section, where we do not assume anything about $\X$, whence, in particular, $f$ and $g_y$ are not assumed to be continuous functions, either, in lack of a topology to make sense of it. 

For a function $f:\,\X\to\bC$, its $n$-th tensor power is defined as
\begin{align*}
f^{\otimes n}(\vecc{x}):=\prod_{j=1}^n f(x_j),\ds\ds\ds \vecc{x}\in\X^n.
\end{align*}
For a function $g:\,\X\times\Y\to\bR$, we will denote $g_y(x):=g(x,y)$, $x\in\X$, $y\in\Y$, and by the above,
\begin{align*}
g_y^{\otimes n}(\vecc{x})=\prod_{j=1}^n g(x_j,y),\ds\ds\ds \vecc{x}\in\X^n,\,y\in\Y.
\end{align*}

\begin{thm}\label{thm:classical}
Let $\X,\Y$ be non-empty sets, $\F$ be a $\sigma$-algebra on $\Y$, and $f:\,\X\to[0,+\infty)$,
$g:\,\X\times\Y\to[0,+\infty)$ be functions such that for all 
$x\in\X$, $g(x,\valt)$ is $\F$-measurable.
Consider the following statements:
\begin{enumerate}
\item\label{classical1}
\emph{(Single-copy upper bound by a weighted geometric mean.)}
There exists a probability measure $\nu$ on $\F$ such that 
\begin{align}\label{eq:classical finiteness condition}
\forall x\in\X:\ds\exists\,\int_{\Y}\log g(x,y)\,d\nu(y)\in[-\infty,+\infty),
\end{align}
and
\begin{align}\label{eq:classical1}
f(x)\le\exp\bz\int_{\Y}\log g(x,y)\,d\nu(y)\jz, \ds\ds\ds x\in\X.
\end{align}

\item\label{classical1-2}
\emph{(Multi-copy upper bounds by a sequence of weighted geometric means.)}
There exists a sequence of probability measures $(\nu_n)_{n\in\bN}$ on $\F$ such that 
\begin{align}\label{eq:classical finiteness condition-2}
\forall n\in\bN,\,\forall \vecc{x}\in\X^n:\ds\exists\,\int_{\Y}\log g_y^{\otimes n}(\vecc{x})\,d\nu_n(y)\in[-\infty,+\infty),
\end{align}
and
\begin{align}\label{eq:classical1-2}
f^{\otimes n}(\vecc{x})\le\exp\bz\int_{\Y}\log g_y^{\otimes n}(\vecc{x})\,d\nu_n(y)\jz,\ds\ds\ds
\vecc{x}\in\X^n,\ds n\in\bN.
\end{align}

\item\label{classical2}
\emph{(Multi-copy upper bounds by a weighted arithmetic mean.)}
There exists a probability measure $\mu$ on $\F$ such that for every $n\in\bN$,
\begin{align}\label{eq:classical2}
f^{\otimes n}(\vecc{x})\le\int_{\Y}g_y^{\otimes n}(\vecc{x})\,d\mu(y),\ds\ds\ds
\vecc{x}\in\X^n.
\end{align}
\item\label{classical3}
\emph{(Multi-copy upper bounds by a sequence of weighted arithmetic means.)}
There exists a 
sequence of positive measures $(\mu_n)_{n\in\bN}$ on $\F$ such that 
\begin{align}\label{eq:subexponential weight}
\limsup_{n\to+\infty}\frac{1}{n}\log\mu_n(\Y)\le 0,
\end{align}
and 
\begin{align}\label{eq:classical3}
f^{\otimes n}(\vecc{x})\le\int_{\Y}g_y^{\otimes n}(\vecc{x})\,d\mu_n(y),\ds\ds\ds
\vecc{x}\in\X^n,\ds n\in\bN.
\end{align}
\end{enumerate}

We have \ref{classical1}$\iff$\ref{classical1-2}$\imp$\ref{classical2}$\imp$\ref{classical3}.
Moreover, if $\Y$ is a compact Hausdorff topological space, $\F$ is its Borel $\sigma$-algebra, and
for every $x\in\X$, 
$g(x,\valt)$ is continuous on $\Y$, then also 
\ref{classical3}$\imp$\ref{classical1} holds.
\end{thm}
\begin{proof}
Clearly, \ref{classical1} implies \ref{classical1-2} with $\nu_n=\nu$, $n\in\bN$,
and the converse implication is trivial. The implication
\ref{classical1}$\imp$\ref{classical2} 
follows by the convexity of the exponential function, and 
the implication \ref{classical2}$\imp$\ref{classical3} is trivial.
Hence, we are left to prove \ref{classical3}$\imp$\ref{classical1}.

Assume therefore that \ref{classical3} holds. 
Note that \eqref{eq:classical1} holds trivially when $f(x)=0$, and hence we only need to prove it for 
every $x\in\supp f:=\{z\in\X:\,f(z)>0\}$. 
Note also that \eqref{eq:classical3} 
implies that for every $\ep>0$, 
\begin{align}\label{eq:classical3-0}
f^{\otimes n}(\vecc{x})\le\int_{\Y}(g_y+\ep)^{\otimes n}(\vecc{x})\,d\mu_n(y),\ds\ds\ds
\vecc{x}\in\X^n,\ds n\in\bN,
\end{align}
or equivalently,
\begin{align}\label{eq:classical3-2}
1\le\int_{\Y}\tilde g_y^{\otimes n}(\vecc{x})\,d\mu_n(y),\ds\ds\ds
\vecc{x}\in(\supp f)^n, \ds n\in\bN,
\end{align}
where $\tilde g(x,y):=(g(x,y)+\ep)/f(x)$, $x\in\supp f$, $y\in\Y$.
We will first prove \eqref{eq:classical1} with $g+\ep$ in place of $g$, which is equivalent to   
\begin{align}\label{eq:classical1-4}
1\le\exp\bz\int_{\Y}\log \tilde g(x,y)\,d\nu(y)\jz, \ds\ds\ds x\in\supp f.
\end{align}

Let $x_1,\ldots,x_m\in\supp f$, and $r_1,\ldots,r_m$ be strictly positive rational numbers summing to $1$. We can write $r_i=a_i/a$ with some natural numbers $a_1,\ldots,a_m$ and 
$a:=\sum_{i=1}^m a_i$. For any $k\in\bN$, \eqref{eq:classical3-2} yields that 
\begin{align*}
1\le\int_{\Y} \prod_{j=1}^m \tilde g_y(x_j)^{ka_j}\,d\mu_{ka}(y)
\le
\bz\sup_{y\in\Y}\prod_{j=1}^m \tilde g_y(x_j)^{a_j}\jz^k\mu_{ka}(\Y).
\end{align*}
Taking the $k$-th root and then the limsup in $k$ yields
\begin{align*}
1\le\sup_{y\in\Y}\prod_{j=1}^m \tilde g_y(x_j)^{a_j},
\end{align*}
according to \eqref{eq:subexponential weight}. By taking the $\log$ of both sides, we get 
\begin{align}\label{eq:classical proof1}
0\le\sup_{y\in\Y}\sum_{j=1}^m r_j\log\tilde g(x_j,y).
\end{align}

Assume now that 
$\Y$ is a compact Hausdorff topological space, $\F$ is its Borel $\sigma$-algebra, and
for every $x\in\X$, 
$g(x,\valt)$ is continuous on $\Y$.
Then \eqref{eq:classical proof1} can be rewritten as 
\begin{align}\label{eq:classical proof1-2}
0\le
\sup_{\nu\in\simplex(\Y)}\int_{\Y}\sum_{j=1}^m r_j\log\tilde g(x_j,y)\,d\nu(y).
\end{align}
Moreover,
the conditions of Lemma \ref{lemma:usc} are satisfied with $\A:=[m]$ and 
$h(i,y):=\tilde g(x_i,y)$,
and therefore
$\S([m])\times\S(\Y)\ni(r,\nu)\mapsto \int_{\Y}\sum_{j=1}^m r_j\log\tilde g(x_j,y)\,d\nu(y)$ is jointly upper semi-continuous, which in turn implies that 
$\S([m])\ni r\mapsto\sup_{\nu\in\S(\Y)}\int_{\Y}\sum_{j=1}^m r_j\log\tilde g(x_j,y)\,d\nu(y)$ is upper semi-continuous (see, e.g., \cite[Lemma 2.3]{mosonyi2022geometric}).
Let $r=(r_1,\ldots,r_m)$ be an arbitrary probability distribution on $[m]$, and 
let $r^{(n)}=(r^{(n)}_1,\ldots,r^{(n)}_m)$, $n\in\bN$, be probability distributions with 
strictly positive rational weights converging to 
$r$. By \eqref{eq:classical proof1-2} and the just established upper semi-continuity, 
\begin{align}\label{eq:classical proof2}
0&\le
\limsup_{n\to+\infty}
\sup_{\nu\in\simplex(\Y)}\int_{\Y}\sum_{j=1}^m r_j^{(n)}\log\tilde g(x_j,y)\,d\nu(y)
\le
\sup_{\nu\in\simplex(\Y)}\int_{\Y}\sum_{j=1}^m r_j\log\tilde g(x_j,y)\,d\nu(y).
\end{align}
Since this is true for any $m\in\bN$, any choice of $x_1,\ldots,x_m\in\supp f$ and $(r_1,\ldots,r_m)\in\S([m])$, we get that 
\begin{align}
0&\le 
\inf_{\mu\in\simplex_f(\supp f)}\sup_{\nu\in\simplex(\Y)}\int_{\Y}\sum_{x\in\X}\mu(x)\log\tilde g(x,y)\,d\nu(y)\label{eq:classical proof3}\\
&=
\max_{\nu\in\simplex(\Y)}\inf_{\mu\in\simplex_f(\supp f)}\int_{\Y}\sum_{x\in\X}\mu(x)\log\tilde g(x,y)\,d\nu(y)
\label{eq:classical proof4}\\
&=
\max_{\nu\in\simplex(\Y)}\inf_{x\in\supp f}\int_{\Y}\log\tilde g(x,y)\,d\nu(y).
\label{eq:classical proof5}
\end{align}
Here, the equality in \eqref{eq:classical proof5} is trivial, and the equality in \eqref{eq:classical proof4} follows by the application of Lemma  \ref{lemma:KF+ minimax}
to $X=\S(\Y)$, $Y=\S_f(\supp f)$, and $h(\nu,\mu):=-\int_{\Y}\sum_{x\in\X}\mu(x)\log\tilde g(x,y)\,d\nu(y)$.
Indeed, both $\S_f(\supp f)$ and $\S(\Y)$ are convex, and $\S(\Y)$ is compact in the weak$^*$-topology induced by the continuous functions on $\Y$, according to the Riesz-Markov-Kakutani representation theorem and the Banach-Alaoglu theorem. Additionally, $h$ is affine in both of its variables, and by Lemma \ref{lemma:usc}, $\S(\Y)\ni\nu\mapsto
\int_{\Y}\sum_{x\in\X}\mu(x)\log\tilde g(x,y)\,d\nu(y)$ is upper semi-continuous for every 
$\mu\in\S_f(\X)$. Hence, Lemma  \ref{lemma:KF+ minimax} applies.

Thus, by \eqref{eq:classical proof5}, there exists a $\nu_{\ep}\in\simplex(\Y)$ such that 
\begin{align}
0&\le 
\int_{\Y}\log\tilde g(x,y)\,d\nu_{\ep}(y),\ds\ds\ds x\in\supp f,
\end{align}
or equivalently,
\begin{align}\label{eq:classical proof6}
f(x)&\le 
\exp\bz\int_{\Y}\log (g(x,y)+\ep)\,d\nu_{\ep}(y)\jz,\ds\ds\ds x\in\supp f.
\end{align}
Thus, for any $\ep>0$, the set $\S_{\ep}(\Y)$ of probability measures $\nu_{\ep}$ for which 
\eqref{eq:classical proof6} holds is non-empty, and it is obviously closed in the 
weak$^*$-topology induced by the continuous functions on $\Y$.
Moreover, for $0<\ep<\ep'$ we have $\S_{\ep}(\Y)\subseteq\S_{\ep'}(\Y)$. Hence, by the 
weak$^*$-compactness of $\S(\Y)$, $\cap_{\ep>0}\S_{\ep}(\Y)\ne \emptyset$. Let 
$\nu\in\cap_{\ep>0}\S_{\ep}(\Y)$; then 
\begin{align*}
f(x)&\le 
\exp\bz\int_{\Y}\log (g(x,y)+\ep)\,d\nu(y)\jz,\ds\ds\ds x\in\supp f,
\end{align*}
holds for every $\ep>0$, and taking $\ep\searrow 0$ yields, by the monotone convergence theorem, that 
\begin{align*}
f(x)&\le 
\exp\bz\int_{\Y}\log g(x,y)\,d\nu(y)\jz,\ds\ds\ds x\in\supp f,
\end{align*}
proving \ref{classical1}.
\end{proof}

\begin{rem}
When $\Y$ is finite, \eqref{eq:classical1} can be written in the simpler form 
\begin{align}\label{eq:classical1-3}
f(x)\le\prod_{y\in\Y}g_y(x)^{\nu(y)}, \ds\ds\ds x\in\X,
\end{align}
where $0^0:=1$ according to \eqref{eq:0on0}, which is the only definition consistent 
with $\prod_{y\in\Y}g_y(x)^{\nu(y)}=\exp(\sum_{y\in\Y}\nu(y)\log g_y(x))$.
Note that with the alternative convention $0^0:=0$, the implication 
$f^{\otimes n}(\vecc{x})\le\sum_yg_y^{\otimes n}(\vecc{x})\mu(y),\,n\in\bN\imp
f(x)\le\prod_{y\in\Y}g_y(x)^{\nu(y)}$ for some $\nu\in\S(\Y)$ would not be true, as one can easily confirm
by considering the simple example $|\X|=1$, $\Y=\{1,2\}$, $f=g_1>0$, $g_2:=0$,
where $f^n=g_1^n\le 1\cdot g_1^n+0\cdot g_2^n$, $n\in\bN$, 
i.e., an arithmetic mean upper bound holds with $\mu$ being the Dirac measure concentrated at $1$,
but 
$f\nleq 0=g_1^tg_2^{1-t}$, $t\in[0,1]$.
\end{rem}

\begin{rem}
A few further remarks are in order related to Theorem \ref{thm:classical}.
First, it is clear that if a measure $\nu$ satisfies \eqref{eq:classical1} then $\mu=\nu$
satisfies \eqref{eq:classical2}. One may wonder whether the opposite implication is true as well; however, it is easily seen to fail even in the simplest setting with $|\X|=1$, $|\Y|=2$. Indeed, 
choosing 
\begin{align*}
\X=\{0\},\ds \Y=\{1,2\},\ds a:=g_1(0):=1/100,\ds b:=g_2(0):=10,\ds f(0):=1,
\end{align*}
with $\mu(1)=\mu(2)=1/2$ we have
$f(0)=1\le 10^n/2\le \mu(1)g_1(0)^n+\mu(2)g_2(0)^n$ for any $n\in\bN$, while 
$g_1(0)^{\mu(1)}g_2(0)^{\mu(2)}=a^{1/2}b^{1/2}=1/\sqrt{10}<1=f(0)$. 

Second, still in the simple setting $|\X|=1$, $|\Y|=2$, it is clear that if \eqref{eq:classical2} holds for $n=1$ then 
\eqref{eq:classical1} holds, and one may choose $\nu$ to be one of the two Dirac measures $\nu=\delta_1$ or 
$\nu=\delta_2$; indeed, assuming without loss of generality that 
$g_1(0)\ge g_2(0)$, $g_1(0)^1g_2(0)^0=g_1(0)=\max\{g_1(0),g_2(0)\}\ge \mu(1)g_1(0)+\mu(2)g_2(0)\ge f(0)$, where the last inequality is by assumption and the rest are trivial. 
Both of these implications fail in general when $|\X|\ge 2$.

Indeed, choosing $\X=\{0,1\}$, $g_1(0)=g_2(1)=1/100$, $g_1(1)=g_2(0)=10$, and $f(0)=f(1)=1$, we get that 
$f(x)\le g_1(x)/2+g_2(x)/2$, $x=0,1$, i.e., \eqref{eq:classical2} holds for $n=1$ with $\mu(1)=\mu(2)=1/2$, but 
$f(0)\le g_1(0)^{\nu(1)}g_2(0)^{\nu(2)}$ yields 
$\nu(1)\le 1/3$, while 
$f(1)\le g_1(1)^{\nu(1)}g_2(1)^{\nu(2)}$ yields 
$\nu(1)\ge 2/3$, showing that 
\eqref{eq:classical1} does not hold with any $\nu$.

Regarding the other implication, with 
$g_1(0)=g_2(1)=1/10$, $g_1(1)=g_2(0)=10$, and $f(0)=f(1)=1$, we get that 
\eqref{eq:classical1} holds with the unique measure $\nu(1)=\nu(2)=1/2$ 
(and obviously \eqref{eq:classical2} holds also for every $n\in\bN$ with $\mu=\nu$),
i.e., even if \eqref{eq:classical1} holds, it may not hold with $\nu$ being a Dirac measure.
\end{rem}
\bigskip

Let us now consider the type of bounds sought for in 
Section \ref{sec:statedisc}
in the context of hypothesis testing. For simplicity, we work in the finite-dimensional setting for the rest of the section, i.e., we assume that $|\X|<+\infty$.
In the setting of hypothesis testing, $g_y$ are probability density functions on some finite set $\X$, and 
$n$-copy tests are represented by functions $\tau:\,\X^n\to[0,1]$. The type of bound we are looking for in this setting is 
\begin{align}\label{eq:classical4}
\sum_{\vecc{x}\in\X^n}f^{\otimes n}(\vecc{x})\tau(\vecc{x})
\le
\sup_{y\in\Y}\sum_{\vecc{x}\in\X^n}g_y^{\otimes n}(\vecc{x})\tau(\vecc{x}),\ds\ds\ds
\tau\in[0,1]^{\X^n};
\end{align}
see \eqref{eq:CR def2}.
We have the following:

\begin{prop}\label{prop:classical}
Let $\X$ be a finite set and $\Y$ be a compact Hausdorff topological space.
Let $f:\,\X\to[0,+\infty)$,
$g:\,\X\times\Y\to[0,+\infty)$ be functions such that for all 
$x\in\X$, $g(x,\valt)$ is continuous.
Then \eqref{eq:classical4} holds for some $n\in\bN$ if and only if there exists 
a probability measure $\mu_n\in\S(\Y)$ such that
\begin{align}\label{eq:classical5}
f^{\otimes n}(\vecc{x})\le\int_{\Y}g_y^{\otimes n}(\vecc{x})\,d\mu_n(y),\ds\ds\ds
\vecc{x}\in\X^n.
\end{align}
\end{prop}
\begin{proof}
If $f$ is constant $0$ then the assertion holds trivially, and hence for the rest we assume that $f$ is not constant $0$.

If \eqref{eq:classical5} holds then for any $\tau\in[0,1]^{\X^n}$, 
\begin{align*}
\sum_{\vecc{x}\in\X^n}f^{\otimes n}(\vecc{x})\tau(\vecc{x})
\le
\sum_{\vecc{x}\in\X^n}\tau(\vecc{x})
\int_{\Y}g_y^{\otimes n}(\vecc{x})\,d\mu_n(y)
=
\int_{\Y}\sum_{\vecc{x}\in\X^n}\tau(\vecc{x})g_y^{\otimes n}(\vecc{x})\,d\mu_n(y)
\le
\sup_{y\in\Y}\sum_{\vecc{x}\in\X^n}\tau(\vecc{x})g_y^{\otimes n}(\vecc{x}),
\end{align*}
proving \eqref{eq:classical4}.

Assume now \eqref{eq:classical4}. Note that this trivially implies that the inequality in 
\eqref{eq:classical4} holds for any non-negative function $\tau$ on $\X^n$. In particular, 
for any probability density function $r$ on $(\supp f)^n$, the choice 
$\tau(\vecc{x}):=\bz r(\vecc{x})/f^{\otimes n}(\vecc{x})\jz\egy_{(\supp f)^n}(\vecc{x})$ yields
\begin{align*}
1=\sum_{\vecc{x}\in\X^n}f^{\otimes n}(\vecc{x})\tau(\vecc{x})
\le 
\sup_{y\in\Y}\sum_{\vecc{x}\in(\supp f)^n}\frac{g_y^{\otimes n}(\vecc{x})}{f^{\otimes n}(\vecc{x})}r(\vecc{x})
=
\sup_{\nu\in\S(\Y)}\int_{\Y}\sum_{\vecc{x}\in(\supp f)^n}\frac{g_y^{\otimes n}(\vecc{x})}{f^{\otimes n}(\vecc{x})}r(\vecc{x})\,d\nu(y).
\end{align*}
Since this holds for every such $r$, we get 
\begin{align*}
1&\le 
\inf_{r\in\S((\supp f)^n)}\sup_{\nu\in\S(\Y)}\int_{\Y}\sum_{\vecc{x}\in(\supp f)^n}\frac{g_y^{\otimes n}(\vecc{x})}{f^{\otimes n}(\vecc{x})}r(\vecc{x})\,d\nu(y)\\
&=
\max_{\nu\in\S(\Y)}\inf_{r\in\S((\supp f)^n)}\int_{\Y}\sum_{\vecc{x}\in(\supp f)^n}\frac{g_y^{\otimes n}(\vecc{x})}{f^{\otimes n}(\vecc{x})}r(\vecc{x})\,d\nu(y)\\
&=
\max_{\nu\in\S(\Y)}\inf_{r\in\S((\supp f)^n)}
\sum_{\vecc{x}\in(\supp f)^n}r(\vecc{x})\int_{\Y}\frac{g_y^{\otimes n}(\vecc{x})}{f^{\otimes n}(\vecc{x})}\,d\nu(y),
\end{align*}
where the first equality follows from Lemma \ref{lemma:KF+ minimax} due to the fact that 
$\S((\supp f)^n)\times \S(\Y)\ni(r,\nu)\mapsto \int_{\Y}\sum_{\vecc{x}\in(\supp f)^n}\frac{g_y^{\otimes n}(\vecc{x})}{f^{\otimes n}(\vecc{x})}r(\vecc{x})\,d\nu(y)$ is affine in both variables
and continuous in $\nu$ on the compact convex set $\S(\Y)$ (using, as always, the weak$^*$-topology on $\S(\Y)$). Thus, there exists a probability measure $\nu_n\in\S(\Y)$ such that 
\begin{align*}
1&\le 
\inf_{r\in\S((\supp f)^n)}\sum_{\vecc{x}\in(\supp f)^n}r(\vecc{x})\int_{\Y}
\frac{g_y^{\otimes n}(\vecc{x})}{f^{\otimes n}(\vecc{x})}\,d\nu_n(y)
=
\min_{\vecc{x}\in(\supp f)^n}\int_{\Y}\frac{g_y^{\otimes n}(\vecc{x})}{f^{\otimes n}(\vecc{x})}\,d\nu_n(y), 
\end{align*}
from which \eqref{eq:classical5} follows for $\vecc{x}\in(\supp f)^n$ by multiplying both sides by 
$f^{\otimes n}(\vecc{x})$, while for $\vecc{x}\in\X^n\setminus(\supp f)^n$, 
\eqref{eq:classical5} holds trivially.
\end{proof}

Theorem \ref{thm:classical} and Proposition \ref{prop:classical} yield immediately the following:

\begin{thm}\label{thm:classical2}
Let $\X$ be a finite set and $\Y$ be a compact Hausdorff topological space, and
let $f:\,\X\to[0,+\infty)$,
$g:\,\X\times\Y\to[0,+\infty)$ be functions such that 
for all $x\in\X$, $g(x,\valt)$ is continuous.
Then \eqref{eq:classical4} holds for every $n\in\bN$ 
if and only if any of the statements
\ref{classical1}--\ref{classical3} in Theorem \ref{thm:classical} hold.
In particular, 
\begin{align*}
\sum_{\vecc{x}\in\X^n}f^{\otimes n}(\vecc{x})\tau(\vecc{x})
\le
\sup_{y\in\Y}\sum_{\vecc{x}\in\X^n}g_y^{\otimes n}(\vecc{x})\tau(\vecc{x}),\ds\ds\ds
\tau\in[0,1]^{\X^n},\,n\in\bN,
\end{align*}
if and only if there exists a probability measure $\nu$ on the Borel $\sigma$-algebra of $\Y$ such that 
\begin{align*}
f(x)\le\exp\bz\int_{\Y}\log g(x,y)\,d\nu(y)\jz, \ds\ds\ds x\in\X.
\end{align*}
\end{thm}
\medskip

With the help of the above theorem, we can solve the problems posed in 
Section \ref{sec:statedisc} in the finite-dimensional commuting case: 
we give an explicit single-copy characterization of the elements of $\C(\R)$, and 
an explicit construction of a set $\what\C(\R)$ satisfying 
$\max\C(\R)\subseteq\what\C(\R)\subseteq\C(\R)$.
We have the following precise statements:

\begin{thm}\label{thm:classical maxCR}
Let $\hil$ be a finite-dimensional Hilbert space, let $\R\subseteq\B(\hil)\p$ be a 
compact set of commuting PSD operators, and $C\in\B(\hil)\p$. Then
\begin{align}\label{eq:classical chat0}
C\in\C(\R)\ds\iff\ds\exists\,\nu\in\S(\R):\, C\le\exp\int_{\R}\log\omega\,d\nu(\omega),
\end{align}
and
\begin{align}\label{eq:classical chat}
\max\C(\R)\subseteq\left\{\exp\int_{\R}\log\omega\,d\nu(\omega):\,\nu\in\S(\R)\right\}
\subseteq\C(\R).
\end{align}
\end{thm}

The proof of \eqref{eq:classical chat0} follows immediately from Theorem \ref{thm:classical2} when 
$\gm$ commutes with the elements of $\R$, which, however, is not obvious from the definition of $\C(\R)$. 
We will give the proof of \eqref{eq:classical chat0} below in Section \ref{sec:comm}.
The proof of \eqref{eq:classical chat} is immediate from \eqref{eq:classical chat0}.

\begin{rem}
Note that the second inclusion in \eqref{eq:classical chat} is almost always strict, 
and the first inclusion may also be strict in general. As a simple example for the latter, 
consider $\R=\{\ops_1,\ops_2\}$ with $0<\ops_1<\ops_2$, so that 
\begin{align*}
\left\{\exp\int_{\R}\log\omega\,d\nu(\omega):\,\nu\in\S(\R)\right\}
=
\left\{\ops_1^t\ops_2^{1-t}:\,t\in[0,1]\right\},
\end{align*}
and $\ops_1^t\ops_2^{1-t}<\ops_2$ for every $t\in(0,1]$, while $\ops_1^0\ops_2^{1-0}=\ops_2$.
Thus, $\max\C(\R)=\{\ops_2\}\subsetneq\{\ops_1^t\ops_2^{1-t}:\,t\in[0,1]\}$.
\end{rem}

\section{Equivalence of weak bounds}
\label{sec:weak bounds}

In this section, we discuss the relation between various weak bounds of the types considered, e.g., in 
\eqref{eq:CR def2} and \eqref{eq:mean bound1}.

For the rest of this section, let $\Y$ be a compact Hausdorff topological space, 
$\hil$ be a complex Hilbert space,
$\ops:\,\Y\to\B(\hil)\p$ be a function, and $\gm\in\B(\hil)\p$. For every $n\in\bN$, we fix a subset 
$\X_n\subseteq\B(\hil^{\otimes n})\p$ 
satisfying the following:

\begin{assumption}\label{trace assumption}
\begin{enumerate}
\item
For every 
$y\in\Y$, every $n\in\bN$, and every $X\in\X_n$, $\ops_y^{\otimes n}X$ is trace-class, and so is $\gm^{\otimes n} X$
as well (this is always true when $\hil$ is finite-dimensional).

\item
For any $n,m\in\bN$,
\begin{align*}
\X_n\otimes\X_m:=\{X_1\otimes X_2:\,X_1\in\X_n,\,X_2\in\X_m\}\subseteq\X_{n+m}.
\end{align*}

\item
For every $n\in\bN$ and every $X\in\X_n$, 
$y\mapsto \Tr \ops_y^{\otimes n}X$ is continuous and hence, in particular, it is bounded.
\end{enumerate}
\end{assumption}

\begin{rem}
Note that the boundedness assumption implies that for any $\nu\in\S(\Y)$, 
\begin{align}\label{eq:nc finiteness condition}
\ds\exists\,\int_{\Y}\log \Tr \ops_y^{\otimes n}X\,d\nu(y)\in[-\infty,+\infty),
\ds\ds\ds X\in\X_n,\ds n\in\bN.
\end{align}
\end{rem}

The following are some examples of natural choices of the sets $\X_n$, using the same choice for all $n$:
\begin{itemize}
\item
The set of all rank $1$ projections on $\hil^{\otimes n}$.
\item
The set of all rank $1$ PSD operators on $\hil^{\otimes n}$.
\item
The set of all  non-zero finite rank PSD operators on $\hil^{\otimes n}$, optionally with the additional constraint that their operator norm is at most $1$.
\item
The set of all  non-zero compact PSD operators on $\hil^{\otimes n}$, optionally with the additional constraint that their operator norm is at most $1$.
\item
The set of all density operators on $\hil^{\otimes n}$.
\item
The set of all density operators $\rho$ on $\hil^{\otimes n}$ such that 
$\ker(\rho)=\{0\}$, when $\hil$ is separable.
\item
The set of all  non-zero tests, i.e., $\X_n=\B(\hil^{\otimes n})_{[0,1]}\setminus\{0\}$; this can be a valid choice if
$A_y$ is trace-class for all $y\in\Y$.
\item
The set of all strictly positive tests, i.e., $\X_n=\B(\hil^{\otimes n})_{(0,1]}$; this can be a valid choice if
$A_y$ is trace-class for all $y\in\Y$.
\end{itemize}

Note that any of the above choices for $(\X_n)_{n\in\bN}$
characterizes the PSD order in the sense that for any $A,B\in\B(\hil^{\otimes n})\p$,
\begin{align}\label{eq:ordering property}
A\le B\ds\iff\ds \Tr AX\le\Tr BX,\ds\ds\ds X\in\X_n,
\end{align}
assuming that $A,B$ are such that $AX$ and $BX$ are trace-class for all $X\in\X_n$.

\begin{rem}
The main application that we have in mind is where $\Y$ is a subset of $\B(\hil)\p$, 
and $\Y\ni y\mapsto A_y:=y$ is just the identity. This map is obviously continuous in the operator norm,
and $\Y\ni y\mapsto \Tr A_y^{\otimes n}X$ is continuous for every $X\in\X_n$, provided that 
$\X_n$ is a subset of the set of PSD trace-class operators and $\Y$ is equipped with the operator norm topology. 
Alternatively, for state discrimination, $\Y$ would be a subset of $\S(\hil)$ equipped with the 
trace-norm topology, and $\Y\ni y\mapsto \Tr A_y^{\otimes n}X$ is continuous for every $X\in\X_n$
in this topology for any of the above listed choices of $\X_n$.
\end{rem}

\begin{rem}
Note that if $A,X\in\B(\hil)\p$ and $AX$ is trace-class then so is $X^{1/2}AX^{1/2}$, and 
\begin{align}\label{eq:trace equality}
\Tr X^{1/2}AX^{1/2}=\Tr A^{1/2}XA^{1/2}=\Tr AX.
\end{align}
Indeed, the first statement follows from the Araki-Lieb-Thirring inequality \cite{Araki} as
$\Tr X^{1/2}AX^{1/2}=\Tr(X^{1/2}AX^{1/2})^{2(1/2)}\le\Tr (XA^2X)^{1/2}=\Tr |AX|$, and 
\eqref{eq:trace equality} follows by Lidskii's theorem (see \cite[Corollary 3.8]{Simon2005TraceIdeals}) from the facts that both $AX$ and $X^{1/2}AX^{1/2}$ are compact, and they have the same non-zero eigenvalues counted with multiplicities.
The converse implication is not true, i.e., if $\Tr X^{1/2}AX^{1/2}<+\infty$ then $AX$ need not be trace-class. 
As a simple example, consider $\hil=\oplus_{n\in\bN}\hil_n$, $\hil_n\cong\bC^2$, 
$\psi_n:=(1,0)$, $\phi_n:=(\cos\theta_n,\sin\theta_n)$, $n\in\bN$, where
$\theta_n\in[0,\pi/2)$ is such that $\cos^2(\theta_n)=n^{-3/2}$, and let 
$A:=\oplus_n\pr{\psi_n}$, $X:=\oplus_n\pr{\phi_n}$. Then 
\begin{align*}
\Tr X^{1/2}AX^{1/2}=\sum_{n\in\bN}n^{-3/2}<+\infty=\sum_{n\in\bN}n^{-3/4}=\Tr|AX|.
\end{align*}
We note that Proposition \ref{prop:weak bounds} below holds with the same proof if 
$\Tr \ops_y^{\otimes n}X$ is replaced everywhere with 
$\Tr X^{1/2}\ops_y^{\otimes n}X^{1/2}$ 
and 
$\Tr \gm^{\otimes n}X$ with 
$\Tr X^{1/2}\gm^{\otimes n}X^{1/2}$, including Assumption \ref{trace assumption}.
This technical subtlety is not relevant for the intended applications, where we always have that 
$\X_n$ consists of trace-class operators whenever 
$\{\ops_y\}_{y\in\Y}$ does not, and therefore we do not discuss it further.
\end{rem}

\medskip

\begin{prop}\label{prop:weak bounds}
In the above setting, consider the following statements:
\begin{enumerate}

\item\label{nc2}
\emph{(Multi-copy upper bounds by a weighted geometric mean.)}
There exists a probability measure $\nu\in\S(\Y)$ such that 
\begin{align}\label{eq:nc2}
\Tr \gm^{\otimes n} X\le\exp\bz\int_{\Y}\log \Tr \ops_y^{\otimes n}X\,d\nu(y)\jz, \ds\ds\ds X\in\X_n,\ds n\in\bN.
\end{align}

\item\label{nc1}
\emph{(Multi-copy upper bounds by a sequence of weighted geometric means.)}
For every $n\in\bN$, there exists 
a probability measure $\nu_n\in\S(\Y)$ such that 
\begin{align}\label{eq:nc1}
\Tr \gm^{\otimes n} X\le\exp\bz\int_{\Y}\log \Tr \ops_y^{\otimes n}X\,d\nu_n(y)\jz, \ds\ds\ds X\in\X_n.
\end{align}

\item\label{nc3}
\emph{(Multi-copy upper bounds by a weighted arithmetic mean.)}
There exists a probability measure $\mu\in\S(\Y)$ such that for every $n\in\bN$,
\begin{align}\label{eq:nc3}
\Tr \gm^{\otimes n} X\le\int_{\Y}\Tr\ops_y^{\otimes n}X\,d\mu(y),\ds\ds\ds
X\in\X_n.
\end{align}

\item\label{nc4}
\emph{(Multi-copy upper bounds by a sequence of weighted arithmetic means.)}
There exists a 
sequence of positive measures $(\mu_n)_{n\in\bN}$ on the Borel $\sigma$-algebra $\F$ of $\Y$ such that 
\begin{align}\label{eq:nc subexponential weight}
\limsup_{n\to+\infty}\frac{1}{n}\log\mu_n(\Y)\le 0,
\end{align}
and 
\begin{align}\label{eq:nc4}
\Tr \gm^{\otimes n} X\le\int_{\Y}\Tr\ops_y^{\otimes n}X\,d\mu_n(y),\ds\ds\ds
X\in\X_n,\ds n\in\bN.
\end{align}

\item\label{nc5}
\emph{(Upper bound by the supremum.)}
For every $n\in\bN$,
\begin{align*}
\Tr \gm^{\otimes n}X\le\sup_{y\in\Y}\Tr\ops_y^{\otimes n}X,\ds\ds\ds X\in\X_n.
\end{align*}
\end{enumerate}

Then we have \ref{nc2}$\iff$\ref{nc1}$\iff$\ref{nc3}$\iff$\ref{nc4}$\imp$\ref{nc5}.
\end{prop}
\begin{proof}
The implications \ref{nc2}$\imp$\ref{nc3}$\imp$\ref{nc4} and \ref{nc2}$\imp$\ref{nc1}$\imp$\ref{nc4} are obvious. To see
\ref{nc4}$\imp$\ref{nc5}, note that \ref{nc4} implies that for all 
$n,k\in\bN$ and $X\in\X_n$,
\begin{align*}
\Tr \gm^{\otimes nk} X^{\otimes k}\le\int_{\Y}\Tr\ops_y^{\otimes nk}X^{\otimes k}\,d\mu_{nk}(y)
\le\bz\sup_{y\in\Y}\Tr\ops_y^{\otimes nk}X^{\otimes k}\jz\mu_{nk}(\Y).
\end{align*}
Taking then the $k$-th root and then the limit $k\to+\infty$ yields
\begin{align*}
\Tr \gm^{\otimes n} X
\le
\bz\limsup_{k\to+\infty}\mu_{nk}(\Y)^{1/nk}\jz^n
\sup_{y\in\Y}\Tr\ops_y^{\otimes n}X
\le
\sup_{y\in\Y}\Tr\ops_y^{\otimes n}X,\ds\ds\ds X\in\X_n,
\end{align*}
where we used that
\begin{align}\label{eq:nc proof1}
\limsup_{k\to+\infty}\frac{1}{k}\log\mu_{nk}(\Y)\le 0,
\end{align}
according to \eqref{eq:nc subexponential weight}. 
This proves \ref{nc5}.
Hence, we are left to prove
\ref{nc4}$\imp$\ref{nc1}$\imp$\ref{nc2}.

Assume first that \ref{nc1} holds, and 
let $\S_n$ be the set of probability measures $\nu$ for which 
$\Tr \gm^{\otimes n} X\le\exp\bz\int_{\Y}\log \Tr \ops_y^{\otimes n}X\,d\nu(y)\jz$ for every 
$X\in\X_n$. Let $(\nu_i)_{i\in\I}$ be a net in $\S_n$ converging to some $\nu\in\S(\Y)$
in the weak$^*$-topology on $\S(\Y)$.
By Lemma \ref{lemma:usc}, for every $X\in\X_n$,
$\S(\Y)\ni\nu\mapsto\int_{\Y}\log \Tr \ops_y^{\otimes n}X\,d\nu(y)$
is upper semi-continuous, and therefore,
\begin{align*}
\Tr \gm^{\otimes n} X
\le
\limsup_i\exp\bz\int_{\Y}\log \Tr \ops_y^{\otimes n}X\,d\nu_i(y)\jz
\le
\exp\bz\int_{\Y}\log \Tr \ops_y^{\otimes n}X\,d\nu(y)\jz,
\end{align*}
i.e., $\nu\in\S_n$. Thus, $\S_n$ is closed.
Now observe that if $\nu\in\S_{nm}$ then for any 
$X\in\X_n$,
\begin{align*}
\bz\Tr \gm^{\otimes n} X\jz^m=
\Tr \gm^{\otimes nm} X^{\otimes m}
\le
\exp\bz\int_{\Y}\log \Tr \ops_y^{\otimes nm}X^{\otimes m}\,d\nu(y)\jz
=
\bz\exp\bz\int_{\Y}\log \Tr \ops_y^{\otimes n}X\,d\nu(y)\jz\jz^m,
\end{align*}
whence $\nu\in\S_n$. Thus, for any $n_1,\ldots,n_k\in\bN$, 
$\S_{n_1}\cap\ldots\cap S_{n_k}\supseteq \S_{n_1\cdot\ldots\cdot n_k}\ne\emptyset$. Hence, by the 
compactness of $\S(\Y)$, $\cap_{n\in\bN}\S_n\ne\emptyset$, and
\eqref{eq:nc2} holds with any $\nu\in\cap_{n\in\bN}\S_n$, proving \ref{nc2}.

Next, assume \ref{nc4}.
For a fixed $n\in\bN$, let 
\begin{align*}
f(X):=\Tr \gm^{\otimes n}X,\ds\ds\ds 
g(X,y):=\Tr \ops_y^{\otimes n}X,\ds\ds\ds X\in\X_n.
\end{align*}
Then $g$ satisfies the conditions of Theorem \ref{thm:classical}, and 
for every $k\in\bN$ and $X_1,\ldots,X_k\in\X_n$,
\begin{align*}
f^{\otimes k}\bz\vecc{X}\jz
=
\prod_{i=1}^k\Tr \gm^{\otimes n}X_i
&=
\Tr \gm^{\otimes nk}(X_1\ootimes X_k)\\
&\le
\int_{\Y}\Tr \ops_y^{\otimes nk}(X_1\ootimes X_k)\,d\mu_{nk}
=
\int_{\Y}g_y^{\otimes k}\bz\vecc{X}\jz\,d\mu_{nk},
\end{align*} 
where the inequality is due to \ref{nc4}. The assumption in \ref{nc4} 
implies \eqref{eq:nc proof1}, and therefore \ref{classical3} of Theorem \ref{thm:classical}
holds with the above $f,g$ and $\tilde\mu_k:=\mu_{nk}$, $k\in\bN$. Thus, by 
\ref{classical3}$\imp$\ref{classical1} of Theorem \ref{thm:classical},
there exists a $\nu_n\in\S(\Y)$ such that 
\begin{align*}
\Tr\gm^{\otimes n} X=f(X)\le
\exp\bz\int_{\Y}\log g_y(X)\,d\nu_n(y)\jz
=
\exp\bz\int_{\Y}\log\Tr\ops_y^{\otimes n} X\,d\nu_n(y)\jz,\ds\ds\ds
X\in\X_n,
\end{align*}
proving \ref{nc1}.
\end{proof}

Note that if $\Y\ni y\mapsto A_y$ is a measurable function that is bounded in the operator norm or the trace norm, depending on the setting, then for any finite measure $\mu$ on the Borel $\sigma$-algebra of $\Y$, $\int_{\Y}A_y^{\otimes n}\,d\mu(y)$ is a well-defined operator in the respective space, and for any of the above listed choices of $\X_n$, \eqref{eq:ordering property} implies that 
the weak inequality \eqref{eq:nc3}
can be equivalently written in the strong form
\begin{align}\label{eq:strong arithmetic bound}
\gm^{\otimes n}\le\int_{\Y}\ops_y^{\otimes n}\,d\mu(y).
\end{align}
The inequalities in \eqref{eq:nc2} and \eqref{eq:nc1} do not have obvious strong versions
(see, however, Section \ref{sec:2-var}), and neither do the type of bounds in 
Proposition \ref{prop:weak bounds} \ref{nc5}
since a set of PSD operators does not have a supremum in the PSD order except for very special cases
(see, e.g., \cite{Kadison1951Order}).
On the other hand, a weak inequality as in Proposition \ref{prop:weak bounds} \ref{nc5}
turns out to be equivalent to the strong inequality in \eqref{eq:strong arithmetic bound} when the 
integral in \eqref{eq:strong arithmetic bound} exists, as we show in the following.

\begin{prop}\label{prop:sup bound implies AM bound}
Assume that $\Y\ni y\mapsto \ops_y$ is continuous in the operator norm, 
and let $C\in\B(\hil)\p$.
The following are equivalent:
\begin{enumerate}
\item\label{sup bound implies AM bound1}
For every $X\in\B(\hil)\p$,
\begin{align}\label{eq:weak sup bound-1}
\Tr X^{1/2}\gm X^{1/2}\le\sup_{y\in\Y}\Tr X^{1/2}\ops_y X^{1/2}.
\end{align}
\item\label{sup bound implies AM bound2}
There exists a convex set $\tilde S\subseteq\S(\hil)$ and a $\delta>0$
such that 
\begin{align}\label{eq:specinf}
\inf_{\rho\in\tilde\S}\Tr\rho Z=\inf\spec(Z),\ds\ds\ds Z\in\B(\hil)\p,
\end{align}
and for any $X\in\tilde\X:=\{(\gm+\ep I)^{-1/2}\rho (\gm+\ep I)^{-1/2}:\,\rho\in\tilde\S,\,\ep\in(0,\delta)\}$,
\begin{align}\label{eq:weak sup bound}
\Tr \gm X\le\sup_{y\in\Y}\Tr \ops_y X.
\end{align}
\item\label{sup bound implies AM bound3}
There exists a probability measure $\mu\in\S(\Y)$ such that 
\begin{align}\label{eq:strong arithmetic bound2}
\gm\le\int_{\Y}\ops_y\,d\mu(y).
\end{align}
\end{enumerate}
\end{prop}
\begin{proof}
The implication \ref{sup bound implies AM bound1}$\imp$\ref{sup bound implies AM bound2} 
follows by \eqref{eq:trace equality}
and choosing an arbitrary
$\delta>0$ and $\tilde \S=\S(\hil)$.
The implication \ref{sup bound implies AM bound3}$\imp$\ref{sup bound implies AM bound1} is trivial.
Thus, we prove \ref{sup bound implies AM bound2}$\imp$\ref{sup bound implies AM bound3}.
Note that \eqref{eq:weak sup bound} implies that
\begin{align}\label{eq:weak sup bound1}
\Tr (\gm+\ep I) X\le\sup_{y\in\Y}\Tr (\ops_y+\ep I) X,\ds\ds\ds X\in\tilde\X.
\end{align}
For an arbitrary $\rho\in\tilde\S$ and $\ep\in(0,\delta)$, let 
$X:=(\gm+\ep I)^{-1/2}\rho (\gm+\ep I)^{-1/2}\in\tilde \X$. 
Then
\begin{align*}
1= \Tr\rho
&=\Tr (\gm+\ep I)(\gm+\ep I)^{-1/2}\rho (\gm+\ep I)^{-1/2}\\
&\le
\sup_{y\in\Y}\Tr (\ops_y+\ep I)(\gm+\ep I)^{-1/2}\rho (\gm+\ep I)^{-1/2}\\
&=
\sup_{\mu\in\S(\Y)}\int_{\Y}\Tr (\ops_y+\ep I)(\gm+\ep I)^{-1/2}\rho (\gm+\ep I)^{-1/2}\,d\mu(y),
\end{align*}
where the inequality follows from \eqref{eq:weak sup bound1}.
Since this holds for every $\rho\in\tilde\S$, we get 
\begin{align*}
1
&\le
\inf_{\rho\in\tilde\S}\sup_{\mu\in\S(\Y)}\int_{\Y}\Tr 
(\ops_y+\ep I)(\gm+\ep I)^{-1/2}\rho (\gm+\ep I)^{-1/2}\,d\mu(y)\\
&=
\max_{\mu\in\S(\Y)}\inf_{\rho\in\tilde\S}
\int_{\Y}\Tr (\ops_y+\ep I)(\gm+\ep I)^{-1/2}\rho (\gm+\ep I)^{-1/2}\,d\mu(y),
\end{align*}
where the equality follows from Lemma \ref{lemma:KF+ minimax} due to the facts that 
$\tilde\S$ and $\S(\Y)$ are convex,
$\tilde\S\times\S(\Y)\ni(\rho,\mu)\mapsto\int_{\Y}\Tr (\ops_y+\ep I)(\gm+\ep I)^{-1/2}\rho (\gm+\ep I)^{-1/2}\,d\mu(y)$
is affine in both variables and continuous in $\mu$
with respect to the weak$^*$-topology in $\S(\Y)$, with respect to which $\S(\Y)$ is compact.
Thus, there exists a $\mu\in\S(\Y)$ such that 
\begin{align}\label{eq:sup bound to AM proof1}
1&\le
\inf_{\rho\in\tilde\S}\int_{\Y}\Tr (\ops_y+\ep I)(\gm+\ep I)^{-1/2}\rho (\gm+\ep I)^{-1/2}\,d\mu(y)\\
&=
\inf_{\rho\in\tilde\S}\int_{\Y}\Tr \rho(\gm+\ep I)^{-1/2}
(\ops_y+\ep I)(\gm+\ep I)^{-1/2}\,d\mu(y)\\
&=
\inf_{\rho\in\tilde\S}\Tr\rho\int_{\Y}(\gm+\ep I)^{-1/2}
(\ops_y+\ep I)(\gm+\ep I)^{-1/2}\,d\mu(y),
\end{align}
where the second equality follows from the fact that $\Y\ni y\mapsto (\gm+\ep I)^{-1/2}
(\ops_y+\ep I)(\gm+\ep I)^{-1/2}$ is continuous in the operator norm.
Thus, by the assumption \eqref{eq:specinf},
\begin{align*}
1\le\inf\spec\bz\int_{\Y}(\gm+\ep I)^{-1/2}
(\ops_y+\ep I)(\gm+\ep I)^{-1/2}\,d\mu(y)\jz,
\end{align*}
or equivalently,
\begin{align*}
I\le
\int_{\Y}(\gm+\ep I)^{-1/2}
(\ops_y+\ep I)(\gm+\ep I)^{-1/2}\,d\mu(y)=
(\gm+\ep I)^{-1/2}\bz\int_{\Y}(\ops_y+\ep I)\,d\mu(y)\jz(\gm+\ep I)^{-1/2},
\end{align*}
which in turn yields
\begin{align*}
\gm+\ep I\le\int_{\Y}(\ops_y+\ep I)\,d\mu(y)=\int_{\Y}\ops_y\,d\mu(y)+\ep I,
\end{align*}
proving \eqref{eq:strong arithmetic bound2}.
\end{proof}

Proposition \ref{prop:sup bound implies AM bound} allows to complete Proposition \ref{prop:weak bounds}
by proving \ref{nc5}$\imp$\ref{nc2} in the latter under some additional technical conditions.

\begin{thm}\label{thm:sup bounds implies weak geometric bound}
Assume that $\Y\ni y\mapsto \ops_y$ is continuous in operator norm, let $\gm\in\B(\hil)\p$, 
and let $(\X_n)_{n\in\bN}$ be any sequence of sets 
satisfying Assumption \ref{trace assumption},
with the additional requirement that for every $n\in\bN$, \eqref{eq:ordering property} holds. 
Assume that for every $n\in\bN$ 
there exists a convex set $\tilde S_n\subseteq\S(\hil^{\otimes n})$ and a $\delta_n>0$
such that 
\begin{align}\label{eq:specinf2}
\inf_{\rho\in\tilde\S_n}\Tr\rho Z=\inf\spec(Z),\ds\ds\ds Z\in\B(\hil^{\otimes n})\p,
\end{align}
and 
\begin{align}\label{eq:tilde X_n}
\tilde\X_n:=\{(\gm^{\otimes n}+\ep I)^{-1/2}\rho (\gm^{\otimes n}+\ep I)^{-1/2}:\,\rho\in\tilde\S_n,\,\ep\in(0,\delta_n)\}
\subseteq\X_n.
\end{align}
Then, the following are equivalent:
\begin{enumerate}
\item\label{weak equivalence1}
For every $n\in\bN$,
\begin{align}\label{eq:weak sup bound2}
\Tr X^{1/2}\gm^{\otimes n}X^{1/2}\le\sup_{y\in\Y}\Tr X^{1/2}\ops_y^{\otimes n}X^{1/2},\ds\ds\ds X\in\B(\hil^{\otimes n})\p.
\end{align}

\item\label{weak equivalence2-1}
For every $n\in\bN$,
\begin{align*}
\Tr\gm^{\otimes n}X\le\sup_{y\in\Y}\Tr \ops_y^{\otimes n}X,\ds\ds\ds X\in\X_n.
\end{align*}

\item\label{weak equivalence2}
For every $n\in\bN$,
\begin{align}\label{eq:weak sup bound3}
\Tr \gm^{\otimes n}X\le\sup_{y\in\Y}\Tr \ops_y^{\otimes n}X,\ds\ds\ds X\in\tilde \X_n.
\end{align}

\item\label{weak equivalence3}
For every $n\in\bN$, there exists a probability measure $\mu_n\in\S(\Y)$ such that 
\begin{align*}
\gm^{\otimes n}\le\int_{\Y}\ops_y^{\otimes n}\,d\mu_n(y).
\end{align*}

\item\label{weak equivalence4}
There exists a probability measure $\mu\in\S(\Y)$ such that 
\begin{align*}
\gm^{\otimes n}\le\int_{\Y}\ops_y^{\otimes n}\,d\mu(y),\ds\ds\ds n\in\bN.
\end{align*}

\item\label{weak equivalence5}
There exists a probability measure $\nu\in\S(\Y)$ such that 
\begin{align}\label{eq:weak equivalence5}
\Tr \gm^{\otimes n} X\le\exp\bz\int_{\Y}\log \Tr \ops_y^{\otimes n}X\,d\nu(y)\jz, \ds\ds\ds X\in\X_n,\ds n\in\bN.
\end{align}
\end{enumerate}
\end{thm}
\begin{proof}
The implication \ref{weak equivalence1}$\imp$\ref{weak equivalence2-1}
follows by \eqref{eq:trace equality}, and the implication
\ref{weak equivalence2-1}$\imp$\ref{weak equivalence2} 
is trivial.
The implication \ref{weak equivalence2}$\imp$\ref{weak equivalence3} 
follows from Proposition \ref{prop:sup bound implies AM bound}, and the 
implication \ref{weak equivalence3}$\imp$\ref{weak equivalence1}
is trivial.
The equivalences \ref{weak equivalence3}$\iff$\ref{weak equivalence4}$\iff$\ref{weak equivalence5} 
follow from Proposition \ref{prop:weak bounds}
due to the assumption that $\X_n$ characterizes the PSD order on $\B(\hil^{\otimes n})$ as in
\eqref{eq:ordering property}. 
\end{proof}

Finally, in the setting considered in Section \ref{sec:statedisc}, we get the following:
\begin{cor}\label{cor:sup bounds implies weak geometric bound}
Let $\R\subseteq\S(\hil)$ be compact in the trace norm and $\gm\in\B(\hil)\p$
be trace-class. Then the following are equivalent:
\begin{enumerate}
\item
$\gm\in\C(\R)$, i.e., 
\begin{align}\label{eq:sdics sup bound}
\Tr\gm^{\otimes n}T\le\sup_{\omega\in\R}\Tr\omega^{\otimes n}T,\ds\ds\ds 
T\in\B(\hil^{\otimes n})_{[0,1]},\,n\in\bN.
\end{align}

\item
There exists a probability measure $\mu$ on the Borel $\sigma$-algebra of $\R$ such that 
\begin{align*}
\gm^{\otimes n}\le\int_{\R}\omega^{\otimes n}\,d\mu(\omega),\ds\ds\ds n\in\bN.
\end{align*}

\item
There exists a probability measure $\nu$ on the Borel $\sigma$-algebra of $\R$ such that 
\begin{align}\label{eq:sdics geometric bound}
\Tr \gm^{\otimes n} X\le\exp\bz\int_{\R}\log \Tr \omega^{\otimes n}X\,d\nu(\omega)\jz, \ds\ds\ds 
X\in\B(\hil^{\otimes n})\p,\ds n\in\bN.
\end{align}
\end{enumerate}
\end{cor}
\begin{proof}
With $\Y:=\R$ and $\ops_{\omega}:=\omega$, $\omega\in\Y$, we have that  
$\Y$ is compact, and $\Y\ni\omega\mapsto\ops_{\omega}$ is continuous in the trace norm,
and therefore also in the operator norm. 
Clearly, \eqref{eq:sdics sup bound} 
is equivalent to $\Tr\gm^{\otimes n}X\le\sup_{\omega\in\R}\Tr\omega^{\otimes n}X$,
$X\in\B(\hil^{\otimes n})\p=:\X_n$, $n\in\bN$,
and with $\tilde\S_n:=\S(\hil^{\otimes n})$
and $\tilde\X_n$ being defined as in \eqref{eq:tilde X_n} with $\delta_n:=+\infty$, $n\in\bN$,
the conditions of Theorem \ref{thm:sup bounds implies weak geometric bound} are satisfied. 
Hence, the asserted equivalences follow from 
Theorem \ref{thm:sup bounds implies weak geometric bound}.
\end{proof}

Corollary \ref{cor:sup bounds implies weak geometric bound} gives various equivalent characterizations of 
$\C(\R)$. These might be easier to work with than \eqref{eq:CR def2}, but none of them 
are single-copy. We will give single-copy characterizations of 
$\C(\R)$ and 
$\max\C(\R)$ in the next section when $|\Y|=2$
and the Hilbert space is finite dimensional, based on the characterization in \eqref{eq:sdics geometric bound}.

\section{Maximality of the weighted Kubo-Ando  geometric means for two variables}
\label{sec:2-var}

In this section, we give complete characterizations of 
$\C(\{\ops_y\}_{y\in\Y})$ and 
$\max\C(\{\ops_y\}_{y\in\Y})$ when $\Y$ is a $2$-point set and the Hilbert space $\hil$ is finite-dimensional.
(We may assume without loss of generality that $\Y=\{1,2\}$).

For this, we will need some additional preparation regarding permutation-invariant operators on tensor product Hilbert spaces. Recall that by the 
Schur-Weyl duality (see, e.g., \cite[Chapter 9]{GoodW}), 
for every $n\in\bN$, $\hil^{\otimes n}$ decomposes as
\begin{align*}
\hil^{\otimes n}=\bigoplus_{\lambda\in Y_{n,d}}\hil_{n,\lambda}\otimes\kil_{n,\lambda},
\end{align*}
where $Y_{n,d}$ is the set of
the Young diagrams up to depth $d:=\dim\hil$,
defined as
\begin{align*}
Y_{n,d}:=\left\{\lambda=(\lambda_1,\lambda_2,\dots,\lambda_d)\,:\,\lambda_i\in\bN\cup\{0\},\,i\in[d],\, 
\lambda_1\ge \lambda_2\ge \dots\ge \lambda_d\ge 0,\;\sum_{i=1}^d \lambda_i=n \right\},
\end{align*}
the dimension of each $\hil_{n,\lambda}$ is polynomially bounded in $n$ as 
\begin{align}\label{eq:SW dimension bound}
\max_{\lambda}\{\dim\hil_{n,\lambda}\} \le (n+1)^{d(d-1)/2},
\end{align}
(see, e.g., \cite[Section 6.2.1]{Hayashi2017GroupTheoretic}),
and every permutation-invariant operator $X$ on $\hil^{\otimes n}$ is of the form
\begin{align}\label{eq:symmetric op decomp}
X=\bigoplus_{\lambda\in Y_{n,d}}X_{n,\lambda}\otimes I_{\kil_{n,\lambda}}.
\end{align}
Here, permutation-invariant means that $X$ can be written as $X=\sum_{k=1}^r X_{1,k}\ootimes X_{n,k}$
with $X_{i,k}\in\B(\hil)$, $i\in[n]$, $k\in[r]$, such that for any permutation $\sigma$ of $[n]$,
$X=\sum_{k=1}^r X_{\sigma\inv(1),k}\ootimes X_{\sigma\inv(n),k}$.
We will denote the projection from $\hil^{\otimes n}$ onto $\hil_{n,\lambda}\otimes\kil_{n,\lambda}$ 
by $P_{n,\lambda}$.
It follows from \eqref{eq:symmetric op decomp}
that for any symmetric PSD operator 
$X=\bigoplus_{\lambda\in Y_{n,d}}X_{n,\lambda}\otimes I_{\kil_{n,\lambda}}$,
\begin{align}
\norm{XP_{n,\lambda}}_{\infty}
=
\norm{X_{n,\lambda}\otimes I_{\kil_{n,\lambda}}}_{\infty}
=
\norm{X_{n,\lambda}}_{\infty}
\le
\Tr X_{n,\lambda}
=
\frac{\Tr (X_{n,\lambda}\otimes I_{\kil_{n,\lambda}})}{\dim\kil_{n,\lambda}}
=
\frac{\Tr(XP_{n,\lambda})}{\dim\kil_{n,\lambda}}.
\label{eq:SW norm inequality}
\end{align}

For a fixed $n\in\bN$, let $\type_n([d])$ denote all the probability density functions on $[d]$
that have rational weights with denominator $n$ (so-called \ki{$n$-types}). For a given orthonormal basis
$(e_i)_{i=1}^d$ and a 
$Q\in\type_n([d])$, let $\Pi^n_Q$ be the projection onto
\begin{align*}
\hil^n_Q:=\spann\left\{e_{\vecc{i}}:=\otimes_{k=1}^n e_{i_k}:\,
\type_{\vecc{i}}(j):=\frac{1}{n}|\{k:\,i_k=j\}|=Q(j),\,j\in[d]\right\}.
\end{align*}
It is clear that $\Pi^n_Q=\sum_{\vecc{i}:\,\type_{\vecc{i}}=Q}\pr{e_{\vecc{i}}}$ is permutation-invariant, and therefore
\begin{align*}
\Pi^n_QP_{n,\lambda}=P_{n,\lambda}\Pi^n_Q,\ds\ds\ds\lambda\in Y_{n,d}.
\end{align*}
Moreover, if $A\in\B(\hil)\pp$ is a positive definite operator that is diagonal in the given basis and therefore can be written as $A=\sum_{i=1}^d a(i)\pr{e_i}$, we have 
\begin{align}\label{eq:type prob}
\Pi^n_Q A^{\otimes n}=A^{\otimes n}\Pi^n_Q =\sum_{e_{\vecc{i}}\in\hil^n_Q}\diad{A^{\otimes n}e_{\vecc{i}}}{e_{\vecc{i}}}
=\prod_{i\in[d]}a(i)^{nQ(i)}\Pi^n_Q
=e^{n\sum_{i=1}^dQ(i)\log a(i)}\Pi^n_Q.
\end{align}

\begin{thm}\label{thm:2-var char}
Let $\ops_1,\ops_2\in\B(\hil)\p$ be positive semi-definite operators on a finite-dimensional Hilbert space $\hil$, and let $C\in\B(\hil)\p$. For a $t\in[0,1]$, the following holds:
\begin{align}\label{eq:2-var weak geometric bound}
\Tr C^{\otimes n}X\le(\Tr \ops_1^{\otimes n}X)^{t}(\Tr \ops_2^{\otimes n}X)^{1-t},\ds\ds\ds X\in\B(\hil^{\otimes n})\p,\ds n\in\bN,
\end{align}
if and only if
\begin{align}\label{eq:KA bound}
C\le\ops_2\#_{t}\ops_1.
\end{align}
\end{thm}
\begin{proof}
The implication \eqref{eq:KA bound}$\imp$\eqref{eq:2-var weak geometric bound}
follows easily, as \eqref{eq:KA bound} implies 
$C^{\otimes n}\le(\ops_2\#_{t}\ops_1)^{\otimes n}=\ops_2^{\otimes n}\#_t\ops_1^{\otimes n}$, whence
\begin{align*}
\Tr C^{\otimes n}X&\le 
\Tr(\ops_2^{\otimes n}\#_{t}\ops_1^{\otimes n})X
\le
(\Tr\ops_2^{\otimes n}X)\#_{t}(\Tr\ops_1^{\otimes n}X)
=
(\Tr \ops_1^{\otimes n}X)^{t}(\Tr \ops_2^{\otimes n}X)^{1-t},
\end{align*}
where the second inequality follows from \eqref{eq:KA posmon} due to the fact that the map
$\E(\valt):=\Tr\bz(\valt)X\jz$ is positive.
Therefore, we only need to prove the converse direction 
\eqref{eq:2-var weak geometric bound}$\imp$\eqref{eq:KA bound}

If $t\in\{0,1\}$ then we actually have the stronger statement
\begin{align*}
\Tr CX\le(\Tr \ops_1 X)^{t}(\Tr \ops_2X)^{1-t},\ds X\in\B(\hil)\p
\ds\ds\ds\iff\ds\ds\ds
C\le\ops_2\#_{t}\ops_1,
\end{align*}
as one can easily see from \eqref{eq:KA cases}. We may therefore assume that $t\in(0,1)$.

Assume first that $\ops_1$ and $\ops_2$ are invertible.
Let $\Gm_t:=\ops_2\#_{t}\ops_1$, and for 
any $n\in\bN$ and 
any $X\in\B(\hil^{\otimes n})\p$, let 
\begin{align*}
\tilde X:=(\Gm_t^{-1/2})^{\otimes n}X(\Gm_t^{-1/2})^{\otimes n}.
\end{align*}
By definition,
\begin{align*}
I=\Gm_t^{-1/2}\bz\ops_2\#_{t}\ops_1\jz\Gm_t^{-1/2}
=
\bz\Gm_t^{-1/2}\ops_2\Gm_t^{-1/2}\jz\#_{t}\bz\Gm_t^{-1/2}\ops_1\Gm_t^{-1/2}\jz
=
\tilde\ops_2\#_{t}\tilde\ops_1,
\end{align*}
where the second equality is due to the transformer identity \eqref{eq:transformer}.
Comparing the leftmost and the rightmost terms above shows by 
a straightforward computation that $\tilde\ops_2=\tilde\ops_1^{-\frac{t}{1-t}}$; in particular, 
$\tilde\ops_1$ and $\tilde\ops_2$ commute, and therefore they can be jointly diagonalized in some orthonormal basis 
$(e_i)_{i=1}^d$ as 
\begin{align*}
\tilde\ops_1=\sum_{i=1}^da_1(i)\pr{e_i},\ds\ds\ds
\tilde\ops_2=\sum_{i=1}^da_2(i)\pr{e_i}.
\end{align*}
Moreover,
\begin{align}\label{eq:geometric normalization2}
I=\tilde\ops_2\#_{t}\tilde\ops_1=\tilde\ops_1^t\tilde\ops_2^{1-t}=\sum_{i=1}^da_1(i)^ta_2(i)^{1-t}\pr{e_i}\ds\imp\ds
a_1(i)^ta_2(i)^{1-t}=1,\ds i\in[d].
\end{align}
Note that \eqref{eq:2-var weak geometric bound} implies that for every $n\in\bN$,
\begin{align}\label{eq:2-var weak geometric bound2}
\Tr \tilde C^{\otimes n}X
=
\Tr C^{\otimes n}\tilde X
\le
(\Tr \ops_1^{\otimes n}\tilde X)^{t}(\Tr \ops_2^{\otimes n}\tilde X)^{1-t}
=
(\Tr \tilde\ops_1^{\otimes n}X)^{t}(\Tr \tilde\ops_2^{\otimes n}X)^{1-t},
\ds\ds X\in\B(\hil^{\otimes n})\p.
\end{align}

Let $P_{n,\lambda}$ and $\Pi_Q^n$ be defined as above, with $\Pi^n_Q$ corresponding to 
an orthonormal basis $(e_i)_{i=1}^d$ that jointly diagonalizes $\tilde A_1$ and $\tilde A_2$.
Then, we have
\begin{align}
\|\tilde C\|_{\infty}^n
 & = \|\tilde C^{\otimes n}\|_{\infty}  
 = \max_{\lambda\in Y_{n,d}}\| \tilde C^{\otimes n}P_{n,\lambda}\|_{\infty}  
 \le \max_{\lambda\in Y_{n,d}}\frac{1}{\dim\kil_{n,\lambda}}\Tr\tilde C^{\otimes n} P_{n,\lambda}
 \label{eq:KA proof1}\\  
 & = \max_{\lambda\in Y_{n,d}}\frac{1}{\dim\kil_{n,\lambda}}\sum_{Q\in\type_{n}([d])}
 \Tr \tilde C^{\otimes n}\Pi^n_QP_{n,\lambda}   \label{eq:KA proof2}\\
 & \le \max_{\lambda\in Y_{n,d}}\frac{1}{\dim\kil_{n,\lambda}}
 \sum_{Q\in\type_{n}([d])}\left(\Tr \tilde \ops_1^{\otimes n}\Pi^n_Q P_{n,\lambda} \right)^t
 \left(\Tr \tilde \ops_2^{\otimes n}\Pi^n_Q P_{n,\lambda}\right)^{1-t}  \label{eq:KA proof3}\\
 & = \max_{\lambda\in Y_{n,d}}\frac{1}{\dim\kil_{n,\lambda}}\sum_{Q\in\type_{n}([d])}
 \left(e^{n\sum_iQ(i)\log a_1(i)}\Tr \Pi^n_QP_{n,\lambda}\right)^t
 \left(e^{n\sum_iQ(i)\log a_2(i)}\Tr \Pi^n_QP_{n,\lambda}\right)^{1-t} \label{eq:KA proof4} \\
 & = \max_{\lambda\in Y_{n,d}}\frac{1}{\dim\kil_{n,\lambda}}\sum_{Q\in\type_{n}([d])}
 e^{n\sum_iQ(i)\log (a_1(i)^ta_2(i)^{1-t})}\Tr \Pi^n_QP_{n,\lambda} \label{eq:KA proof5}\\ 
 & = \max_{\lambda\in Y_{n,d}}\frac{1}{\dim\kil_{n,\lambda}}\sum_{Q\in\type_{n}([d])}\Tr \Pi^n_QP_{n,\lambda} 
 \label{eq:KA proof6} \\
 & = \max_{\lambda\in Y_{n,d}}\frac{1}{\dim\kil_{n,\lambda}}\Tr P_{n,\lambda}
   = \max_{\lambda\in Y_{n,d}}\frac{\dim(\hil_{n,\lambda}\otimes \kil_{n,\lambda})}{\dim\kil_{n,\lambda}}
   = \max_{\lambda\in Y_{n,d}}\dim\hil_{n,\lambda}  
 \le (n+1)^{d(d-1)/2},\label{eq:KA proof7}
\end{align}
where the equalities in the first line are obvious and the inequality follows from \eqref{eq:SW norm inequality},
\eqref{eq:KA proof2} follows from the fact that 
$\Pi^n_Q$ commutes with $P_{n,\lambda}$ and $\sum_{Q\in\type_{n}([d])}\Pi^n_Q=I$,
\eqref{eq:KA proof3} follows from \eqref{eq:2-var weak geometric bound2},
\eqref{eq:KA proof4} is due to \eqref{eq:type prob},
\eqref{eq:KA proof5} is obvious,
\eqref{eq:KA proof6} is due to \eqref{eq:geometric normalization2},
the equalities in \eqref{eq:KA proof7} are again obvious, and the inequality follows from 
\eqref{eq:SW dimension bound}.

Thus, $\|\tilde\gm\|_{\infty}\le\sqrt[n]{(n+1)^{d(d-1)/2}}$, and letting $n\to\infty$ gives 
$\|\tilde\gm\|_{\infty}\le 1$. This shows that $\tilde\gm\le I$, which is equivalent to \eqref{eq:KA bound}.

In the general case, note that \eqref{eq:2-var weak geometric bound} implies that for every 
$\ep>0$,
\begin{align*}
\Tr C^{\otimes n}X\le(\Tr (\ops_1+\ep I)^{\otimes n}X)^{t}(\Tr (\ops_2+\ep I)^{\otimes n}X)^{1-t},\ds\ds\ds X\in\B(\hil^{\otimes n})\p,\ds n\in\bN.
\end{align*}
By the above, 
\begin{align*}
C\le(\ops_2+\ep I)\#_t(\ops_1+\ep I),\ds\ds\ds \ep>0,
\end{align*}
and taking $\ep\searrow 0$ gives \eqref{eq:KA bound}.
\end{proof}

Proposition \ref{prop:weak bounds} and Theorems \ref{thm:2-var char} and \ref{thm:sup bounds implies weak geometric bound} yield immediately the following:

\begin{thm}\label{thm:2-var char2}
Let $\ops_1,\ops_2\in\B(\hil)\p$ be positive semi-definite operators on a finite-dimensional Hilbert space $\hil$, and let $C\in\B(\hil)\p$. The following are equivalent:
\begin{enumerate}
\item
$C\in\C(\{\ops_1,\ops_2\})$, i.e., 
for every $n\in\bN$, 
\begin{align}\label{eq:2-point eq1}
\Tr C^{\otimes n}X\le\max\{\Tr \ops_1^{\otimes n}X,\Tr \ops_2^{\otimes n}X\},\ds\ds\ds X\in\B(\hil^{\otimes n})\p.
\end{align}

\item
There exists a $p\in[0,1]$ such that 
\begin{align}\label{eq:2-point eq2}
C^{\otimes n}\le p\ops_1^{\otimes n}+(1-p)\ops_2^{\otimes n},\ds\ds\ds n\in\bN.
\end{align}

\item
For every $n\in\bN$, there exist $p_n,q_n\in[0,+\infty)$ such that 
\begin{align}\label{eq:2-point as AM}
C^{\otimes n}\le p_n\ops_1^{\otimes n}+q_n\ops_2^{\otimes n},\ds\ds\ds n\in\bN,
\end{align}
and
\begin{align*}
\lim_{n\to+\infty}\frac{1}{n}\log(p_n+q_n)=0.
\end{align*}

\item
There exists a $t\in[0,1]$ such that 
\begin{align}\label{eq:2-point eq3}
\Tr C^{\otimes n}X\le(\Tr \ops_1^{\otimes n}X)^{t}(\Tr \ops_2^{\otimes n}X)^{1-t},\ds\ds\ds X\in\B(\hil^{\otimes n})\p,\ds n\in\bN.
\end{align}

\item
There exists a $t\in[0,1]$ such that 
\begin{align}\label{eq:2-point eq4}
C\le\ops_2\#_{t}\ops_1.
\end{align}
\end{enumerate}
\end{thm}

\begin{rem}
By Theorem \ref{thm:2-var char}, a $t\in[0,1]$ satisfies \eqref{eq:2-point eq4} if and only if it satisfies 
\eqref{eq:2-point eq3}, and for any such $t$, \eqref{eq:2-point eq2} holds with $p=t$.
However, for a $p\in[0,1]$ that satisfies \eqref{eq:2-point eq2}, \eqref{eq:2-point eq3}
(equivalently, \eqref{eq:2-point eq4}) might not hold with $t=p$.
\end{rem}

Theorem \ref{thm:2-var char2} provides a single-copy characterization of $\C(\{\ops_1,\ops_2\})$ as
\begin{align*}
\C(\{\ops_1,\ops_2\})=\{C\in\B(\hil)\p:\,\exists\,t\in[0,1],\,C\le\ops_2\#_t\ops_1\}.
\end{align*}
The following Theorem \ref{thm:2-point maximal operators} gives a single-copy characterization of 
$\max\C(\{\ops_1,\ops_2\})$.

\begin{thm}\label{thm:2-point maximal operators}
Let $\ops_1,\ops_2\in\B(\hil)\p$ be positive semi-definite operators on a finite-dimensional Hilbert space $\hil$, let $P:=\ops_1^0\wedge\ops_2^0$, and  $\sh\ops_1:=\acc{(\ops_1)}{P}$, 
$\sh\ops_2:=\acc{(\ops_2)}{P}$ be the corresponding shorted operators. Then 
\begin{align}\label{eq:2-point maximal operators}
\max\C(\{\ops_1,\ops_2\})=
\begin{cases}
\{\ops_1\},&\ops_1\ge\ops_2,\\
\{\ops_2\},&\ops_1\le\ops_2,\\
\{\ops_1,\ops_2\},&\sh\ops_1\ge\sh\ops_2,\,\ops_1\ngeq\ops_2,\,\ops_1\nleq\ops_2,\\
\{\ops_1,\ops_2\},&\sh\ops_1\le\sh\ops_2,\,\ops_1\ngeq\ops_2,\,\ops_1\nleq\ops_2,\\
\{\ops_2\#_t\ops_1:\,t\in[0,1]\},&\text{otherwise}.
\end{cases}
\end{align} 
\end{thm}
\begin{proof}
According to Theorem \ref{thm:2-var char2},
every $C\in\C(\{\ops_1,\ops_2\})$ satisfies $C\le \ops_2\#_t\ops_1$ for some $t\in[0,1]$, and 
hence $\max\C(\{\ops_1,\ops_2\})\subseteq\{\ops_2\#_t\ops_1:\,t\in[0,1]\}$.
Obviously, $\ops_1,\ops_2\in\C(\{\ops_1,\ops_2\})$.

Assume that $\sh\ops_1\le\sh\ops_2$. By \eqref{eq:KA cases},
\begin{align*}
\ops_2\#_t\ops_1=\begin{cases}
\ops_2,&t=0,\\
\sh\ops_2^{1/2}\bz\sh\ops_2^{-1/2}\sh\ops_1\sh\ops_2^{-1/2}\jz^t\sh\ops_2^{1/2}\le
\sh\ops_2\le\ops_2,&t\in(0,1),\\
\ops_1,&t=1,
\end{cases}
\end{align*}
where in the middle line we used the operator monotonicity of $\id_{[0,+\infty)}^t$ for $t\in(0,1)$.
Thus, 
if $\ops_1\le\ops_2$ then $\max\C(\{\ops_1,\ops_2\})=\{\ops_2\}$, 
if $\ops_1\ge\ops_2$ then $\max\C(\{\ops_1,\ops_2\})=\{\ops_1\}$, 
while if $\ops_1\nleq\ops_2$ and $\ops_1\ngeq\ops_2$ then 
$\max\C(\{\ops_1,\ops_2\})=\{\ops_1,\ops_2\}$.
This proves the assertion in the second and the fourth cases in \eqref{eq:2-point maximal operators}, and the first and the third cases
follow by a completely analogous argument.

Finally, assume that neither 
$\sh\ops_1\le\sh\ops_2$ nor 
$\sh\ops_1\ge\sh\ops_2$ hold.
If $\ops_2\#_s\ops_1\le\ops_2\#_t\ops_1$ holds for some 
$0\le s<t\le 1$ then also 
$(\sh\ops_2^{-1/2}\sh\ops_1\sh\ops_2^{-1/2})^s\le
(\sh\ops_2^{-1/2}\sh\ops_1\sh\ops_2^{-1/2})^t$, whence
$\sh\ops_2^{-1/2}\sh\ops_1\sh\ops_2^{-1/2}\ge P$, and therefore
$\sh\ops_1\ge \sh\ops_2$, contradicting our assumption.
Likewise, if $\ops_2\#_s\ops_1\ge\ops_2\#_t\ops_1$ holds for some 
$0\le s<t\le 1$ then also 
$(\sh\ops_2^{-1/2}\sh\ops_1\sh\ops_2^{-1/2})^s\ge
(\sh\ops_2^{-1/2}\sh\ops_1\sh\ops_2^{-1/2})^t$, whence
$\sh\ops_2^{-1/2}\sh\ops_1\sh\ops_2^{-1/2}\le P$, and therefore
$\sh\ops_1\le \sh\ops_2$, again contradicting our assumption.
Hence, in this case, $\ops_2\#_t\ops_1$ is maximal in $\C(\{\ops_1,\ops_2\})$ for every 
$t\in[0,1]$, and these are all the maximal elements by Theorem \ref{thm:2-var char2}.
\end{proof}

\begin{cor}\label{cor:invertible KA char}
Let $\ops_1,\ops_2\in\B(\hil)\pp$ be positive definite operators such that $\ops_1\nleq\ops_2$ and 
$\ops_1\ngeq\ops_2$. Then
\begin{align*}
\max\C(\{\ops_1,\ops_2\})=\left\{\ops_2\#_t\ops_1:\,t\in[0,1]\right\}.
\end{align*}
\end{cor}
\begin{proof}
Immediate from Theorem \ref{thm:2-point maximal operators}, as in this case $P=1$,
$\acc{(\ops_1)}{P}=\ops_1$,
$\acc{(\ops_2)}{P}=\ops_2$.
\end{proof}

\begin{rem}
None of the cases considered in \eqref{eq:2-point maximal operators} is empty.
The case considered in the last line is disjoint from all the others by definition.
The first two cases intersect in the trivial way only, when $\ops_1=\ops_2$, and they are
disjoint from all the other cases. 
The third and the fourth cases intersect if and only if 
$\tilde\ops_1=\tilde\ops_2$ while $\ops_1$ and $\ops_2$ are incomparable. For an example of this 
situation, consider $a,b,c\in(0,+\infty)$ and 
\begin{align*}
\ops_1:=\begin{bmatrix} a & & \\ & b & \\ & & 0\end{bmatrix},\ds\ds\ds
\ops_2:=\begin{bmatrix} 0 & & \\ & b & \\ & & c\end{bmatrix}.
\end{align*}
Then
\begin{align*}
P:=\ops_1^0\wedge\ops_2^0=\begin{bmatrix} 0 & & \\ & 1 & \\ & & 0\end{bmatrix},\ds\ds\ds
\sh\ops_1=\begin{bmatrix} 0 & & \\ & b & \\ & & 0\end{bmatrix}=
\sh\ops_2.
\end{align*}
Thus,
$\sh\ops_1=\sh\ops_2$, but neither $\ops_1\le\ops_2$ nor $\ops_1\ge\ops_2$ hold, and we have 
$\max\C(\{\ops_1,\ops_2\})=\{\ops_1,\ops_2\}$.
\end{rem}

\begin{rem}
According to \eqref{eq:ac explicit2} and the operator anti-monotonicity of the inverse on PSD operators of equal supports, the condition $\tilde \ops_1\ge\tilde\ops_2$ in \eqref{eq:2-point maximal operators}
is equivalent to $(\ops_1^0\wedge\ops_2^0)\ops_1\inv(\ops_1^0\wedge\ops_2^0)\le(\ops_1^0\wedge\ops_2^0)\ops_2\inv(\ops_1^0\wedge\ops_2^0)$.
\end{rem}

Theorem \ref{thm:2-point maximal operators} solves the mathematical problem posed in
\eqref{eq:holygrail} related to composite state discrimination in the case where
$|\R|=2$ and the Hilbert space is finite-dimensional.
Corollary \ref{cor:invertible KA char} gives a new characterization 
of the weighted Kubo-Ando geometric means as the maximal elements in $\C(\{\ops_1,\ops_2\})$ when 
$\ops_1$ and $\ops_2$ are invertible and incomparable.

This in turn leads to another new characterization  
of the weighted Kubo-Ando geometric means as the maximal block additive operator means that satisfy
the arithmetic-geometric mean inequality. In fact, we have the following even stronger characterization:

\begin{thm}\label{thm:KA maximal}
Let $\Gm$ be a $2$-variable positive operator function
that is block superadditive and 
asymptotically individually AM bounded. 
Then there exists a $t\in[0,1]$ such that $\Gm\le \Gm_{(t,1-t)}^{\KA}$, i.e.,
\begin{align*}
\Gm(\ops_1,\ops_2)\le\Gm_{(t,1-t)}^{\KA}(\ops_1,\ops_2)\ds(=\ops_2\#_t\ops_1),\ds\ds\ds (\ops_1,\ops_2)\in\dom(\Gm).
\end{align*}
If, moreover, $G=G_{\nu}$ is a $\nu$-weighted operator geometric mean, then the above $t$ is unique
as $t=\nu(1)$; in particular, $G_{\nu}\le\Gm_{\nu}^{\KA}$. 
\end{thm}
\begin{proof}
By assumption, for any $(\ops_1,\ops_2)\in\dom(\Gm)$
and any $n\in\bN$, there exists a positive measure $\mu_n$ on $\{1,2\}$ such that  
$\lim_{n\to+\infty}(1/n)\log(\mu_n(1)+\mu_n(2))=0$,
\begin{align*}
\Gm(\ops_1,\ops_2)^{\otimes n}
\le
\mu_n(1)\ops_1^{\otimes n}+\mu_n(2)\ops_2^{\otimes n}.
\end{align*}
Thus, by Theorems \ref{thm:sup bounds implies weak geometric bound} and \ref{thm:2-var char2}, 
there exists a $t\in[0,1]$ such that 
$\Gm(\ops_1,\ops_2)\le\ops_2\#_tA_1=\Gm_{(t,1-t)}^{\KA}(\ops_1,\ops_2)$.
Hence, $T(\ops_1,\ops_2):=\{t\in[0,1]:\,\Gm(\ops_1,\ops_2)\le\ops_2\#_tA_1\}$
is non-empty, and it is easily seen to be closed. 
Indeed, let $(t_n)_{n\in\bN}$ be a sequence in $T(\ops_1,\ops_2)$ converging to some $t\in[0,1]$. If $t\in(0,1)$
then
$\Gm(\ops_1,\ops_2)\le\ops_2\#_{t_n}\ops_1$, $n\in\bN$, implies
$\Gm(\ops_1,\ops_2)\le\ops_2\#_t\ops_1$ according to \eqref{eq:KA cont1}.
If $t=0$ then the sequence $(\ops_2\#_{t_n}\ops_1)_{n\in\bN}$ has at most two limit points, 
$\ops_2\#_0\ops_1=\ops_2$ and $\lim_{s\searrow 0}\ops_2\#_s\ops_1=\acc{(\ops_2)}{\ops_1}\le\ops_2$, 
whence $0\in T(\ops_1,\ops_2)$. The case $t=1$ follows analogously.
Now, for any 
$(A_{i,1},A_{i,2})\in\dom(\Gm)$, $i\in[r]$, 
and any $t\in T\bz\oplus_iA_{i,1},\oplus_iA_{i,2}\jz$,
we have
\begin{align*}
\oplus_{i=1}^r\Gm\bz A_{i,1},A_{i,2}\jz
\le
\Gm\bz\oplus_{i=1}^rA_{i,1},\oplus_{i=1}^rA_{i,2}\jz
\le
\bz\oplus_{i=1}^rA_{i,2}\jz\#_t\bz\oplus_{i=1}^rA_{i,1}\jz
=
\oplus_{i=1}^r\bz A_{i,2}\#_t A_{i,1}\jz,
\end{align*}
where we use the block superadditivity of $\Gm$ in the first inequality.
Thus, $\Gm\bz A_{i,1},A_{i,2}\jz\le \ops_{i,2}\#_t\ops_{i,1}$ for every $i\in[r]$, whence
$t\in\cap_{i=1}^rT(\ops_{i,1},\ops_{i,2})$.
This shows that for any finite collection of $(A_{i,1},A_{i,2})\in\dom(\Gm)$, $i\in[r]$,
$\cap_{i=1}^rT(\ops_{i,1},\ops_{i,2})\ne\emptyset$, and thus, by the compactness of $[0,1]$, we get that 
\begin{align*}
\bigcap_{(\ops_1,\ops_2)\in\dom(\Gm)}T(\ops_1,\ops_2)\ne\emptyset.
\end{align*}
By definition, for any $t$ in the above set,
\begin{align*}
\Gm(A_1,A_2)\le A_2\#_tA_1,\ds\ds\ds (A_1,A_2)\in\dom(\Gm),
\end{align*}
proving the first assertion.

If $G=G_{\nu}$ for some probability distribution $\nu$ on $\{1,2\}$ then 
$G_{\nu}(A_1,A_2)=G_{\nu}^{\KA}(A_1,A_2)$ for any pair of commuting positive definite $A_1,A_2$, and the 
assertion follows from the fact that the classical $2$-variable geometric means are incomparable.
Indeed, if $\ops_1=(p,1-p)$ and $\ops_2=(q,1-q)$ are probability distributions with $0<p<q<1$ then 
$[0,1]\ni t\mapsto p^tq^{1-t}$ is strictly decreasing, while
$[0,1]\ni t\mapsto (1-p)^t(1-q)^{1-t}$ is strictly increasing, whence for any $0\le s<t\le 1$, 
$\ops_1^s\ops_2^{1-s}\nleq\ops_1^t\ops_2^{1-t}$ and 
$\ops_1^s\ops_2^{1-s}\ngeq\ops_1^t\ops_2^{1-t}$.
\end{proof}

As an application of the above result, we show the following.

\begin{prop}\label{prop:no AM-GM}
The following operator geometric means, defined in \eqref{altmean1}--\eqref{altmean4}, 
are not  
asymptotically individually AM bounded:
\begin{enumerate}
\item\label{no AM-GM1}
$G_{t,z}$, $t\in(0,1)$, $z\in(0,+\infty)$;

\item\label{no AM-GM2}
$\wtilde G_{t,z}$, $t\in(0,1)$, $z\in(0,+\infty)$;

\item\label{no AM-GM3}
$\what G_{t,z}$, $t\in(0,1)$, $z\in(1,+\infty]$.
\end{enumerate}
\end{prop}
\begin{proof}
It is clear from their definitions that all the above operator geometric means are block additive.
Hence, by Theorem \ref{thm:KA maximal}, it is sufficient to show that for any fixed $t\in(0,1)$, 
none of them is upper bounded by $\Gm_{(t,1-t)}^{\KA}$, which we show in Lemma \ref{lemma:no KA bound} below.
\end{proof}

For the next lemma, it will be useful to extend the definitions of the quantities defined in 
\eqref{altmean1}--\eqref{altmean4} for the case when the arguments are not assumed to be invertible. Note that 
\eqref{altmean1}--\eqref{altmean3} readily make sense when $\ops_1$ and $\ops_2$ are only assumed to be positive semi-definite, and it is straightforward to verify that 
\begin{align}\label{eq:gmean extensions}
\Gm'_t(\ops_1,\ops_2)=\lim_{\ep\searrow 0}\Gm'_t(\ops_1+\ep I,\ops_2+\ep I)
\end{align}
holds for any $t\in(0,1)$ 
and $\ops_1,\ops_2\in\B(\hil)\p$,
if $\Gm'_t$ is equal to any of the expressions in 
\eqref{altmean1}--\eqref{altmean3}.
For $\Gm'_t=\what G_{t,+\infty}$, \eqref{eq:gmean extensions} holds with the definition
\begin{align*}
\what G_{t,+\infty}(\ops_1,\ops_2):=
P\exp\bz tP(\logn\ops_1)P+(1-t)P(\logn\ops_2)P\jz,
\end{align*}
where $P:=\ops_1^0\wedge\ops_2^0$; see \cite[Section 4]{HP-GT}. Moreover, by 
\cite[Theorem 2.3]{NoAM-GM}, 
\begin{align*}
\what G_{t,+\infty}(\ops_1,\ops_2)=\lim_{z\to+\infty}\what G_{t,z}(\ops_1,\ops_2)
\end{align*}
holds also for general positive semi-definite $\ops_1,\ops_2\in\B(\hil)\p$.

\begin{lemma}\label{lemma:no KA bound}
For a fixed $t\in(0,1)$, let $\Gm'_{t}$ be any of the operator geometric means in Proposition
\ref{prop:no AM-GM}. Then there exist positive definite operators 
$\ops_1,\ops_2\in\B(\bC^2)\pp$ such that 
\begin{align}\label{eq:no KA bounds}
\Gm'_t(\ops_1,\ops_2)\nleq \ops_2\#_t\ops_1.
\end{align} 
\end{lemma}
\begin{proof}
Note that by \eqref{eq:gmean extensions},
it is sufficient to find $\ops_1,\ops_2\in\B(\bC^2)\p$ PSD (not necessarily invertible) operators such that 
\eqref{eq:no KA bounds} holds.

It is well known that $(\ops_2\#_t\ops_1)^0=\ops_2^0\wedge\ops_1^0$ 
\cite[Proposition 3.7]{KA}. Thus, in particular, if $\ops_1=\pr{\psi}$ and $\ops_2=\pr{\phi}$ are rank $1$
projections with $|\inner{\psi}{\phi}|\in(0,1)$, then for any $t\in(0,1)$, $\ops_2\#_t\ops_1=0$, while it is easy to see that 
$G_{t,z}(\ops_1,\ops_2)\ne 0$, $\wtilde G_{t,z}(\ops_1,\ops_2)\ne 0$, $z\in(0,+\infty)$, 
proving the assertion for these two families.

Next, let 
\begin{align*}
\ops_1:=\pr{\psi}\ds\text{with}\ds \psi:=\frac{1}{\sqrt{2}}\begin{bmatrix} 1 \\ 1\end{bmatrix},\ds\ds\ds
\ops_2:=\begin{bmatrix}a & 0 \\ 0 & b\end{bmatrix}\ds\text{with}\ds 0<a<b,
\end{align*}
and let $p:=1/z$ for some $z\in(1,+\infty)$. Then a straightforward computation yields
\begin{align*}
\what G_{t,z}(\ops_1,\ops_2)=
\bz\ops_2^{\frac{p}{2}}\bz\ops_2^{-\frac{p}{2}}\pr{\psi}\ops_2^{-\frac{p}{2}}\jz^t\ops_2^{\frac{p}{2}}\jz^{\frac{1}{p}}
=
\norm{\ops_2^{-\frac{p}{2}}\psi}^{\frac{2(t-1)}{p}}\pr{\psi}
=\bz\frac{a^{-p}+b^{-p}}{2}\jz^{\frac{t-1}{p}}\pr{\psi}.
\end{align*}
Thus, by the strict convexity of $\id_{(0,+\infty)}^{\frac{1}{p}}$, we have 
\begin{align*}
\what G_{t,z}(\ops_1,\ops_2)
=
\bz\frac{a^{-p}+b^{-p}}{2}\jz^{\frac{t-1}{p}}\pr{\psi}
\gneq
\bz\frac{a\inv+b\inv}{2}\jz^{t-1}\pr{\psi}
=
\what G_{t,1}(\ops_1,\ops_2)
=
\ops_2\#_t\ops_1,
\end{align*}
proving the assertion for $z\in(1,+\infty)$. 
The case $z=+\infty$ follows 
as $\what G_{t,+\infty}(\ops_1,\ops_2)=\sqrt{ab}^{1-t}\pr{\psi}\gneq\ops_2\#_t\ops_1$,
due to the strict inequality between the geometric and the harmonic means for unequal arguments.
\end{proof}

\section{Exact characterization of the weighted Kubo-Ando geometric means}

We can actually get an even stronger characterization of the Kubo-Ando means
than the one in Theorem \ref{thm:KA maximal}, under some mild additional conditions. 
For the following, we introduce the notation
\begin{align*}
\dom(\Gm)\pp:=\left\{(\ops_1,\ops_2)\in\dom(\Gm):\,\ops_1>0,\,\ops_2>0\right\}.
\end{align*}
We will say that a positive operator function $\Gm$ is \ki{regular} if for any 
$(\ops_1,\ops_2)\in\dom(\Gm)$, there exists a $\delta>0$ such that for every 
$\ep\in(0,\delta)$, $(\ops_1+\ep I,\ops_2+\ep I)\in\dom(\Gm)$, and 
\begin{align*}
\Gm(\ops_1,\ops_2)=\lim_{\ep\searrow 0}\Gm(\ops_1+\ep I,\ops_2+\ep I).
\end{align*}

For a positive definite $A\in\B(\bC^d)\pp$, let us introduce the shorthand notation
\begin{align*}
A\tinv:=(A\inv)\trans=(A\trans)\inv
\end{align*}
for the transpose of the inverse of $A$.
Here, the transposition is taken in an arbitrary orthonormal basis (ONB)
$(e_i)_{i\in[d]}$ of $\bC^d$, defined as 
$X\trans:=\sum_{i,j\in[d]}\inner{e_i}{Xe_j}\diad{e_j}{e_i}$, $X\in\B(\bC^d)$. 
Note that transpositions corresponding to different ONBs are generally different,
but are connected by a unitary conjugation. In particular, the condition in \ref{converse bound2} of 
Proposition \ref{prop:KA char} below holds for a transposition in some ONB if and only if it holds
for every transposition, due to the unitary covariance of $\Gm$.

\begin{prop}\label{prop:KA char}
Let $\Gm$ be a $2$-variable positive operator function that is 
block superadditive and asymptotically individually AM bounded. 
Assume, moreover, that for any 
$d\in\bN$ and any 
positive definite $(\ops_1,\ops_2)\in\dom_d(\Gm)\pp$,
$\Gm(\ops_1,\ops_2)$ is also positive definite, and 
one of the following holds:
\begin{enumerate}
\item\label{converse bound1}
$(\ops_1\inv,\ops_2\inv)\in\dom(\Gm)$,  and
\begin{align}\label{eq:inverse invariance}
\Gm(\ops_1\inv,\ops_2\inv)=\Gm(\ops_1,\ops_2)\inv. 
\end{align}
\item\label{converse bound2}
$(\ops_1\tinv,\ops_2\tinv)\in\dom(\Gm)$,  and
\begin{align}\label{eq:inverse trans invariance}
\Gm(\ops_1\tinv,\ops_2\tinv)=\Gm(\ops_1,\ops_2)\tinv,
\end{align}
for a transposition taken in some orthonormal basis of $\bC^d$.
\end{enumerate}
Then there exists a $t\in[0,1]$ such that for any two positive definite $(\ops_1,\ops_2)\in\dom(\Gm)\pp$, 
\begin{align}\label{eq:KA equality1}
\Gm(\ops_1,\ops_2)=\ops_2\#_t\ops_1.
\end{align}
If, moreover, $\Gm$ is regular, then \eqref{eq:KA equality1} holds for every 
$(\ops_1,\ops_2)\in\dom(\Gm)$.
\end{prop}
\begin{proof}
By Theorem \ref{thm:KA maximal}, there exists a $t\in[0,1]$ such that for any 
$(\ops_1,\ops_2)\in\dom(\Gm)$, $\Gm(\ops_1,\ops_2)\le\ops_2\#_t\ops_1$. 
Now, if $(\ops_1,\ops_2)\in\dom(\Gm)\pp$ and \ref{converse bound2} holds then we also have
\begin{align*}
\Gm(\ops_1,\ops_2)=\Gm(\ops_1\tinv,\ops_2\tinv)\tinv\ge
\bz\ops_2\tinv\#_t\ops_1\tinv\jz\tinv
=
\ops_2\#_t\ops_1,
\end{align*}
completing the proof of \eqref{eq:KA equality1}. 
The proof under condition \ref{converse bound1} goes the same way. 
The assertion about a regular $\Gm$ then follows immediately.
\end{proof}

Note that Proposition \ref{prop:KA char} provides a slightly different general framework 
to disproving the asymptotic AM boundedness of an operator geometric mean compared to 
Theorem \ref{thm:KA maximal}.
Indeed, while the application based on the latter requires disproving an inequality, the approach based on the former requires disproving an equality, while additionally verifying that at least 
one of the invariance properties
\eqref{eq:inverse invariance} or \eqref{eq:inverse trans invariance} holds.
Regarding this second approach, it is very easy to show that for a fixed $t\in(0,1)$ and $z\in(0,+\infty)$, all of the operator geometric means
$\Gm_{t,z}$, 
$\wtilde\Gm_{t,z}$, 
$\what\Gm_{t,z}$ are block additive and satisfy both invariance conditions
\eqref{eq:inverse invariance} and \eqref{eq:inverse trans invariance}, but none of them 
are equal to $\Gm_{(t,1-t)}^{\KA}$, with the exception of $\what\Gm_{t,1}$. 
Therefore, they are not asymptotically individually AM bounded.
In particular, we can complete 
Proposition \ref{prop:no AM-GM} using this second approach as follows:

\begin{prop}
For any $t\in(0,1)$ and $z\in(0,1)$, 
$\what\Gm_{t,z}$ is not asymptotically individually AM bounded.
\end{prop}
\begin{proof}
Using the same construction as in the proof of Lemma 
\ref{lemma:no KA bound}, we can now appeal to the strict concavity of 
$\id_{(0,+\infty)}^{\frac{1}{p}}$ to obtain that 
$\what\Gm_{t,z}(\ops_1,\ops_2)\ne\what\Gm_{t,1}(\ops_1,\ops_2)$, and the 
conclusion then follows by Proposition \ref{prop:KA char}. 
\end{proof}

\begin{rem}
Further notions of $2$-variable weighted matrix geometric means have been studied recently in the literature; e.g., different notions of spectral geometric means \cite{FriedlerPtak1997,LeeYim2007,DinhTamVuong2024}
and their quasi versions \cite{NoAM-GM}. It was shown in \cite{NoAM-GM} that none of these means satisfy the arithmetic-geometric mean inequality. Since it is easy to see that they are all 
block additive, satisfy \eqref{eq:inverse invariance}, and are
different from the Kubo-Ando geometric means, Proposition \ref{prop:KA char} yields that they are not even 
asymptotically individually AM bounded.
\end{rem}

\medskip

Note that the invariance property \eqref{eq:inverse invariance} is quite natural to be imposed on any notion of non-commutative weighted geometric mean as it trivially holds for weighted geometric means of scalars.
In contrast, the notion of transposition does not have a scalar version (alternatively, it becomes trivial), whence the invariance property \eqref{eq:inverse trans invariance} seems more technical than conceptual.
However, we show below that it follows from other properties that are very natural to require of a 
notion of non-commutative weighted geometric mean. For this, we will use some simple facts about antisymmetric tensor products.

Let $\wedge^k\bC^d$ denote the $k$-th antisymmetric tensor power of $\bC^d$, and for any 
$A\in\B(\bC^d)$, let $A^{\wedge k}$ denote the restriction of $A^{\otimes k}$ onto $\wedge^k\bC^d$. 
Recall that if $(e_i)_{i=1}^d$ is an orthonormal basis of $\bC^d$ then 
\begin{align*}
e_{i_1}\wedge\ldots\wedge e_{i_k}
:=
\frac{1}{\sqrt{k!}}\sum_{\sigma}\ep(\sigma)e_{i_{\sigma(1)}}\ootimes e_{i_{\sigma(k)}},\ds\ds\ds
1\le i_1<\ldots<i_k\le d,
\end{align*}
forms an orthonormal basis in $\wedge^k\bC^d$. (In the above formula, the summation runs over all permutations of $k$ elements, and $\ep(\sigma)$ denotes the sign of the permutation $\sigma$.)
Then
\begin{align}
&V:\,\bC\to \wedge^{d}\bC^d,\ds\ds\ds\ds\ds Ve_1:=e_1\wedge\ldots\wedge e_d,
\label{eq:V def}\\
&U:\,\bC^d\to \wedge^{d-1}\bC^d,\ds\ds\ds Ue_i:=
(-1)^{i-1}e_1\wedge\ldots\wedge e_{i-1}\wedge e_{i+1}\wedge\ldots\wedge e_{d},\ds\ds i\in[d],
\label{eq:U def}
\end{align} 
define unitaries, and a straightforward computation verifies that for any $A\in\B(\bC^d)$, 
\begin{align}\label{eq:d power}
V^* A^{\wedge d}V=\det(A),
\end{align}
where $\det(A)$ stands for the determinant of $A$, interpreted here as a linear operator on the one-dimensional space $\wedge^{d}\bC^d$, 
and for any invertible $A$,
\begin{align}\label{eq:d-1 power}
U^*A^{\wedge(d-1)}U=(\det (A))A\tinv,
\end{align}
where the transposition is taken in the given orthonormal basis. 
Here, $(\inner{e_i}{U^*A^{\wedge(d-1)}Ue_j})_{i,j=1}^d$ is the \ki{cofactor matrix} of 
$(\inner{e_i}{Ae_j})_{i,j=1}^d$.

We consider the invariance property \eqref{eq:inverse trans invariance} more generally for
$\Y$-variable positive operator functions, 
where $\Y$ may be an arbitrary non-empty set,
as it may be of interest in itself.
In fact, for Lemmas \ref{lemma:determinant formula}--\ref{lemma:KA char2} below, 
we do not even need to assume that $\Y$ is equipped with a $\sigma$-algebra $\F$ and 
$y\mapsto \ops_y$ is measurable for $A\in\dom(\Gm)$ (alternatively, we may take $\F$
to be the full power set of $\Y$).
We introduce the following technical properties of a $\Y$-variable positive operator function $\Gm$ that will be used in the following.
\begin{itemize}
\item
\ki{Strict positivity:} 
$\ops\in\dom(\Gm)\pp$ $\imp$ $\Gm(\ops)>0$.

\item
\ki{Strict scalar positivity:} 
$\ops\in\dom_1(\Gm)\pp$ $\imp$ $\Gm(\ops)>0$.

\item
\ki{Reduction property:} 
For any $(\ops_y)_{y\in\Y}\in\dom_d(\Gm)$, and any subspace $\{0\}\ne\kil\subseteq\bC^d$ that is invariant under all 
$A_y$, we have $\bz A_y\big\vert_{\kil}\jz_{y\in\Y}\in\dom_{\kil}(\Gm)$, where 
$A_y\big\vert_{\kil}$ is interpreted as an operator on $\kil$.

\item
\ki{Scalar tensor multiplicativity:} 
For any $a\in\dom_1(\Gm)$ and any $A\in\dom(\Gm)$, $(a_yA_y)_{y\in\Y}\in\dom(\Gm)$, and 
\begin{align*}
\Gm((a_yA_y)_{y\in\Y})=\Gm(a)\Gm(A).
\end{align*}
\end{itemize}
Note that both scalar tensor multiplicativity and tensor multiplicativity 
defined in Section \ref{sec:operator geometric means}
are special cases of strong tensor multiplicativity, which would mean that 
$(\ops_y)_{y\in\Y},(\tilde\ops_y)_{y\in\Y}\in\dom(\Gm)$ implies 
$(\ops_y\otimes \tilde\ops_y)_{y\in\Y}\in\dom(\Gm)$, and 
$\Gm ((\ops_y\otimes \tilde\ops_y)_{y\in\Y})=\Gm(\ops)\otimes\Gm(\tilde\ops)$. 
However, we will not need this stronger property.
Note also that for a finite $\Y$, a $\Y$-variable matrix geometric mean is usually either defined on all families of positive definite operators indexed by $\Y$, or on all families of PSD operators indexed by $\Y$, and in either case, the reduction property holds trivially. 

\begin{lemma}\label{lemma:determinant formula}
\emph{(Determinant identity)}
Let $\Gm$ be a $\Y$-variable positive operator function that is tensor multiplicative, 
block additive, and has the reduction property. 
Then 
\begin{align}\label{eq:determinant identity}
\ops=(\ops_y)_{y\in\Y}\in\dom(\Gm)\ds\ds\imp\ds\ds
\left\{\begin{array}{l}
\det \ops:=(\det \ops_y)_{y\in\Y}\in\dom(\Gm),\\
\s\\
\det \Gm(A)=\Gm(\det A).
\end{array}\right.
\end{align}
\end{lemma}
\begin{proof}
Let $A=(\ops_y)_{y\in\Y}\in\dom_d(\Gm)$.
We may write $(\bC^d)^{\otimes d}=\wedge^{d}\bC^d\oplus\kil$, and
choose a $V$ as in \eqref{eq:V def}, corresponding to an arbitrary ONB. 
Since $\wedge^d\bC^d$ is a reducing subspace of $(\bC^d)^{\otimes d}$
for all $\ops_y^{\otimes d}$, 
we get by \eqref{eq:d power} that
\begin{align}\label{eq:determinant identity proof1}
A^{\otimes d}
=
\bz V(\det \ops_y )V^*\oplus\ops_y^{\otimes d}\big\vert_{\kil}\jz_{y\in\Y},
\end{align}
and the reduction property and unitary covariance  imply that $\det(A)\in\dom(\Gm)$ and 
$\bz\ops_y^{\otimes d}\big\vert_{\kil}\jz_{y\in\Y}\in\dom(\Gm)$.
Moreover,
\begin{align*}
V(\det\Gm(A))V^*\oplus\Gm(A)^{\otimes d}\big\vert_{\kil}
&=
\Gm(A)^{\otimes d}
=
\Gm(A^{\otimes d})
=
\Gm\bz \bz V(\det(\ops_y))V^*\oplus\ops_y^{\otimes d}\big\vert_{\kil}\jz_y\jz,\\
&=
V\Gm\bz\det\ops\jz V^*\oplus\Gm\bz(\ops_y^{\otimes d}\big\vert_{\kil})_{y\in\Y}\jz,
\end{align*}
where the first equality is due to \eqref{eq:d power}, 
the second equality is due to tensor multiplicativity,
the third equality follows from \eqref{eq:determinant identity proof1},
and the last equality from the reduction property, block additivity,
and unitary covariance. Thus,
\begin{align*}
\det\Gm(A)=\Gm\bz\det\ops\jz,
\end{align*}
as stated.
\end{proof}

\begin{lemma}\label{lemma:operator positivity}
Let $\Gm$ be a $\Y$-variable positive operator function that satisfies 
\eqref{eq:determinant identity}. Then $\Gm$ is strictly positive
if and only if it is strictly scalar positive.
\end{lemma}
\begin{proof}
Obvious from the fact that a PSD operator $X$ is positive definite if and only if $\det(X)>0$.
\end{proof}

\begin{lemma}\label{lemma:KA char2}
Let $\Gm$ be a $\Y$-variable positive operator function that is 
strictly scalar positive,
tensor multiplicative, 
scalar tensor multiplicative,  block additive, and has the reduction property. 
Assume, moreover, that for every $(A_y)_{y\in\Y}\in\cD_1(\Gm)\pp$, also 
$(A_y\inv)_{y\in\Y}\in\cD(\Gm)\pp$.
Then for any 
$(\ops_y)_{y\in\Y}\in\dom(\Gm)\pp$ and any transposition $\trans$, 
$(\ops_y\tinv)_{y\in\Y}\in\dom(\Gm)\pp$,  
$\Gm((\ops_y)_{y\in\Y})$ is also positive definite, and 
\begin{align}\label{eq:inverse mean}
\Gm(\ops\tinv)=\Gm(A)\tinv, 
\end{align}
where $\ops\tinv=(\ops_y\tinv)_{y\in\Y}$.
\end{lemma}
\begin{proof}
Let $A=(\ops_y)_{y\in\Y}\in\dom_d(\Gm)\pp$.
By Lemmas \ref{lemma:determinant formula} and \ref{lemma:operator positivity}, 
$\Gm$ is strictly positive, and therefore
$\Gm(A)>0$. 

Assume first that $d=1$. Then, by assumption, $(A_y\inv)_{y\in\Y}\in\cD(\Gm)\pp$.
By scalar tensor multiplicativity, $(A_yA_y\inv)_{y\in\Y}=(1)_{y\in\Y}\in\cD(\Gm)$,
and $G((1)_{y\in\Y})=G(A)G(A\inv)>0$.
Again by scalar tensor multiplicativity, $G((1)_{y\in\Y})=G((1^2)_{y\in\Y})
=G((1)_{y\in\Y})^2$, whence $G((1)_{y\in\Y})=1$. Thus, by $G((1)_{y\in\Y})=G(A)G(A\inv)$
established above, we get $G(A\inv)=G(A)\inv$. Since in one dimension,
the inverse and the inverse transpose coincide, this proves the assertion.

For the rest, let $d\ge 2$.
We have $(\bC^d)^{\otimes(d-1)}=\wedge^{d-1}\bC^d\oplus\tilde\kil$, and
we may choose $U$ as in \eqref{eq:U def}, corresponding to the ONB in which the given 
transposition is defined. Then 
\begin{align}
&\det(\Gm(A))\,U\Gm(A)\tinv U^*\oplus\Gm((\ops_y)_{y\in\Y})^{\otimes (d-1)}\big\vert_{\tilde\kil}\nn\\
&\ds=
\Gm(A)^{\otimes (d-1)}
=
\Gm(A^{\otimes (d-1)})\nn\\
&\ds=
\Gm\bz\bz(\det\ops_y)\,U\ops_y\tinv U^*\oplus\ops_y^{\otimes (d-1)}\big\vert_{\tilde\kil}\jz_y\jz\nn\\
&\ds=
\Gm\bz\bz(\det\ops_y)\,U\ops_y\tinv U^*\jz_{y\in\Y}\jz
\oplus\Gm\bz\bz\ops_y^{\otimes (d-1)}\big\vert_{\tilde\kil}\jz_{y\in\Y}\jz,
\label{eq:KA char proof}
\end{align}
where the first and the third equalities are due to \eqref{eq:d-1 power}, 
the second equality is due to tensor multiplicativity, 
and the fourth equality follows by block additivity.

By tensor multiplicativity and the reduction property, $\bz\det(\ops_y)\,U\ops_y\tinv U^*\jz_{y\in\Y}\in\cD(\Gm)$.
By Lemma \ref{lemma:determinant formula}, $\det A\in\cD(\Gm)\pp$, and thus also 
$((\det A_y)\inv)_{y\in\Y}\in\cD(\Gm)\pp$, according to the assumption made in the statement. 
Thus,
\begin{align*}
\bz U\ops_y\tinv U^*\jz_{y\in\Y}&=
\bz\left[(\det\ops_y)\inv\right]\left[(\det\ops_y)\,U\ops_y\tinv U^*\right]\jz_{y\in\Y}\in\cD(\Gm),
\end{align*}
due to scalar tensor multiplicativity. Therefore,
\begin{align*}
\Gm\bz\bz(\det\ops_y)\,U\ops_y\tinv U^*\jz_{y\in\Y}\jz=
\Gm\bz \bz\det\ops_y\jz_{y\in\Y}\jz\, U\Gm\bz\bz\ops_y\tinv\jz_{y\in\Y}\jz U^*
=
\det(\Gm(A))\,U\Gm(A\tinv) U^*,
\end{align*}
where the first equality follows by scalar tensor multiplicativity and unitary covariance, 
and the second equality is due to Lemma \ref{lemma:determinant formula}.
Writing this back into \eqref{eq:KA char proof} yields
\begin{align*}
\det(\Gm(A))\,U\Gm(A)\tinv U^*\oplus\Gm((\ops_y)_{y\in\Y})^{\otimes (d-1)}\big\vert_{\tilde\kil}
=
\det(\Gm(A))\,U\Gm(A\tinv) U^*\oplus\Gm((\ops_y)_{y\in\Y})^{\otimes (d-1)}\big\vert_{\tilde\kil},
\end{align*}
from which \eqref{eq:inverse mean} follows.
\end{proof}

Proposition \ref{prop:KA char} and Lemma \ref{lemma:KA char2} yield immediately the following
alternative characterization of the weighted Kubo-Ando geometric means:

\begin{thm}\label{thm:KA char}
Let $\Gm$ be a $2$-variable positive operator function that is 
strictly scalar positive,
tensor multiplicative, scalar tensor multiplicative, asymptotically
individually AM bounded,
block additive, and has the reduction property. 
Assume, moreover, that for every $(a_1,a_2)\in\cD_1(\Gm)\pp$, also 
$(a_1\inv,a_2\inv)\in\cD(\Gm)\pp$.
Then there exists a $t\in[0,1]$ such that for any two positive definite $(\ops_1,\ops_2)\in\dom(\Gm)\pp$, 
\begin{align}\label{eq:KA equality2}
\Gm(\ops_1,\ops_2)=\ops_2\#_t\ops_1.
\end{align}
If, moreover, $\Gm$ is regular, then \eqref{eq:KA equality2} holds for every 
$(\ops_1,\ops_2)\in\dom(\Gm)$.
\end{thm}

\begin{cor}\label{cor:KA char}
Let $\Gm_{\nu}$ be a $2$-variable $\nu$-weighted operator geometric mean
that is strictly scalar positive,
tensor multiplicative, scalar tensor multiplicative, 
asymptotically individually AM bounded,
block additive, and has the reduction property. 
Assume, moreover, that for every $(a_1,a_2)\in\cD_1(\Gm_{\nu})\pp$, also 
$(a_1\inv,a_2\inv)\in\cD(\Gm_{\nu})\pp$.
Then $\Gm_{\nu}(\ops_1,\ops_2)=\Gm_{\nu}^{\KA}(\ops_1,\ops_2)$ 
for every $(\ops_1,\ops_2)\in\dom(\Gm_{\nu})\pp$, and also for 
every $(\ops_1,\ops_2)\in\dom(\Gm_{\nu})$, if $\Gm_{\nu}$ is regular. 
\end{cor}
\begin{proof}
Follows immediately from Theorem \ref{thm:KA char} in the same way as 
in the last paragraph of the 
proof of Theorem \ref{thm:KA maximal}.
\end{proof}

\section{Commutative case}
\label{sec:comm}

Here, we prove \eqref{eq:classical chat0} of Theorem \ref{thm:classical maxCR}.
The proof goes almost exactly the same way as that of Theorem \ref{thm:2-var char},
with the commutative geometric mean in place of the Kubo-Ando geometric mean.
We give it in full detail for readers' convenience.
To make the connection with Theorem \ref{thm:2-var char} clearer, we adopt a notational convention that is different from that of Theorem \ref{thm:classical maxCR}; the latter can be recovered by choosing
$\Y:=\R$, $A_y=y$, $y\in\R$. 

For what follows, let $\hil$ be a finite-dimensional Hilbert space with $d:=\dim\hil$, let 
$\Y$ be a compact Hausdorff space, and $A:\,\Y\to\B(\hil)\p$ be continuous. 
We will assume that there exists an orthonormal basis $(e_i)_{i=1}^d$ in $\hil$ such that each 
$A_y$ is diagonal in this basis as 
\begin{align}\label{eq:comm diagonalizing}
A_y=\sum_{i=1}^da_y(i)\pr{e_i},\ds\ds\ds y\in\Y.
\end{align}

Recall the definitions of type classes and the Schur-Weyl decomposition from Section \ref{sec:2-var}.

\begin{thm}\label{thm:comm char}
In the above setting, let $\nu$ be a probability measure on $\Y$ and $C\in\B(\hil)\p$. Then
\begin{align}\label{eq:comm weak geometric bound}
\Tr C^{\otimes n}X
\le
\exp\bz\int_{\Y}\log \Tr A_y^{\otimes n}X\,d\nu(y)\jz,\ds\ds\ds 
X\in\B(\hil^{\otimes n})\p,\ds n\in\bN,
\end{align}
if and only if
\begin{align}\label{eq:comm bound}
C\le
\Gm_{\nu}(A):=
\exp\bz\int_{\Y}\log A_y\,d\nu(y)\jz:=
\sum_{i=1}^d\pr{e_i}\exp\bz\int_{\Y}\log a_y(i)\,d\nu(y)\jz.
\end{align}
\end{thm}
\begin{proof}
It is sufficient to prove that for any $\ep>0$, 
\begin{align*}
\Tr C^{\otimes n}X
\le
\exp\bz\int_{\Y}\log \Tr (A_y+\ep I)^{\otimes n}X\,d\nu(y)\jz,\ds\ds\ds 
X\in\B(\hil^{\otimes n})\p,\ds n\in\bN,
\end{align*}
if and only if
\begin{align*}
C\le
\Gm_{\nu}(A+\ep I)=
\exp\bz\int_{\Y}\log (A_y+\ep I)\,d\nu(y)\jz=
\sum_{i=1}^d\pr{e_i}\exp\bz\int_{\Y}\log (a_y(i)+\ep)\,d\nu(y)\jz,
\end{align*}
since the original assertion then follows by taking $\ep\searrow 0$.
Hence, for the rest we may assume that there exists some $\ep>0$ such that $A_y\ge \ep I$, $y\in\Y$. 

The implication \eqref{eq:comm bound}$\imp$\eqref{eq:comm weak geometric bound}
follows easily, as \eqref{eq:comm bound} implies 
$C^{\otimes n}\le (\Gm_{\nu}(A))^{\otimes n}=\Gm_{\nu}(A^{\otimes n})$
whence for any density operator $X\in\S(\hil^{\otimes n})$, 
\begin{align}
\Tr C^{\otimes n}X
\le 
\Tr \Gm_{\nu}(A^{\otimes n})X
&=
\sum_{\vecc{i}\in[d]^n}
\inner{e_{\vecc{i}}}{Xe_{\vecc{i}}}
\exp\bz\int_{\Y}\log a_y^{\otimes n}(\vecc{i})\,d\nu(y)\jz\nn\\
&\le
\exp\bz\int_{\Y}\log \bz \sum_{\vecc{i}\in[d]^n}\inner{e_{\vecc{i}}}{Xe_{\vecc{i}}}
a_y^{\otimes n}(\vecc{i})\jz\,d\nu(y)\jz\nn\\
&=
\exp\bz\int_{\Y}\log\Tr A_y^{\otimes n}X\,d\nu(y)\jz,
\label{eq:comm char proof0}
\end{align}
where the inequality follows from a well-known concavity property; see, e.g., 
Sections 3.1.5 and 3.1.8 in \cite{BoydVandenberghe2004}.
Multiplying $X$ by an arbitrary positive number $\lambda>0$ preserves the inequality with $\lambda X$ in place of $X$, and hence \eqref{eq:comm char proof0} holds for any $X\in\B(\hil^{\otimes n})\p$. 
This proves \eqref{eq:comm weak geometric bound}.

Assume now \eqref{eq:comm weak geometric bound}.
%
Let $\Gm_{\nu}:=\Gm_{\nu}(A)$, and for 
any $n\in\bN$ and 
any $X\in\B(\hil^{\otimes n})\p$, let 
\begin{align*}
\tilde X:=(\Gm_{\nu}^{-1/2})^{\otimes n}X(\Gm_{\nu}^{-1/2})^{\otimes n}.
\end{align*}
For every $y\in\Y$, define $\tilde a_y(i):=\sinner{e_i}{\tilde A_ye_i}$, so that 
$\tilde A_y=\sum_{i=1}^d\tilde a_y(i)\pr{e_i}$.
By definition,
\begin{align}\label{eq:comm char proof-1}
\int_{\Y}\log\tilde A_y\,d\nu(y)
=
-\log \Gm_{\nu}+\int_{\Y}\log A_y\,d\nu(y)=0.
\end{align}
Note that \eqref{eq:comm weak geometric bound} implies that for every $n\in\bN$
and every $X\in\B(\hil^{\otimes n})\p$,
\begin{align}\label{eq:comm weak geometric bound2}
\Tr \tilde C^{\otimes n}X
=
\Tr C^{\otimes n}\tilde X
\le
\exp\bz\int_{\Y}\log\Tr A_y^{\otimes n}\tilde X\,d\nu(y)\jz
=
\exp\bz\int_{\Y}\log\Tr \tilde A_y^{\otimes n} X\,d\nu(y)\jz.
\end{align}

Let $P_{n,\lambda}$ and $\Pi_Q^n$ be defined as 
in the beginning of Section \ref{sec:2-var}, with $\Pi^n_Q$ corresponding to 
the basis $(e_i)_{i=1}^d$ given in \eqref{eq:comm diagonalizing}.
Then, we have
\begin{align}
\|\tilde C\|_{\infty}^n
 & = \|\tilde C^{\otimes n}\|_{\infty}  
 = \max_{\lambda\in Y_{n,d}}\| \tilde C^{\otimes n}P_{n,\lambda}\|_{\infty}  
 \le \max_{\lambda\in Y_{n,d}}\frac{1}{\dim\kil_{n,\lambda}}\Tr\tilde C^{\otimes n} P_{n,\lambda}
 \label{eq:comm proof1}\\  
 & = \max_{\lambda\in Y_{n,d}}\frac{1}{\dim\kil_{n,\lambda}}\sum_{Q\in\type_{n}([d])}
 \Tr \tilde C^{\otimes n}\Pi^n_QP_{n,\lambda}   \label{eq:comm proof2}\\
 & \le \max_{\lambda\in Y_{n,d}}\frac{1}{\dim\kil_{n,\lambda}}
 \sum_{Q\in\type_{n}([d])}
\exp\bz\int_{\Y}\log\Tr \tilde A_y^{\otimes n} \Pi^n_QP_{n,\lambda}\,d\nu(y)\jz
  \label{eq:comm proof3}\\
 & = \max_{\lambda\in Y_{n,d}}\frac{1}{\dim\kil_{n,\lambda}}
 \sum_{Q\in\type_{n}([d])}
\Tr \Pi^n_QP_{n,\lambda}\exp\bz\int_{\Y}n\sum_iQ(i)\log \tilde a_y(i)\,d\nu(y)\jz
  \label{eq:comm proof4}\\
 & = \max_{\lambda\in Y_{n,d}}\frac{1}{\dim\kil_{n,\lambda}}\sum_{Q\in\type_{n}([d])}\Tr \Pi^n_QP_{n,\lambda} 
 \label{eq:comm proof6} \\
 & = \max_{\lambda\in Y_{n,d}}\frac{1}{\dim\kil_{n,\lambda}}\Tr P_{n,\lambda}
   = \max_{\lambda\in Y_{n,d}}\frac{\dim(\hil_{n,\lambda}\otimes \kil_{n,\lambda})}{\dim\kil_{n,\lambda}}
   = \max_{\lambda\in Y_{n,d}}\dim\hil_{n,\lambda}  
 \le (n+1)^{d(d-1)/2}.\label{eq:comm proof7}
\end{align}
The steps above follow the same way as in the proof of Theorem \ref{thm:2-var char},
using \eqref{eq:type prob} 
to go from \eqref{eq:comm proof3} to \eqref{eq:comm proof4}, 
and \eqref{eq:comm char proof-1}
to go from \eqref{eq:comm proof4} to \eqref{eq:comm proof6}. 

Thus, $\|\tilde\gm\|_{\infty}\le\sqrt[n]{(n+1)^{d(d-1)/2}}$, and letting $n\to\infty$ gives 
$\|\tilde\gm\|_{\infty}\le 1$. This shows that $\tilde\gm\le I$, which is equivalent to \eqref{eq:comm bound}.
\end{proof}
\medskip

\noindent\textit{Proof of Theorem \ref{thm:classical maxCR} \eqref{eq:classical chat0}:}
Immediate from Theorem \ref{thm:comm char} and 
Theorem \ref{thm:sup bounds implies weak geometric bound}, using in the latter 
$\X_n:=\B(\hil^{\otimes n})\p$, $\tilde\S_n:=\S(\hil^{\otimes n})$, and 
$\tilde\X_n$ as in \eqref{eq:tilde X_n} with $\delta_n=+\infty$ for every $n\in\bN$.

\section{Application to quantum channel discrimination}
\label{sec:channel disc}

In this section, we show how the results of the previous sections can be applied to composite channel discrimination. 
We refer to Appendix \ref{sec:supop persp} for terminology and notation regarding 
completely positive (CP) maps, 
the Choi operator $C_{\Psi}(\M)$ of a CP map $\M$ in \eqref{eq:Choiop def}, 
and the weighted Kubo-Ando geometric mean $\M\#_t\N$ of two CP maps $\M,\N$, given in 
\eqref{eq:supop gmean}.

In the problem of composite i.i.d.~channel discrimination, an experimenter has to decide whether a given black 
box implements a quantum operation (or quantum channel, i.e., a completely positive trace-preserving map) 
belonging to a given set of quantum operations $\C_0$, or one belonging to another set $\C_1$.
I.i.d.~means that if the quantum operation that the channel implements is $\E$, then $n$ applications of the
black box implements $\E^{\otimes n}$. 

An $n$-copy \ki{parallel discrimination strategy} is described by a quantum state $\vfi_{RA^n}$
on $\hil_{RA^n}:=\hil_R\otimes\hil_A^{\otimes n}$ 
(where $R$ labels some reference system, and $A$ labels the input system of the channel, and 
$RA^n$ denotes the joint system of the reference and $n$ copies of the input system), and
a test $T\in\B(\hil_{RB^n})_{[0,1]}$,
where $B$ labels the output system of the 
channel. The type I error probability of erroneously deciding that the channel belongs to $\C_1$ 
and the type II error probability of erroneously deciding that the channel belongs to $\C_0$
are then given by 
\begin{align*}
\alpha_n(\vfi,T)=\sup_{\N\in\C_0}\Tr(I-T)(\id_R\otimes\N^{\otimes n})\vfi_{RA^n},
\ds\ds\ds
\beta_n(\vfi,T)=\sup_{\M\in\C_1}\Tr T(\id_R\otimes\M^{\otimes n})\vfi_{RA^n},
\end{align*}
respectively.

A more general $n$-copy \ki{adaptive strategy} is specified by an input state $\vfi_{R_1A_1}$, a sequence of channels
$\vec{\A}=(\A_{B_iR_i\to A_{i+1}R_{i+1}})_{i=1}^{n-1}$, and a 
test $T\in\B(\hil_{B_nR_n})_{[0,1]}$,
where $A_i$ denotes the input system of the $i$-th application of the 
black box and $B_i$ its output system. The 
error probabilities are then given by 
\begin{align*}
\alpha_n(\vfi,\vec{\A},T):=\sup_{\N\in\C_0}\Tr (I-T)\smap_{\vec{\A}}(\N^{\otimes n})\vfi_{R_1A_1},\ds\ds\ds
\beta_n(\vfi,\vec{\A},T):=\sup_{\M\in\C_1}\Tr T\smap_{\vec{\A}}(\M^{\otimes n})\vfi_{R_1A_1},
\end{align*}
where $\smap_{\vec{\A}}$ is a completely positive supermap that 
acts on a tensor power channel $\E^{\otimes n}$ as
\begin{align*}
\smap_{\vec{\A}}(\E^{\otimes n}):=&
\bz\id_{\B(\hil_{R_n})}\otimes\E_{A_n\to B_n}\jz
\circ
\A_{R_{n-1}B_{n-1}\to R_nA_n}
\circ
\bz\id_{\B(\hil_{R_{n-1}})}\otimes\E_{A_{n-1}\to B_{n-1}}\jz
\circ\ldots\\
&\ldots\circ\A_{R_{1}B_{1}\to R_2A_2}
\circ
\bz\id_{\B(\hil_{R_1})}\otimes\E_{A_1\to B_1}\jz.
\end{align*}
It is easy to see that any parallel strategy can be written as an adaptive strategy, and therefore adaptive strategies 
are at least as powerful as parallel ones.

The most general $n$-copy discrimination strategy is given by a 
state $\vfi_{C_n}\in\S(\hil_{C_n})$, a test $T\in\B(\hil_{D_n})_{[0,1]}$, 
and a
completely positive supermap 
$\smap:\,\B(\B(\hil_{A^{n}}),\B(\hil_{B^{n}}))\to\B(\B(\hil_{C_n}),\B(\hil_{D_n}))$
with the property that it outputs a quantum channel on every input of the form 
$\F^{\otimes n}$, where $\F$ is a quantum channel. The 
 corresponding error probabilities are
\begin{align*}
\alpha_n(\vfi,\smap,T):=\sup_{\N\in\C_0}\Tr (I-T)\smap(\N^{\otimes n})\vfi_{C_n},\ds\ds\ds
\beta_n(\vfi,\smap,T):=\sup_{\M\in\C_1}\Tr T\smap(\M^{\otimes n})\vfi_{C_n}.
\end{align*}

The direct exponents $\dli_r^{\st}(\C_0\|\C_1)$, $\dls_r^{\st}(\C_0\|\C_1)$, and the strong converse exponents
$\sci_r^{\st}(\C_0\|\C_1)$, $\scs_r^{\st}(\C_0\|\C_1)$
can be defined analogously to \eqref{eq:direct exp def1}--\eqref{eq:sc exp def2}, 
where the asymptotics of the error probabilities are optimized over all 
parallel strategies when $\st=\para$, over all adaptive strategies when $\st=\ad$,
and over general strategies when $\st=\gen$.
Just as in the case of state discrimination, we have the trivial bounds
\begin{align*}
\dls_r^{\st}(\C_0\|\C_1)
\le
\inf_{\N\in\C_0,\,\M\in\C_1}\dls_r^{\st}(\{\N\}\|\{\M\}),\ds\ds\ds
\sci_r^{\st}(\C_0\|\C_1)\ge\sup_{\N\in\C_0,\,\M\in\C_1}\sci_r^{\st}(\{\N\}\|\{\M\}).
\end{align*} 
Improved bounds of the form
\begin{align*}
\dls_r^{\st}(\C_0\|\C_1)
\le
\dls_r^{\st}(\{\tilde\E\}\|\{\E\}),\ds\ds\ds
\sci_r^{\st}(\C_0\|\C_1)
\ge
\sup_{\N\in\C_0}\sci_r^{\st}(\{\N\}\|\{\E\})
\end{align*}
can be obtained if one finds CP maps $\E,\tilde\E$ such that 
a) in the case $\st=\para$,
\begin{align*}
\sup_{\N\in\C_0}\Tr (I-T)(\id_R\otimes\N^{\otimes n})\vfi_{RA^n}
&\ge
\Tr (I-T)(\id_R\otimes\tilde\E^{\otimes n})\vfi_{RA^n},\\
\sup_{\M\in\C_1}\Tr T(\id_R\otimes\M^{\otimes n})\vfi_{RA^n}
&\ge
\Tr T(\id_R\otimes\E^{\otimes n})\vfi_{RA^n},
\end{align*} 
for any $n\in\bN$ and any $n$-copy parallel strategy $(\vfi,T)$;
b) in the case $\st=\ad$,
\begin{align*}
\sup_{\N\in\C_0}\Tr (I-T)\smap_{\vec{\A}}(\N^{\otimes n})\vfi_{R_1A_1}
&\ge
\Tr (I-T)\smap_{\vec{\A}}(\tilde\E^{\otimes n})\vfi_{R_1A_1},\\
\sup_{\M\in\C_1}\Tr T\smap_{\vec{\A}}(\M^{\otimes n})\vfi_{R_1A_1}
&\ge
\Tr T\smap_{\vec{\A}}(\E^{\otimes n})\vfi_{R_1A_1},
\end{align*}
for any $n\in\bN$ and any $n$-copy adaptive strategy $(\vfi,\vec{\A},T)$;
and c) in the case $\st=\gen$,
\begin{align*}
\sup_{\N\in\C_0}\Tr (I-T)\smap(\N^{\otimes n})\vfi_{C_n}
&\ge
\Tr (I-T)\smap(\tilde\E^{\otimes n})\vfi_{C_n},\\
\sup_{\M\in\C_1}\Tr T\smap(\M^{\otimes n})\vfi_{C_n}
&\ge
\Tr T\smap(\E^{\otimes n})\vfi_{C_n}.
\end{align*}

This leads to an analogous mathematical problem as in the case of state discrimination, and the results 
related to state discrimination can be directly applied here. For simplicity, below we restrict to a setting corresponding to the case where finite sets of channels are to be distinguished. 

\begin{thm}
Let $\Y$ be a finite set and $\N_y\in\cp(\hil_A,\hil_B)$, $y\in\Y$, 
and $\E\in\cp(\hil_A,\hil_B)$ be completely positive maps. Then the following are equivalent:
\begin{enumerate}
\item\label{ch disc0}
For every $n\in\bN$, and every $n$-copy general strategy $(\vfi,\smap,T)$,
\begin{align*}
\Tr T\smap(\E^{\otimes n})\vfi_{C_n}
&\le
\max_{y\in\Y}\Tr T\smap(\N_y^{\otimes n})\vfi_{C_n}.
\end{align*}

\item\label{ch disc1}
For every $n\in\bN$, and every $n$-copy adaptive strategy $(\vfi,\vec{\A},T)$,
\begin{align*}
\Tr T\smap_{\vec{\A}}(\E^{\otimes n})\vfi_{R_1A_1}
&\le
\max_{y\in\Y}\Tr T\smap_{\vec{\A}}(\N_y^{\otimes n})\vfi_{R_1A_1}.
\end{align*}

\item\label{ch disc2}
For every $n\in\bN$, and every $n$-copy parallel strategy $(\vfi,T)$,
\begin{align*}
\Tr T(\id_R\otimes\E^{\otimes n})\vfi_{RA^n}
&\le
\max_{y\in\Y}\Tr T(\id_R\otimes\N_y^{\otimes n})\vfi_{RA^n}.
\end{align*} 

\item\label{ch disc3}
There exists a probability measure $\mu\in\S(\Y)$ such that for every $n\in\bN$, 
\begin{align}\label{eq:ch disc3}
\E^{\otimes n}\le_{\cp}\sum_{y\in\Y}\mu(y)\N_y^{\otimes n}.
\end{align}
\setcounter{szamlalo}{\value{enumi}}
\end{enumerate}
If, moreover, $\Y$ is a $2$-element set (w.l.o.g.~$\Y=\{1,2\}$), then the above are further equivalent to the following:
\begin{enumerate}
\setcounter{enumi}{\value{szamlalo}}
\item\label{ch disc4}
There exists a $t\in[0,1]$ such that 
\begin{align}\label{eq:ch disc4}
\E\le_{\cp}\N_2\#_t\N_1.
\end{align}
\end{enumerate} 
\end{thm}
\begin{proof}
The implications \ref{ch disc0}$\imp$\ref{ch disc1}$\imp$\ref{ch disc2} are trivial,
and the implication \ref{ch disc3}$\imp$\ref{ch disc0} is trivial from the 
CP monotonicity of $\smap$. 

To see \ref{ch disc2}$\imp$\ref{ch disc3}, let 
$\Psi=\sum_{i=1}^{d_A}\tilde e_i\otimes e_i\in\hil_{\tilde A}\otimes\hil_A$ be a Choi vector,
where  
$(e_i)_{i=1}^{d_A}$ is an orthonormal basis in $\hil_A$ and
$(\tilde e_i)_{i=1}^{d_A}$ is an orthonormal basis in $\hil_{\tilde A}$, and consider
$n$-copy parallel strategies of the type
$R=\tilde A^n$, $\vfi_{RA^n}=(\dim\hil_{A_n})\inv\pr{\Psi^{\otimes n}}$, and $T\in\B(\hil_{\tilde A^nB^n})_{[0,1]}$ arbitrary. Then 
\ref{ch disc2} implies that
\begin{align*}
\Tr T\choi_{\Psi}(\E)^{\otimes n}
\le
\max_{y\in\Y}\Tr T\choi_{\Psi}(\N_y)^{\otimes n},\ds\ds\ds
T\in\B(\hil_{\tilde A^nB^n})_{[0,1]},\ds n\in\bN,
\end{align*} 
where we used the tensor multiplicativity \eqref{eq:Choiop mult2} of the Choi operator.
By Theorem \ref{thm:sup bounds implies weak geometric bound}, this implies the existence of a probability measure 
$\mu\in\S(\Y)$ such that 
\begin{align*}
\Tr T\choi_{\Psi}(\E)^{\otimes n}
&\le
\sum_{y\in\Y}\mu(y)\Tr T\choi_{\Psi}(\N_y)^{\otimes n},
\ds\ds\ds T\in\B(\hil_{\tilde A^nB^n})_{[0,1]},\ds n\in\bN.
\end{align*}
Thus, 
\begin{align}
\choi_{\Psi}(\E)^{\otimes n}\le\sum_{y\in\Y}\mu(y)\choi_{\Psi}(\N_y)^{\otimes n}
=\choi_{\Psi^{\otimes n}}\bz\sum_{y\in\Y}\mu(y)\N_y^{\otimes n}\jz,\ds\ds\ds n\in\bN,
\label{eq:ch disc proof1}
\end{align}
which in turn is equivalent to \eqref{eq:ch disc3}.

Finally, assume that $\Y=\{1,2\}$. As we have seen above, \ref{ch disc0}--\ref{ch disc3}
are equivalent to  
\eqref{eq:ch disc proof1}, 
which in turn is equivalent to the existence of some $t\in[0,1]$ such that 
$\choi_{\Psi}(\E)\le\choi_{\Psi}(\N_2)\#_t\choi_{\Psi}(\N_1)$, 
according to Theorem \ref{thm:2-var char2}, 
and this latter inequality is equivalent to 
\eqref{eq:ch disc4} by definition.
\end{proof}

\section{Conclusion}

We have developed further an approach first put forward in \cite{MWSz2020} to obtaining tighter than previously known 
single-copy bounds on the error exponents of binary composite i.i.d.~state discrimination by 
comparing not only individual states from the two sets representing the two hypotheses, but also various 
subnormalized PSD operators associated to each set. We have shown a number of equivalent characterizations
of such operators to be useful for this approach; in particular, we have shown that weak sup-type, 
arithmetic mean-type, and weak geometric mean-type inequalities all become equivalent when required to hold 
for an arbitrary number of copies. Most importantly, we have obtained 
explicit single-copy characterizations of the set of such operators
in two important special cases: when the set representing a hypothesis 
is compact and it contains only commuting density 
operators, and when it contains two density operators on a finite-dimensional Hilbert space. In the first case, it is exactly the set of operators dominated by some weighted geometric mean of the states, which are uniquely defined for any probability measure on the (Borel $\sigma$-algebra of the) set of states. 
In the second case, 
it is exactly the set of operators dominated by a $t$-weighted Kubo-Ando geometric mean for some $t\in[0,1]$. The resulting error bounds in the second case are
\begin{align}
\dli_r(\{\rho_1,\rho_2\}\|\{\sigma_1,\sigma_2\})&\le
\inf_{s,t\in[0,1]}
\sup_{\alpha\in(0,1)}\frac{\alpha-1}{\alpha}
\left[r-D_{\alpha}(\rho_2\#_s\rho_1\|\sigma_2\#_t\sigma_1)\right],\label{geometric bounds1}\\
\sci_r(\{\rho_1,\rho_2\}\|\{\sigma_1,\sigma_2\})
&\ge
\max_{i=1,2}\sup_{t\in[0,1]}
\sup_{\alpha\ge 1}\frac{\alpha-1}{\alpha}\left[r-D\nw_{\alpha}(\rho_i\|\sigma_2\#_t\sigma_1)\right].
\label{geometric bounds2}
\end{align}
It was shown in \cite{bunth2021equivariant} that \eqref{geometric bounds2} holds as an equality when the 
$\sigma$ operators commute and are both invertible, and probably the most important open question related to our results is 
to characterize when
equality holds in \eqref{geometric bounds2} or \eqref{geometric bounds1} in the non-commutative case. 

For more than two operators, there are many different notions of weighted geometric mean
for any given non-degenerate probability measure on the set of density operators (more generally, PSD operators), that 
all satisfy the desirable properties of tensor multiplicativity and the AM-GM inequality, and therefore 
provide valid error bounds in the above approach. This also means that 
any possible characterization extending the $2$-variable case above 
would need to take this plethora of options into account. Here we make the (admittedly bold) proposal to 
define multi-variate matrix geometric means as the maximal operators satisfying the equivalent conditions
e.g., in Theorem \ref{thm:sup bounds implies weak geometric bound}.
Whether a complete description of these maximal elements is possible in a more or less explicit form is an open question already for three operators.
The natural question to ask then from the point of view of state discrimination, is whether
these maximal elements give an exact single-copy expression for the error exponents, as was shown to be the case 
for the strong converse exponent for a commutative alternative hypothesis in \cite{bunth2021equivariant}.

\section*{Acknowledgments}

This work was partially funded by the
National Research, Development and 
Innovation Office of Hungary (NKFIH) via the research grants K 146380 and EXCELLENCE 151342, and
by the Ministry of Culture and Innovation and the National Research, Development and Innovation Office within the Quantum Information National Laboratory of Hungary (Grant No. 2022-2.1.1-NL-2022-00004).
The work of P\'eter Vrana was further supported by a Bolyai fellowship of the Hungarian Academy of Sciences, and 
by the NKFIH FK 146643 grant.
The work of PEF was partially supported by the NKFIH  KKP 139502 grant.
MM was supported by the Ministry of Education, Singapore, through grant T2EP20124-0005.
MM is grateful to Fumio Hiai for discussions on matrix geometric means
and to Frits Verhagen for discussions on Jordan's lemma.
The authors are grateful to an anonymous referee for their detailed comments that helped to greatly improve the presentation of the paper. ChatGPT-6 Astra was  used to fix various minor issues in the 
revised version of the journal submission.

\appendix

\section{Proof of \eqref{direct upper}--\eqref{sc lower}}
\label{sec:improved bounds}

Let $\tilde C\in\C(\N)\setminus\{0\}$ and $C\in\C(\A)\setminus\{0\}$, so that 
for any test sequence $T_n\in\B(\hil^{\otimes n})_{[0,1]}$, $n\in\bN$, we have
\begin{align*}
\alpha_n(T_n)&=\sup_{\rho\in\nh}\Tr\rho^{\otimes n}(I-T_n)
\ge
\Tr \tilde C^{\otimes n}(I-T_n),\\
\beta_n(T_n)&=\sup_{\sigma\in\ah}\Tr\sigma^{\otimes n}T_n
\ge
\Tr C^{\otimes n}T_n,\nn
\end{align*}
where we used the alternative characterization of $\C(\R)$ from \eqref{eq:CR def2}.
Hence,
\begin{align}
\limsup_{n\to+\infty}-\frac{1}{n}\log\alpha_n(T_n)&\le
\limsup_{n\to+\infty}-\frac{1}{n}\log\Tr
\bz\frac{\tilde  C}{\Tr \tilde C}\jz^{\otimes n}(I-T_n)-\log\Tr \tilde C,\nn\\
\liminf_{n\to+\infty}-\frac{1}{n}\log\beta_n(T_n)&\le
\liminf_{n\to+\infty}-\frac{1}{n}
\log\Tr\bz\frac{C}{\Tr C}\jz^{\otimes n}T_n-\log\Tr C.
\label{eq:type II bound}
\end{align}
Thus, if $(\tilde r,r)$ is a weakly achievable direct rate pair for the 
$\N$ vs.~$\A$ composite state discrimination problem in the sense that for some test sequence $(T_n)_{n\in\bN}$, 
\begin{align*}
\tilde r\le\limsup_{n\to+\infty}-\frac{1}{n}\log\alpha_n(T_n),\ds\ds\ds
r<\liminf_{n\to+\infty}-\frac{1}{n}\log\beta_n(T_n),
\end{align*}
 then 
$(\tilde r+\log\Tr \tilde C,r+\log\Tr  C)$ is a weakly achievable rate pair for the simple i.i.d. state discrimination problem 
$\frac{\tilde C}{\Tr\tilde C}$ vs. $\frac{C}{\Tr C}$. Hence,
\begin{align*}
\tilde r+\log\Tr \tilde C
&\le 
\dls_{r+\log\Tr C}
\bz\frac{\tilde C}{\Tr\tilde C}\Bigg\|\frac{C}{\Tr C}\jz\\
&=
\sup_{\alpha\in(0,1)}\frac{\alpha-1}{\alpha}\left[r+\log\Tr C
-D_{\alpha}\bz\frac{\tilde C}{\Tr\tilde C}
\Bigg\|\frac{C}{\Tr C}\jz\right],
\end{align*}
where the equality follows from \eqref{eq:single-copy direct}.
Taking the supremum in $\tilde r$ then yields
\begin{align*}
\dls_r(\N\|\A)+\log\Tr \tilde C\le
\sup_{\alpha\in(0,1)}\frac{\alpha-1}{\alpha}\left[r+\log\Tr C
-D_{\alpha}\bz\frac{\tilde C}{\Tr\tilde C}\Bigg\|\frac{C}{\Tr C}\jz\right]
\end{align*}
which is equivalent to \eqref{direct upper}, due to the scaling property 
$D_{\alpha}(\lambda A\|\eta B)=D_{\alpha}(A\|B)+\alpha/(\alpha-1)\log\lambda-\log\eta$, 
$\lambda,\eta>0$, $A,B\in\B(\hil)\pne$.

Similarly, for any $\rho\in\N$, we have 
\begin{align*}
1-\alpha_n(T_n)=\inf_{\omega\in\N}\Tr\omega^{\otimes n}T_n\le\Tr\rho^{\otimes n}T_n, 
\end{align*}
whence
\begin{align*}
\liminf_n-\frac{1}{n}\log(1-\alpha_n(T_n))\ge
\liminf_n-\frac{1}{n}\log\Tr\rho^{\otimes n}T_n.
\end{align*}
Combining this with \eqref{eq:type II bound} yields that if 
$(\tilde r,r)$ is a weakly achievable rate pair in the converse domain for the composite state discrimination problem $\N$ vs.~$\A$
in the sense that for some test sequence $(T_n)_{n\in\bN}$, 
\begin{align*}
\tilde r\ge\liminf_{n\to+\infty}-\frac{1}{n}\log(1-\alpha_n(T_n)),\ds\ds\ds
r<\liminf_{n\to+\infty}-\frac{1}{n}\log\beta_n(T_n),
\end{align*}
then 
$(\tilde r,r+\log\Tr C)$ is a weakly achievable rate pair in the converse domain for 
the i.i.d.~state discrimination problem
$\rho$ vs.~$\frac{C}{\Tr C}$. Thus,
\begin{align*}
\tilde r\ge
\sci_{r+\log\Tr C}
\bz\rho\Big\|\frac{C}{\Tr C}\jz 
&=
\sup_{\alpha\ge 1}\frac{\alpha-1}{\alpha}\left[r+\log\Tr C
-D_{\alpha}\nw\bz\rho\Big\|\frac{C}{\Tr C}\jz\right]\\
&=
\sup_{\alpha\ge 1}\frac{\alpha-1}{\alpha}\left[r-D_{\alpha}\nw\bz\rho\|C\jz\right].
\end{align*}
Taking the infimum over $\tilde r$ yields
\begin{align*}
\sci_r(\N\|\A)
 \ge
\sup_{\alpha\ge 1}\frac{\alpha-1}{\alpha}\left[r-D_{\alpha}\nw\bz\rho\|C\jz\right].
\end{align*}
Taking then the supremum over $\rho\in\N$ yields \eqref{sc lower}.

The reason why the above argument does not yield a supremum also over $\tilde C\in\C(\N)$ in 
\eqref{sc lower}, analogously to \eqref{direct upper}, is that 
$\tilde C\in\C(\N)$ does not give a bound on  
$1-\alpha_n(T_n)=\inf_{\rho\in\N}\Tr\rho^{\otimes n}T_n$, since here we have an infimum over
$\N$ instead of a supremum appearing in the definition of $\C(\N)$.
On top of this, the strong converse exponent is expected to be well-behaved in the null hypothesis in the sense
that $\sci_r(\N\|\A)=\sup_{\rho\in\N}\sci_r(\{\rho\}\|\A)$ for any alternative hypothesis; see, e.g., \eqref{eq:sc lower comm} for the case when $\A$ is commutative, and \cite[Section V.A]{MWSz2020} for the case when $\N$ is finite.

\section{Examples for (non-)improvement via geometric bounds}
\label{sec:examples}

For the rest of this section, let $\hil$ be a finite-dimensional Hilbert space, and 
$\N,\A\subseteq\S(\hil)$ be the null- and the alternative hypothesis, respectively,   
of a composite i.i.d.~state discrimination problem.
Throughout this appendix, we assume that $\N$ and $\A$ are closed (equivalently, compact).
For any $\R\subseteq\S(\hil)$, let $\C(\R)$ be as in \eqref{eq:CR def2}, and 
recall the obvious inclusion 
\begin{align}\label{eq:trivial C2}
\R\subseteq\C(\R),
\end{align}
By Proposition \ref{prop:sup bound implies AM bound}, 
if $\R$ is compact then for any $\gm\in\C(\R)$, there exists a probability measure $\mu_{\gm}\in\S(\R)$ such that 
\begin{align}\label{eq:AM bound2}
C\le\int_{\R}\omega\,d\mu_{\gm}(\omega)=:\Am_{\mu_{\gm}}(\R).
\end{align} 
For a set $\R\subseteq\S(\hil)$, let $\co(\R)$ denote its convex hull.
Then, we have 
\begin{align*}
\inf_{\rho\in\N}\inf_{\sigma\in\A}\sup_{\alpha\in(0,1)}
\frac{\alpha-1}{\alpha}[r-D_{\alpha}(\rho\|\sigma)]
&\ge
\inf_{\tilde C\in\C(\N)}\inf_{C\in\C(\A)}\sup_{\alpha\in(0,1)}
\frac{\alpha-1}{\alpha}[r-D_{\alpha}(\tilde C\|C)]\\
&\ge
\inf_{\tilde\mu\in\S(\N)}\inf_{\mu\in\S(\A)}\sup_{\alpha\in(0,1)}
\frac{\alpha-1}{\alpha}[r-D_{\alpha}(\Am_{\tilde\mu}(\N)\|\Am_{\mu}(\A))]\\
&=
\inf_{\rho\in\co(\N)}\inf_{\sigma\in\co(\A)}\sup_{\alpha\in(0,1)}
\frac{\alpha-1}{\alpha}[r-D_{\alpha}(\rho\|\sigma)],
\end{align*}
where the first inequality follows from \eqref{eq:trivial C2}, the second inequality from \eqref{eq:AM bound2}
and \eqref{eq:Renyi anti-mon}, and the equality is trivial.
In particular, if $\N$ and $\A$ are both convex, then all the inequalities above are equalities,
and therefore the bound in \eqref{direct upper} does not give an improvement over the trivial bound in
\eqref{eq:trivial bounds}. 
Similarly,
\begin{align*}
\sup_{\rho\in\N}\sup_{\sigma\in\A}\sup_{\alpha\ge 1}\frac{\alpha-1}{\alpha}
[r-D\nw_{\alpha}(\rho\|\sigma)]
&\le
\sup_{\rho\in\N}\sup_{C\in\C(\A)}\sup_{\alpha\ge 1}\frac{\alpha-1}{\alpha}
[r-D\nw_{\alpha}(\rho\|C)]\\
&\le
\sup_{\rho\in\N}\sup_{\mu\in\S(\A)}\sup_{\alpha\ge 1}\frac{\alpha-1}{\alpha}
[r-D\nw_{\alpha}(\rho\|\Am_{\mu}(\A))]\\
&\le
\sup_{\rho\in\N}\sup_{\sigma\in\co(\A)}\sup_{\alpha\ge 1}\frac{\alpha-1}{\alpha}
[r-D\nw_{\alpha}(\rho\|\sigma)].
\end{align*}
In particular, if $\A$ is convex, 
all the above expressions are equal, and
the bound in \eqref{sc lower} does not give an improvement over the trivial bound in
\eqref{eq:trivial bounds}. 

Examples with a simple null hypothesis $\N=\{\rho\}$ and a composite alternative hypothesis 
$\A=\{\sigma_1,\sigma_2\}$ where the geometric bound
\eqref{direct upper} gives a strict improvement over the trivial bound in
\eqref{eq:trivial bounds} were given in \cite{MWSz2020}
(for this, the states $\sigma_1$ and $\sigma_2$ necessarily have to be non-commuting; see
\cite[Section III.A]{MWSz2020}.)
 Below 
we show an example with $\N=\{\rho\}$ and $\A=\{\sigma_1,\sigma_2\}$
such that $\rho,\sigma_1,\sigma_2$ are commuting, and 
\begin{align}
\sup_{\rho\in\N}\sup_{\sigma\in\A}\sup_{\alpha\ge 1}\frac{\alpha-1}{\alpha}
[r-D\nw_{\alpha}(\rho\|\sigma)]
&<
\sup_{\rho\in\N}\sup_{C\in\C(\A)}\sup_{\alpha\ge 1}\frac{\alpha-1}{\alpha}
[r-D\nw_{\alpha}(\rho\|C)]\\
&=
\sup_{\rho\in\N}\sup_{\nu\in\S(\A)}
\sup_{\alpha\ge 1}\frac{\alpha-1}{\alpha}
[r-D\nw_{\alpha}(\rho\|\Gm_{\nu}(\A))]\label{eq:strict1}\\
&=\sconv_r(\N\|\A)\\
&<
\sup_{\rho\in\N}\sup_{\sigma\in\co(\A)}\sup_{\alpha\ge 1}\frac{\alpha-1}{\alpha}
[r-D\nw_{\alpha}(\rho\|\sigma)].\label{eq:strict2}
\end{align}
(The first equality follows from Theorem \ref{thm:2-var char2} and \eqref{eq:Renyi anti-mon}
by taking into account that $G_{\nu}(\ah)\in\C(\ah)$
for any $\nu$, and the second one from 
\cite{bunth2021equivariant,BunthPhD}.)
That is, the geometric bound \eqref{sc lower} gives a strict improvement over the trivial bound \eqref{eq:trivial bounds},
while extending the optimization to the convex hull of $\A$ overshoots the strong converse exponent. 

To demonstrate such an example, we will need some simple facts about R\'enyi divergences and the Hoeffding 
anti-divergences for commuting operators, which we will always naturally identify with functions on a finite set.
For such operators, the Petz-type and the sandwiched R\'enyi divergences coincide, and we will denote both by 
$D_{\alpha}$.
For a finite set $\X$ and for non-zero non-negative functions $\rho,\sigma\in[0,+\infty)^{\X}$
such that $\supp\rho\subseteq\supp\sigma$, let 
\begin{align}\label{eq:Dmax def}
D_{\infty}(\rho\|\sigma):=\log\max_{x\in\supp(\sigma)}\frac{\rho(x)}{\sigma(x)}
\end{align}
be their \ki{max-relative entropy} \cite{Datta,RennerPhD}, and 
\begin{align}\label{eq:r infty}
\X_{\infty}:=\left\{x\in\supp(\sigma):\,\log\frac{\rho(x)}{\sigma(x)}=D_{\infty}(\rho\|\sigma)\right\},\ds\ds\ds
r_{\infty}:=-\log\sum_{x\in\X_{\infty}}\sigma(x).
\end{align}

\begin{lemma}\label{lemma:classical Renyi}
Let $\X$ be a finite set and $\rho,\sigma$ be non-zero non-negative functions 
on $\X$ such that  $\supp\rho\subseteq\supp\sigma$ and $\sum_{x\in\X}\rho(x)=1$.
Assume that there exists no $\kappa\in(0,+\infty)$ such that 
$\rho(x)=\kappa\sigma(x)$, $x\in\supp\rho$. Then the following hold:
\begin{enumerate}
\item\label{classical Renyi1}
$(0,+\infty)\ni\alpha\mapsto\psi(\alpha):= \log\Tr\rho^{\alpha}\sigma^{1-\alpha}$ is strictly convex, and
$(0,+\infty)\ni\alpha\mapsto D_{\alpha}(\rho\|\sigma)$ is strictly increasing, with 
\begin{align*}
D(\rho\|\sigma)=D_1(\rho\|\sigma):=\lim_{\alpha\to 1}D_{\alpha}(\rho\|\sigma)
<D_{\alpha}(\rho\|\sigma)<D_{\infty}(\rho\|\sigma)=\lim_{\alpha\to+\infty}D_{\alpha}(\rho\|\sigma),
\ds\ds\ds\alpha\in(1,+\infty).
\end{align*}

\item\label{classical Renyi2}
For any $r\in(D(\rho\|\sigma),+\infty)$, 
\begin{align}
r-D_{\infty}(\rho\|\sigma)\le H_r^*(\rho\|\sigma)
&:=\sup_{\alpha>1}\frac{\alpha-1}{\alpha}[r-D_{\alpha}(\rho\|\sigma)]\label{eq:Hoeffding bounds}\\
&\le\max\left\{\half(r-D(\rho\|\sigma)),r-D_2(\rho\|\sigma)\right\}
<r-D(\rho\|\sigma),\label{eq:Hoeffding bounds2}
\end{align}
and the first inequality is also strict whenever $r<r_{\infty}$.
\end{enumerate}
\end{lemma}
\begin{proof}
The claims in \ref{classical Renyi1} follow by straightforward computations; see, e.g., \cite[Section II.B]{MWSz2020}.
The inequality in \eqref{eq:Hoeffding bounds} is obvious from 
$\lim_{\alpha\to+\infty}\frac{\alpha-1}{\alpha}[r-D_{\alpha}(\rho\|\sigma)]=r-D_{\infty}(\rho\|\sigma)$.
The inequality in \eqref{eq:Hoeffding bounds2} follows as
\begin{align*}
\frac{\alpha-1}{\alpha}[r-D_{\alpha}(\rho\|\sigma)]<
\begin{cases}
\half(r-D(\rho\|\sigma)),&\alpha\in(1,2],\\
\half(r-D(\rho\|\sigma)),&\alpha>2,\,r<D_2(\rho\|\sigma),\\
\frac{\alpha-1}{\alpha}[r-D_2(\rho\|\sigma)]\le r-D_2(\rho\|\sigma),&\alpha>2,\,r\ge D_2(\rho\|\sigma).
\end{cases}
\end{align*}
Note that 
\begin{align}\label{eq:Hoeffding Legendre}
H_r^*(\rho\|\sigma)=\sup_{u\in(0,1)}\left[ur-\tilde\psi(u)\right],\ds\ds\text{where}\ds\ds
\tilde\psi(u):=(1-u)\psi(1/(1-u)).
\end{align}
It is easy to see that $\tilde\psi$ is also strictly convex, $\tilde\psi'(0)=D(\rho\|\sigma)$, 
$\lim_{u\nearrow 1}\tilde\psi'(u)=r_{\infty}$, whence for any 
$r\in(D(\rho\|\sigma),r_{\infty})$, the supremum in \eqref{eq:Hoeffding Legendre}
is uniquely attained at some $u_r\in(0,1)$, whereby
$H_r^*(\rho\|\sigma)>\lim_{u\to 1}[ur-\tilde\psi(u)]=r-D_{\infty}(\rho\|\sigma)$, showing the strictness of the first 
inequality in \eqref{eq:Hoeffding bounds}.
We refer again to \cite[Section II.B]{MWSz2020} for details.
\end{proof}

\begin{prop}
Let $\rho=(1/2,1/2)$, $\sigma_1=(1/4,3/4)$, $\sigma_2=(3/4,1/4)$ be probability density functions on 
$\{0,1\}$ (identified with commuting states on a $2$-dimensional Hilbert space).
 Then for any 
$k\in\bN$, and any $r>D(\rho^{\otimes k}\|\sigma_j^{\otimes k})=k\log(2/\sqrt{3})$,
\begin{align}
&\max_{j=1,2}\sup_{\alpha>1}\frac{\alpha-1}{\alpha}
[r-D_{\alpha}(\rho^{\otimes k}\|\sigma_j^{\otimes k})]\label{eq:strict ineq ex1}\\
&\ds<
r-D(\rho^{\otimes k}\|\sigma_j^{\otimes k})=
r-k\log(2/\sqrt{3})\label{eq:strict ineq ex2}\\
&\ds=
\sup_{t\in[0,1]}
\sup_{\alpha>1}\frac{\alpha-1}{\alpha}
\left[r-D_{\alpha}\bz\rho^{\otimes k}\|(\sigma_1^{\otimes k})^t(\sigma_2^{\otimes k})^{1-t}\jz\right]
=\sconv_r(\{\rho^{\otimes k}\}\|\{\sigma_1^{\otimes k},\sigma_2^{\otimes k}\})
\label{eq:strict ineq ex3}\\
&\ds\le
\sup_{t\in[0,1]}\sup_{\alpha>1}\frac{\alpha-1}{\alpha}
[r-D_{\alpha}(\rho^{\otimes k}\|t\sigma_1^{\otimes k}+(1-t)\sigma_2^{\otimes k})],
\label{eq:strict ineq ex4}
\end{align}
and the last inequality is strict whenever
$k$ is odd, or $k$ is even and $r<k\log\frac{4}{\sqrt{3}}-\log\binom{k}{k/2}$.
\end{prop}
\begin{proof}
Let $k\in\bN$ be fixed. 
The equality $D(\rho^{\otimes k}\|\sigma_j^{\otimes k})=k\log(2/\sqrt{3})$, $j=1,2$, follows by a straightforward computation. Assume for the rest that $r>k\log(2/\sqrt{3})$.
Then, the inequality in 
\eqref{eq:strict ineq ex2} follows immediately from Lemma \ref{lemma:classical Renyi}.
%

The second equality in \eqref{eq:strict ineq ex3} is due to \cite{bunth2021equivariant,BunthPhD}, and the first equality can be seen as follows. The function $[0,1]\ni t\mapsto D_{\alpha}\bz\rho^{\otimes k}\|(\sigma_1^{\otimes k})^t(\sigma_2^{\otimes k})^{1-t}\jz$ is symmetric under $t\leftrightarrow 1-t$, and a simple H\"older inequality shows that it is convex in $t$, whence its minimum is attained at $t=1/2$. A straightforward computation shows that 
\begin{align*}
D_{\alpha}\bz\rho^{\otimes k}\big\|(\sigma_1^{\otimes k})^{1/2}(\sigma_2^{\otimes k})^{1/2}\jz
=
\frac{1}{\alpha-1}\log\sum_{\vecc{x}\in\{0,1\}^k}\bz\frac{1}{2^k}\jz^{\alpha}\bz\frac{\sqrt{3}}{4}\jz^{k(1-\alpha)}
=k\log\frac{2}{\sqrt{3}}\,.
\end{align*}
Since this is independent of $\alpha$, the supremum over $\alpha$ in \eqref{eq:strict ineq ex3} is attained 
at $\alpha=+\infty$, proving the first equality in \eqref{eq:strict ineq ex3}.

Finally, consider the expression in \eqref{eq:strict ineq ex4}, and let
\begin{align*}
Q_{\alpha}(t)&:=\exp\bz(\alpha-1)D_{\alpha}(\rho^{\otimes k}\|t\sigma_1^{\otimes k}+(1-t)\sigma_2^{\otimes k})\jz\\
&=
\sum_{\vecc{x}\in\{0,1\}^k}\bz\frac{1}{2^k}\jz^{\alpha}
\left[
(1-t)\bz\frac{1}{4}\jz^{m(\vecc{x})}\bz\frac{3}{4}\jz^{k-m(\vecc{x})}
+t\bz\frac{3}{4}\jz^{m(\vecc{x})}\bz\frac{1}{4}\jz^{k-m(\vecc{x})}
\right]^{1-\alpha}\\
&
=
\bz\frac{1}{2^k}\jz^{\alpha}\bz\frac{1}{4^k}\jz^{1-\alpha}\sum_{m=0}^k\binom{k}{m}
\left[(1-t)3^{k-m}+t 3^m\right]^{1-\alpha},
\end{align*}
where $m(\vecc{x})$ is the number of $0$s in the sequence $\vecc{x}$.
This is again a convex function of $t$ that is symmetric under $t\leftrightarrow 1-t$, whence the 
supremum over $t$ in \eqref{eq:strict ineq ex4} is attained at $t=1/2$.
Replace $\rho$ with $\hat\rho^{(k)}:=\rho^{\otimes k}$ and $\sigma$ with 
$\hat\sigma^{(k)}:=(\sigma_1^{\otimes k}+\sigma_2^{\otimes k})/2$ in \eqref{eq:Dmax def}--\eqref{eq:r infty}; then 
\begin{align*}
D_{\infty}(\hat\rho^{(k)}\|\hat\sigma^{(k)})&=\log\max_{\vecc{x}\in\{0,1\}^k}\frac{\frac{1}{2^k}}{
\half\bz\frac{1}{4}\jz^{m(\vecc{x})}\bz\frac{3}{4}\jz^{k-m(\vecc{x})}
+\half\bz\frac{3}{4}\jz^{m(\vecc{x})}\bz\frac{1}{4}\jz^{k-m(\vecc{x})}}
=\begin{cases}
k\log\frac{2}{\sqrt{3}},& k\text{ even},\\
(k-1)\log\frac{2}{\sqrt{3}},& k\text{ odd},
\end{cases}\\
r_{\infty}&=-\log\sum_{\vecc{x}\in\{0,1\}^k:\,\frac{k-1}{2}\le m(\vecc{x})\le\frac{k+1}{2}}\left[
\half\bz\frac{1}{4}\jz^{m(\vecc{x})}\bz\frac{3}{4}\jz^{k-m(\vecc{x})}
+\half\bz\frac{3}{4}\jz^{m(\vecc{x})}\bz\frac{1}{4}\jz^{k-m(\vecc{x})}
\right]\\
&=
\begin{cases}
k\log\frac{4}{\sqrt{3}}-\log\binom{k}{k/2},& k\text{ even},\\
(k-1)\log\frac{4}{\sqrt{3}}-\log\binom{k}{(k-1)/2},& k\text{ odd}.
\end{cases}
\end{align*}
Note that $H_r\nw(\hat\rho^{(k)}\|\hat\sigma^{(k)})\ge 
\lim_{\alpha\to+\infty}\frac{\alpha-1}{\alpha}[r-D_{\alpha}(\hat\rho^{(k)}\|\hat\sigma^{(k)})]=
r-D_{\infty}(\hat\rho^{(k)}\|\hat\sigma^{(k)})$. 
When $k$ is odd, 
$r-D_{\infty}(\hat\rho^{(k)}\|\hat\sigma^{(k)})=r-(k-1)\log\frac{2}{\sqrt{3}}
>r-k\log\frac{2}{\sqrt{3}}$, proving the strict inequality in \eqref{eq:strict ineq ex4}.
When $k$ is even, $r-D_{\infty}(\hat\rho^{(k)}\|\hat\sigma^{(k)})=r-k\log\frac{2}{\sqrt{3}}$, proving the inequality in 
\eqref{eq:strict ineq ex4}. If, moreover, $r<r_{\infty}$, Lemma \ref{lemma:classical Renyi} gives the strict inequality
$H_r\nw(\hat\rho^{(k)}\|\hat\sigma^{(k)})> r-k\log\frac{2}{\sqrt{3}}$, showing that also the inequality in 
\eqref{eq:strict ineq ex4} is strict.
\end{proof}

\begin{rem}
The $k=1$ case in the above proposition is special in the sense that $\rho\in\co(\{\sigma_1,\sigma_2\})$, while for every 
$k>1$, $\rho^{\otimes k}\notin\co(\{\sigma_1^{\otimes k},\sigma_2^{\otimes k}\})$.
\end{rem}

\section{More on the operator perspective} 
\label{sec:oppersp}

For the rest of this section, $f:\,(0,+\infty)\to\bR$ will denote an operator convex function, 
$\tilde f(x):=xf(1/x)$, $x\in(0,+\infty)$, denotes its transpose function, and we will use the notations
\begin{align*}
f(0^+):=\lim_{x\searrow 0}f(x),\ds\ds\ds
\tilde f(0^+):=\lim_{x\searrow 0}\tilde f(x)=\lim_{x\to+\infty}\frac{f(x)}{x}\,.
\end{align*}
Recall the definition of the operator perspective function $\persp{f}$ corresponding to $f$ given in  
\eqref{eq:oppersp def1}--\eqref{eq:oppersp def2}: for any $A,B\in\B(\hil)\p$, where $\hil$ is a finite-dimensional Hilbert space,
\begin{align*}
\persp{f}(A,B):=\lim_{\ep\searrow 0}(B+\ep I)^{1/2}f\bz (B+\ep I)^{-1/2}(A+\ep I)(B+\ep I)^{-1/2}\jz (B+\ep I)^{1/2},
\end{align*}
whenever the limit exists.

The following properties are straightforward to verify from the definition:
\begin{itemize}
\item
\ki{Direct sum property:} If for every $i\in[r]$, $A_i,B_i\in\B(\hil_i)\p$ are such that $\persp{f}(A_i,B_i)$
is defined then also $\persp{f}\bz\oplus_iA_i,\oplus_iB_i\jz$ is defined, and 
\begin{align}\label{eq:dsum1}
\persp{f}\bz\oplus_{i=1}^rA_i,\oplus_{i=1}^rB_i\jz=\oplus_{i=1}^r\persp{f}\bz A_i,B_i\jz.
\end{align}
Equivalently, if for every $i\in[r]$, $A_i,B_i\in\B(\hil)\p$ are such that $\persp{f}(A_i,B_i)$ is defined, and 
$A_i^0\vee B_i^0\perp_{i\ne j}A_j^0\vee B_j^0$ then 
$\persp{f}(\sum_iA_i,\sum_iB_i)$ is defined, and 
\begin{align}\label{eq:dsum2}
\persp{f}\bz\sum_{i=1}^rA_i,\sum_{i=1}^rB_i\jz
=
\sum_{i=1}^r\persp{f}\bz A_i,B_i\jz.
\end{align}

\item
\ki{Positive scalar homogeneity:} If $\persp{f}(A,B)$ is defined then so is 
$\persp{f}(\lambda A,\lambda B)$ for any $\lambda\in(0,+\infty)$, and 
\begin{align}\label{eq:oppersp scalar hom}
\persp{f}(\lambda A,\lambda B)=\lambda \persp{f}(A,B).
\end{align}

\item
\ki{Tensor homogeneity with a projection:}
If $A,B\in\B(\hil)\p$ are such that $\persp{f}(A,B)$ is defined then for any projection
$P$ on some other finite-dimensional Hilbert space $\kil$, 
$\persp{f}(A\otimes P,B\otimes P)$ is defined, and 
\begin{align}\label{eq:oppersp pr tensor hom}
\persp{f}(A\otimes P,B\otimes P)=\persp{f}(A,B)\otimes P.
\end{align}

\item
\ki{Tensor homogeneity:} 
If $A,B\in\B(\hil)\p$ are such that $\persp{f}(A,B)$ is defined then for any PSD operator
$C\in\B(\kil)\p$ on some other finite-dimensional Hilbert space $\kil$, 
$\persp{f}(A\otimes C,B\otimes C)$ is defined, and 
\begin{align}\label{eq:oppersp tensor hom}
\persp{f}(A\otimes C,B\otimes C)=\persp{f}(A,B)\otimes C.
\end{align}
This follows immediately using any eigen-decomposition $C=\sum_{i=1}c_i\pr{e_i}$, 
the direct sum property \eqref{eq:dsum2}, positive scalar homogeneity \eqref{eq:oppersp scalar hom},
and tensor homogeneity with a projection \eqref{eq:oppersp pr tensor hom}.
\end{itemize}

The following was given in \cite[Proposition 3.30]{HiaiMosonyi2017} (see also \cite[Proposition 2.5]{HP} and 
\cite[Lemma 3]{Matsumoto_newfdiv}). We include the short proof for completeness.

\begin{lemma}\label{lemma:oppersp mon}
Let $A,B\in\B(\hil)\pp$ be positive definite, and $\E:\,\B(\hil)\to\B(\kil)$ be positive. Then 
\begin{align}\label{eq:opersp posmon}
\persp{f}(\E(A),\E(B))\le\E(\persp{f}(A,B)).
\end{align}
\end{lemma}
\begin{proof}
The map $\E_B(\valt):=\E(B)^{-1/2}\E\bz B^{1/2}(\valt)B^{1/2}\jz\E(B)^{-1/2}:\,\B(\hil)\to
\E(B)^0\B(\kil)\E(B)^0$ is 
unital and positive, whence, by 
\cite{Davis_Schwarz} and \cite[Theorem 2.1]{Choi_Schwarz}, $f\bz\E_B\bz B^{-1/2}AB^{-1/2}\jz\jz\le \E_B\bz f\bz  B^{-1/2}AB^{-1/2}\jz\jz$,
from which \eqref{eq:opersp posmon} follows by a trivial rearrangement.
\end{proof}

The above can be extended to more general pairs of PSD operators as follows. 
For a proof, see Propositions 3.25, 3.26, 3.30 in \cite{HiaiMosonyi2017}.

\begin{lemma}\label{lemma:support cond for persp}
Let $f:\,(0,+\infty)\to\bR$ be operator convex and $A,B\in\B(\hil)\p$ be such that at least one of the following conditions hold:
\begin{enumerate}
\renewcommand{\theenumi}{(S\arabic{enumi})}
\item\label{S1}
$A^0=B^0$.
\item\label{S2}
$f(0^+)<+\infty$ and $\tilde f(0^+)<+\infty$.
\item\label{S3}
$f(0^+)<+\infty$ and $A^0\le B^0$.
\item\label{S4}
$\tilde f(0^+)<+\infty$ and $A^0\ge B^0$.
\renewcommand{\theenumi}{(\roman{enumi})}
\end{enumerate}
Then for any sequence $C_n\in\B(\hil)\p$ such that $A+C_n,B+C_n\in\B(\hil)\pp$, $n\in\bN$, and 
$\lim_{n\to+\infty}C_n=0$, we have
\begin{align*}
\persp{f}(A,B)=\lim_{n\to+\infty}\persp{f}(A+C_n,B+C_n).
\end{align*}
Moreover, for any positive linear map $\E:\,\B(\hil)\to\B(\kil)$, 
$\persp{f}(\E(A),\E(B))$ is defined, and
the monotonicity inequality \eqref{eq:opersp posmon} holds.
\end{lemma}

\begin{rem}
Note that if $f:\,(0,+\infty)\to\bR$ is a non-negative operator monotone function 
(e.g., $f=\id_{(0,+\infty)}^t$ with some $t\in[0,1]$) then $-f$ satisfies 
\ref{S2} above. In particular, the 
conclusions of Lemma \ref{lemma:support cond for persp}
hold for any $A,B\in\B(\hil)\p$ and $-f$ in place of $f$.
For the weighted Kubo-Ando geometric means this gives the well-known fact that for any 
$t\in[0,1]$, any $A,B\in\B(\hil)\p$, and any positive map $\E:\,\B(\hil)\to\B(\kil)$, 
\begin{align}\label{eq:KA posmon}
\E(B)\#_t\E(A)\ge \E(B\#_t A).
\end{align}
\end{rem}

\begin{cor}\label{cor:oppersp posop convex}
\ki{(Generalized convexity with positive superoperators)}
Let $f$ and $A_i,B_i\in\B(\hil_i)\p$, $i\in[r]$, be such that 
for some $k\in[4]$, $(Sk)$ in Lemma \ref{lemma:support cond for persp} is satisfied for all $i\in[r]$.
For all $i\in[r]$, let $\E_i:\,\B(\hil_i)\to\B(\kil)$ be positive and $\lambda_i\in(0,+\infty)$.
Then $\persp{f}\bz\sum_i\lambda_i\E_i(A_i),\sum_i\lambda_i\E_i(B_i)\jz$ is defined, and 
\begin{align}
\persp{f}\bz\sum_{i=1}^r\lambda_i\E_i(A_i),\sum_{i=1}^r\lambda_i\E_i(B_i)\jz
\le
\sum_{i=1}^r\lambda_i\E_i\bz\persp{f}\bz A_i,B_i\jz\jz.
\end{align}
\end{cor} 
\begin{proof}
Let $\hat\kil:=\kil\otimes\bC^r$, $\hat A:=\sum_i\lambda_i\E_i(A_i)\otimes \pr{e_i}$, 
$\hat B:=\sum_i\lambda_i\E_i(B_i)\otimes \pr{e_i}$, where $(e_i)_{i=1}^r$ is an orthonormal basis in 
$\bC^r$, and let 
$\Tr_2:\,\B(\hat\kil)\to\B(\kil)$, $\Tr_2:\,X\otimes Y\mapsto X\Tr Y$ be the partial trace. Then 
\begin{align*}
\persp{f}\bz\sum_{i=1}^r\lambda_i\E_i(A_i),\sum_{i=1}^r\lambda_i\E_i(B_i)\jz
&=
\persp{f}\bz\Tr_2\hat A,\Tr_2\hat B\jz
\le
\Tr_2\persp{f}\bz\hat A,\hat B\jz\\
&=
\Tr_2\sum_{i=1}^r\persp{f}(\E_i(A_i),\E_i(B_i))\otimes\lambda_i\pr{e_i}\\
&=\sum_{i=1}^r\lambda_i\persp{f}(\E_i(A_i),\E_i(B_i))
\le
\sum_{i=1}^r\lambda_i\E_i\bz\persp{f}\bz A_i,B_i\jz\jz,
\end{align*}
where the first inequality is due to Lemma \ref{lemma:support cond for persp}, the second equality
follows from the direct sum property \eqref{eq:dsum1} and the tensor homogeneity \eqref{eq:oppersp tensor hom},
the third equality is trivial, and the second inequality follows again from Lemma \ref{lemma:support cond for persp}.
\end{proof} 

Corollary \ref{cor:oppersp posop convex} has a number of special cases that are worth spelling out in detail.

\begin{cor}\label{cor:oppersp posconvex}
\ki{(Generalized convexity with operator weights)}
Let $f$ and $A_i,B_i\in\B(\hil_i)\p$, $i\in[r]$, be such that 
for some $k\in[4]$, $(Sk)$ in Lemma \ref{lemma:support cond for persp} is satisfied for all $i\in[r]$, and 
let $K_i\in\B(\hil_i,\kil)$, $i\in[r]$, for some finite-dimensional Hilbert space $\kil$.
Then $\persp{f}\bz\sum_iK_iA_iK_i^*,\sum_iK_iB_iK_i^*\jz$ is defined, and 
\begin{align}
\persp{f}\bz\sum_{i=1}^rK_iA_iK_i^*,\sum_{i=1}^rK_iB_iK_i^*\jz
\le
\sum_{i=1}^rK_i\persp{f}\bz A_i,B_i\jz K_i^*.
\end{align}
\end{cor} 
\begin{proof}
Follows from Corollary \ref{cor:oppersp posop convex} by taking $\lambda_i=1$,
$\E_i(\valt):=K_i(\valt)K_i^*$, $i\in[r]$.
\end{proof} 

\begin{cor}
\emph{(Transformer inequality for the operator perspective)}
Assume that $f$ and $A,B\in\B(\hil)\p$ are related such that (Sk) 
holds for some $k\in[4]$ in Lemma \ref{lemma:support cond for persp}.
Then for any $X\in\B(\hil,\kil)$, $f$ and $XAX^*,XBX^*$ satisfy the same (Sk), and
\begin{align}\label{eq:oppersp transformer}
\persp{f}(XAX^*,XBX^*)\le X\persp{f}(A,B)X^*.
\end{align}
In particular, if $X$ is invertible then we have the following
\emph{operator homogeneity} relation:
\begin{align*}
\persp{f}(XAX^*,XBX^*)= X\persp{f}(A,B)X^*.
\end{align*}
\end{cor}
\begin{proof}
Follows from Corollary \ref{cor:oppersp posconvex} by taking $r=1$.
\end{proof} 
 
\begin{cor}\label{cor:oppersp conv}
\ki{(Generalized convexity of the operator perspective)}
If $A_i,B_i\in\B(\hil)\p$, $i\in[r]$, are such that $\persp{f}(A_i,B_i)$ is defined for every $i\in[r]$, 
then for any $\lambda_i\in(0,+\infty)$, $i\in[r]$, 
$\persp{f}\bz\sum_i\lambda_iA_i,\sum_i\lambda_iB_i\jz$ is defined, and 
\begin{align}\label{eq:oppersp subadd}
\persp{f}\bz\sum_{i=1}^r\lambda_iA_i,\sum_{i=1}^r\lambda_iB_i\jz
\le
\sum_{i=1}^r\lambda_i\persp{f}\bz A_i,B_i\jz.
\end{align}
\end{cor}
\begin{proof}
For positive definite $A_i,B_i$, the assertion follows from  
Lemma \ref{lemma:oppersp mon} using the proof of Corollary \ref{cor:oppersp posop convex},
and the extension to general PSD arguments is trivial.
\end{proof}
 
\begin{rem}
The generalized joint convexity \eqref{eq:oppersp subadd} of the operator perspective function was first shown for positive definite arguments in \cite{Effros,ENG} by a different method.
We also note that operator convex perspectives have been developed in \cite{HiaiUedaWada2022} for general operator
convex functions on $(0,+\infty)$ and for general PSD operators, whose values
are allowed in the extended lower semibounded self-adjoint operators.
\end{rem}

\section{Superoperator perspective and geometric means for CP maps}
\label{sec:supop persp}

For two finite-dimensional Hilbert spaces $\hil,\kil$, 
let $\cp(\hil,\kil)$ denote the set of all completely positive linear maps from $\B(\hil)$ to $\B(\kil)$, and 
$\cptp(\hil,\kil)$ the set of completely positive trace-preserving linear maps  
from $\B(\hil)$ to $\B(\kil)$. We will often use the terminology in which elements of 
$\B(\B(\hil),\B(\kil))$ are called \ki{superoperators}.

Let $\hil$ be a finite-dimensional Hilbert space. Any vector of the form 
$\Psi\in\tilde\hil\otimes\hil$, $\Psi=\sum_{i=1}^d\tilde e_i\otimes e_i$, where
$\dim\tilde\hil=\dim\hil$, $(e_i)_{i=1}^d$ is an ONB in $\hil$ and 
$(\tilde e_i)_{i=1}^d$ is an ONB in $\tilde \hil$, is called 
a \ki{Choi vector} on $\hil$. 
For any finite-dimensional Hilbert space $\kil$, $\Psi$ defines an
isomorphism
\begin{align*}
\choi_{\Psi}:\,\B(\B(\hil),\B(\kil))\to\B(\tilde\hil)\otimes\B(\kil),
\end{align*}
called the \ki{Choi-Jamio\l kowski isomorphism} \cite{Choi1975,Jamiolkowski1972}, as
\begin{align}\label{eq:Choiop def}
\choi_{\Psi}:\,\N\mapsto(\id_{\B(\tilde\hil)}\otimes \N)\pr{\Psi}=\sum_{i,j=1}^d
\diad{\tilde e_i}{\tilde e_j}\otimes\N(\diad{e_i}{e_j}),\ds\ds\ds
\N\in\B(\B(\hil),\B(\kil)).
\end{align}
Note that if $\Phi=\sum_{i=1}^d\hat f_i\otimes f_i\in\hat\hil\otimes\hil$ is another Choi vector then there
exists a unitary $U:\,\tilde\hil\to\hat\hil$ such that 
$\Phi=(U\otimes I)\Psi$, whence
\begin{align}\label{eq:Choi basis transformation}
\choi_{\Phi}=\Ad_{U\otimes I}\circ\choi_{\Psi},\ds\text{i.e.,}\ds
 \choi_{\Phi}(\N)=(U\otimes I)\choi_{\Psi}(\N)(U^*\otimes I),\ds\ds\ds
 \N\in\B(\B(\hil),\B(\kil)).
\end{align} 
For the completely depolarizing channel $\cd_{\hil\to\kil}(\valt):=\frac{1}{\dim\kil}I_{\kil}\Tr_{\hil}(\valt)$, we have
\begin{align*}
\choi_{\Psi}(\cd_{\hil\to\kil})=\frac{1}{\dim\kil}I_{\tilde\hil}\otimes I_{\kil}.
\end{align*}
For $\N,\M\in\cp(\hil,\kil)$, we will write 
\begin{align*}
\N\le_{\cp}\M,\ds\ds\text{if}\ds\ds \M-\N\in\cp(\hil,\kil)\ds\ds\text{iff}\ds\ds 
\choi_{\Psi}(\N)\le\choi_{\Psi}(\M),
\end{align*}
where the last inequality is for any (equivalently, all) Choi vector, and 
the inequality is in the PSD order.

The notion of the operator perspective can be extended to define the 
\ki{superoperator perspective} function of a continuous function $f:\,(0,+\infty)\to\bR$ on a pair of completely positive maps 
$\N,\M\in\cp(\hil,\kil)$ as 
\begin{align}\label{eq:superoppersp def}
\persp{f}(\N,\M):=C_{\Psi}\inv\bz\persp{f}\bz\choi_{\Psi}(\N),\choi_{\Psi}(\M)\jz \jz,
\end{align}
whenever $\persp{f}\bz\choi_{\Psi}(\N),\choi_{\Psi}(\M)\jz$ is well defined. 
According to \eqref{eq:oppersp def2}, we have 
\begin{align}\label{eq:superoppersp def2}
\persp{f}(\N,\M)=\lim_{\ep\searrow 0}\persp{f}(\N+\ep\cd,\M+\ep\cd),
\end{align}
whenever this limit exists; in particular, this is the case when $f$ is a non-negative operator monotone 
function and $\N,\M\in\cp(\hil,\kil)$ are arbitrary, and hence the $t$-weighted Kubo-Ando geometric mean 
\begin{align}\label{eq:supop gmean}
\M\#_{t}\N=C_{\Psi}\inv\bz C_{\Psi}(\M)\#_t C_{\Psi}(\N)\jz
\end{align} 
is well defined for any pair of $\N,\M\in\cp(\hil,\kil)$ and any $t\in[0,1]$.
More generally, if $f$ is operator convex, and $f$ and $A:=\choi_{\Psi}(\N)$, $B:=\choi_{\Psi}(\M)$ 
satisfy at least one of the conditions \ref{S1}--\ref{S4} in Lemma \ref{lemma:support cond for persp} 
then $\persp{f}(\N,\M)$ is defined.

According to the following lemma, 
$\persp{f}(\N,\M)$ is independent of the Choi vector $\Psi$ in the definitions 
\eqref{eq:superoppersp def}--\eqref{eq:superoppersp def2}, and only depends on the pair of CP maps
$(\N,\M)$. 
\begin{lemma}
For any Choi vectors $\Psi\in\tilde\hil\otimes\hil$ and $\Phi\in\hat\hil\otimes\hil$, and any 
$\N,\M\in\CP(\hil,\kil)$, 
\begin{align*}
C_{\Psi}\inv\bz\persp{f}\bz\choi_{\Psi}(\N),\choi_{\Psi}(\M)\jz \jz
=
C_{\Phi}\inv\bz\persp{f}\bz\choi_{\Phi}(\N),\choi_{\Phi}(\M)\jz \jz.
\end{align*}
\end{lemma}
\begin{proof}
Immediate from \eqref{eq:Choi basis transformation}.
\end{proof}

It is easy to see that
the notion of the superoperator perspective is a natural extension of the corresponding notion of the operator perspective, in the following sense:
\begin{lemma}
Let $A,B\in\B(\kil)\p$ be such that $\persp{f}(A,B)$ is defined, and let 
$\A(\valt):=A\Tr_{\hil}(\valt)$, 
$\B(\valt):=B\Tr_{\hil}(\valt)$. Then $\persp{f}(\A,\B)$ is defined, and 
\begin{align*}
\persp{f}(\A,\B)=\persp{f}(A,B)\Tr_{\hil}(\valt).
\end{align*}
\end{lemma}
\begin{proof}
Easy to verify from the fact that $\choi_{\Psi}(\A)=I_{\tilde\hil}\otimes A$ and 
$\choi_{\Psi}(\B)=I_{\tilde\hil}\otimes B$ using \eqref{eq:oppersp pr tensor hom}.
\end{proof}

Next, we establish various properties of the superoperator perspective in analogy to those of the operator perspective. 

\begin{lemma}
\ki{(Generalized convexity)}
If $f$ is operator convex and $\N_i,\M_i\in\cp(\hil,\kil)$, $i\in[r]$ are such that 
$\persp{f}(\N_i,\M_i)$ are defined for every $i\in[r]$, then 
$\persp{f}(\sum_i\lambda_i\N_i,\sum_i\lambda_i\M_i)$ is defined for any 
$\lambda_i\in(0,+\infty)$, $i\in[r]$, and 
\begin{align}\label{eq:supop persp conv}
\persp{f}\bz\sum_{i=1}^r\lambda_i\N_i,\sum_{i=1}^r\lambda_i\M_i\jz
\le_{\cp}
\sum_{i=1}^r\lambda_i\persp{f}(\N_i,\M_i).
\end{align}
\end{lemma}
\begin{proof}
Immediate from the generalized convexity of the operator perspective function; see
Appendix \ref{sec:oppersp}.
\end{proof}

It is easy to see that 
if $\Psi_1\in\tilde\hil_1\otimes\hil_1$ and 
$\Psi_2\in\tilde\hil_2\otimes\hil_2$ are Choi vectors then 
$\Psi_1\otimes\Psi_2$ is a Choi vector, and 
\begin{align}\label{eq:Choiop mult}
\choi_{\Psi_1\otimes\Psi_2}(\F_1\otimes\F_2)=\choi_{\Psi_1}(\F_1)\otimes\choi_{\Psi_2}(\F_2)
\end{align}
for any $\F_i\in\cp(\hil_i,\kil_i)$, $i=1,2$. In particular, 
\begin{align}\label{eq:Choiop mult2}
C_{\Psi^{\otimes n}}(\F^{\otimes n})=\bz C_{\Psi}(\F)\jz^{\otimes n},
\end{align}
for any $\F\in\cp(\hil,\kil)$ and any $n\in\bN$.

\begin{lemma} \ki{(Tensor homogeneity/stability)}
If $f$ and $\N,\M\in\cp(\hil_1,\kil_1)$ are such that 
$\persp{f}(\N,\M)$ is defined,
then for any $\E\in\cp(\hil_2,\kil_2)$, 
$\persp{f}(\N\otimes\E,\M\otimes\E)$ is defined, and
\begin{align}\label{eq:supop persp tensor hom}
\persp{f}(\N\otimes\E,\M\otimes\E)
=
\persp{f}(\N,\M)\otimes\E.
\end{align}
\end{lemma}
\begin{proof}
This follows from the tensor homogeneity \eqref{eq:oppersp tensor hom} of the operator perspective by taking into account \eqref{eq:Choiop mult}.
\end{proof}

\begin{prop}\label{prop:postprop mon}
\ki{(Monotonicity under post-processing)}
Assume that $f$ is operator convex, and $\N,\M\in\cp(\hil,\kil)$
are such that at least one of \ref{S1}--\ref{S4} in Lemma \ref{lemma:support cond for persp} holds with $A:=\choi_{\Psi}(\N)$ and 
$B:=\choi_{\Psi}(\M)$ for some (equivalently, any) Choi vector $\Psi$.
Then  $\persp{f}(\F\circ\N,\F\circ\M)$ is defined for any 
$\F\in\cp(\kil,\lil)$, and
\begin{align}\label{eq:postprop mon}
\persp{f}(\F\circ\N,\F\circ\M)\le_{\cp}\F\circ\persp{f}(\N,\M).
\end{align}
If, moreover,
$\F$ has a completely positive inverse (in particular, if $\F$ is a unitary channel) then 
equality holds in \eqref{eq:postprop mon}.
\end{prop}
\begin{proof}
Note that $\choi_{\Psi}(\F\circ\N)=(\id\otimes\F)\choi_{\Psi}(\N)$, and similarly for $\M$. 
In particular, one of \ref{S1}--\ref{S4} holds also with 
$A:=\choi_{\Psi}(\F\circ\N)$ and $B:=\choi_{\Psi}(\F\circ\M)$.
Hence,
\begin{align*}
\choi_{\Psi}\bz\persp{f}(\F\circ\N,\F\circ\M)\jz
&=
\persp{f}((\id\otimes\F)\choi_{\Psi}(\N),(\id\otimes\F)\choi_{\Psi}(\M))\\
&\le
(\id\otimes\F)\persp{f}(\choi_{\Psi}(\N),\choi_{\Psi}(\M))\\
&=
(\id\otimes\F)\choi_{\Psi}(\persp{f}(\N,\M))\\
&=
\choi_{\Psi}\bz\F\circ(\persp{f}(\N,\M))\jz,
\end{align*}
where the equalities follow from the respective definitions, and the inequality is due to 
the monotonicity of the operator perspective function stated in Lemma \ref{lemma:support cond for persp}.
\end{proof}

\begin{prop}\label{prop:preprop mon}
\ki{(Monotonicity under pre-processing)}
Assume that $f$ is operator convex, and $\N,\M\in\cp(\hil,\kil)$
are such that at least one of \ref{S1}--\ref{S4} in Lemma \ref{lemma:support cond for persp} holds with $A:=\choi_{\Psi}(\N)$ and 
$B:=\choi_{\Psi}(\M)$ for some (equivalently, any) Choi vector $\Psi$.
Then $\persp{f}(\N\circ\E,\M\circ\E)$ is defined
for any $\E\in\cp(\lil,\hil)$, and 
\begin{align}\label{eq:preprop mon}
\persp{f}(\N\circ\E,\M\circ\E)\le_{\cp}\persp{f}(\N,\M)\circ\E.
\end{align}
\end{prop}
\begin{proof}
Let $\tilde\hil\otimes\hil\ni\Psi=\sum_{i}\tilde e_i\otimes e_i$
and 
$\tilde\lil\otimes\lil\ni\Phi=\sum_{j}\tilde f_j\otimes f_j$
be Choi vectors on $\hil$ and on $\lil$, respectively.
Note that, according to the standard transpose trick, 
if $\E(\valt)=\sum_{i=1}^r K_i(\valt)K_i^*$ is a Kraus decomposition of $\E$ then 
\begin{align*}
\choi_{\Phi}(\E)
=
\sum_{i=1}^r(I\otimes K_i)\pr{\Phi}(I\otimes K_i^*)
=
\sum_{i=1}^r(K_i\trans\otimes I)\pr{\Psi}((K_i\trans)^*\otimes I),
\end{align*}
where $K_i\trans:=\sum_{k=1}^{\dim\hil}\sum_{l=1}^{\dim\lil}\inner{e_k}{K_if_l}\sdiad{\tilde f_l}{\tilde e_k}$.
Thus,
\begin{align*}
\choi_{\Phi}(\N\circ\E)
&=
(\id\otimes\N)(\id\otimes\E)\pr{\Phi}\\
&=
\sum_{i=1}^r(\id\otimes\N)(K_i\trans\otimes I)\pr{\Psi}((K_i\trans)^*\otimes I)\\
&=
\sum_{i=1}^r(K_i\trans\otimes I)\underbrace{(\id\otimes\N)\pr{\Psi}}_{=\choi_{\Psi}(\N)}((K_i\trans)^*\otimes I),
\end{align*}
and similarly for $\choi_{\Phi}(\M\circ\E)$. Therefore,
\begin{align*}
\choi_{\Phi}\bz\persp{f}(\N\circ\E,\M\circ\E)\jz
&=
\persp{f}\bz\choi_{\Phi}(\N\circ\E),\choi_{\Phi}(\M\circ\E)\jz\\
&=
\persp{f}\bz\sum_{i=1}^r(K_i\trans\otimes I)\choi_{\Psi}(\N)((K_i\trans)^*\otimes I),
\sum_{i=1}^r(K_i\trans\otimes I)\choi_{\Psi}(\M)((K_i\trans)^*\otimes I)\jz\\
&\le
\sum_{i=1}^r(K_i\trans\otimes I)
\persp{f}\bz\choi_{\Psi}(\N),\choi_{\Psi}(\M)\jz
((K_i\trans)^*\otimes I)\\
&=
\sum_{i=1}^r(K_i\trans\otimes I)
\choi_{\Psi}\bz\persp{f}\bz\N,\M\jz\jz
((K_i\trans)^*\otimes I)\\
&=
\sum_{i=1}^r(K_i\trans\otimes I)
\bz\id\otimes\persp{f}\bz\N,\M\jz\jz\pr{\Psi}
((K_i\trans)^*\otimes I)\\
&=
\bz\id\otimes\persp{f}\bz\N,\M\jz\jz
\sum_{i=1}^r(K_i\trans\otimes I)\pr{\Psi}
((K_i\trans)^*\otimes I)\\
&
=\bz\id\otimes\persp{f}\bz\N,\M\jz\jz\bz\id\otimes\E\jz\pr{\Phi}\\
&=
\choi_{\Phi}\bz\persp{f}(\N,\M)\circ\E\jz,
\end{align*}
where the inequality is due to Corollary \ref{cor:oppersp posconvex}, and the rest of the steps are straightforward.
\end{proof}

Quantum supermaps are linear transformations mapping superoperators to superoperators;
they are called completely positive (CP) if 
whenever they are acting on half of a bipartite completely positive superoperator, the resulting superoperator is again completely positive. 
This implies as a special case that CP supermaps map CP superoperators to CP superoperators.
According to \cite{Chiribella_supermaps}, any 
completely positive supermap $\smap:\,\B(\B(\hil_1),\B(\kil_1))\to\B(\B(\hil_2),\B(\kil_2))$ can be written as a composition of three supermaps: first tensoring the input with a fixed CP map, and then applying pre- and postprocessing. Thus, \eqref{eq:supop persp tensor hom} and Propositions \ref{prop:postprop mon} and \ref{prop:preprop mon} yield immediately the following:

\begin{prop}\label{prop:supermap mon}
\ki{(Monotonicity under CP supermaps)}
Assume that $f$ is operator convex, and $\N,\M\in\cp(\hil_1,\kil_1)$
are such that at least one of \ref{S1}--\ref{S4} in Lemma \ref{lemma:support cond for persp} holds with $A:=\choi_{\Psi}(\N)$ and 
$B:=\choi_{\Psi}(\M)$ for some (equivalently, any) Choi vector $\Psi$.
Then, for any completely positive supermap
$\smap:\,\B(\B(\hil_1),\B(\kil_1))\to\B(\B(\hil_2),\B(\kil_2))$,
$\persp{f}(\smap(\N),\smap(\M))$ is defined, and 
\begin{align*}
\persp{f}(\smap(\N),\smap(\M))\le_{\CP}\smap\bz\persp{f}(\N,M)\jz.
\end{align*}
\end{prop}

Finally, we note that similarly to the above, any notion of geometric mean, in fact, any (multi-variate) positive operator function $\Gm$ can be extended to channels as 
\begin{align}\label{eq:supop G}
\Gm((\N_y)_{y\in\Y}):=
C_{\Psi}\inv\Gm\bz(C_{\Psi}(\N_y))_{y\in\Y}\jz.
\end{align}
By the same arguments as above, this is well-defined in the sense of being independent of the choice of $\Psi$.
If, moreover, $\Gm$ satisfies the invariance property 
$A\in\cD(\Gm)$ $\imp$ $(I_{\bC^m}\otimes A_y)_{y\in\Y}\in\cD(\Gm)$ and 
$\Gm\bz(I_{\bC^m}\otimes A_y)_{y\in\Y}\jz=I_m\otimes\Gm((\ops_y)_{y\in\Y})$,
 $m\in\bN$, then \eqref{eq:supop G} 
gives an extension of $\Gm$ defined on operators in the sense that
\begin{align*}
\Gm((N_y\Tr(\valt))_{y\in\Y})=
\Gm((N_y)_{y\in\Y})\Tr(\valt),\ds\ds\ds (N_y)_{y\in\Y}\in\cD(\Gm).
\end{align*}

\begin{rem}
Note that $\N\in\CP(\hil,\kil)$
is trace-preserving (resp., trace non-increasing,
i.e., $\Tr\circ\N\le\Tr$) if and only if $\Tr_{\kil}C_{\Psi}(\N)=I_{\tilde\hil}$ 
(resp., $\Tr_{\kil}C_{\Psi}(\N)\le I_{\tilde\hil}$).
Now, let $\N,\M\in\cptp(\hil,\kil)$. 
Then it follows from \eqref{eq:KA posmon}  that $\N\#_t\M$ is 
trace non-increasing for every $t\in[0,1]$.
Also, it follows from Lemma \ref{lemma:oppersp mon} that if 
$f:\,(0,+\infty)\to\bR$ is operator convex with $f(1)=1$ and $\persp{f}(\N,\M)$ is 
defined, then it is trace non-decreasing.
\end{rem}

\section{The case of disjoint supports}
\label{sec:disjoint}

For the rest, $\hil$ will denote a finite-dimensional Hilbert space with $d:=\dim\hil$. 
According to Jordan's lemma (see \cite{Jordan1875} for the original, and \cite[Lemma 1]{Regev_Jordan} for a simple proof), for any two projections $S,Q\in\bP(\hil)$, there exists an
orthonormal basis $\{e_k,e_k^{\perp}\}_{k=1}^m\cup\{f_k\}_{k=1}^{d-2m}$, such that the subspaces
$\hil_k:=\spann\{e_k,e_k^{\perp}\}$, $k\in[m]$, and $\hil':=\spann\{f_k:\,k\in[d-2m]\}$ are invariant under both $S$ and $Q$, and we have
\begin{align}\label{eq:Jordan}
S=\left[\medoplus_{k\in[m]}\pr{e_k}\right]\medoplus S',\ds\ds\ds\ds
Q=\left[\medoplus_{k\in[m]}\pr{\phi_k}\right]\medoplus Q',
\end{align}
where
\begin{align*}
S':=S\big\vert_{\hil'},\ds\ds\ds
Q':=Q\big\vert_{\hil'},\ds\ds\ds
S'Q'=Q'S',
\end{align*}
and
\begin{align*}
\phi_k=(\cos\theta_k)e_k+(\sin\theta_k)e_k^{\perp}\ds\ds\ds\text{with some}\ds\ds\ds
\theta_k\in(0,\pi/2),\ds\ds k\in[m].
\end{align*}
The \ki{overlap} of $S$ and $Q$ is defined as
\begin{align*}
\overlap(S,Q)&:=\norm{SQ}\\
&=\max\left\{0,|\inner{v}{w}|:\,v\in\ran(S),\,w\in\ran(Q),\,\norm{v}=\norm{w}=1\right\}\\
&=
\max\left\{\max_{k\in[m]}\cos\theta_k,\,\norm{S'Q'}\right\}.
\end{align*}

\begin{rem}
Note that the decompositions in \eqref{eq:Jordan} are not unique. 
However, $\hil',S',Q'$ are uniquely determined by $S$ and $Q$, as $\hil'$ is the largest subspace invariant under both $S$ and $Q$ on which the restrictions of the two projections commute. Moreover,
the angles $\theta_k$ with their multiplicities are also uniquely determined by $S$ and $Q$, since the 
eigenvalues of $S+Q$ in $(0,1)\cup(1,2)$ are exactly $1\pm\cos\theta_k$, $k\in[m]$.
As a consequence, the subspaces 
\begin{align*}
\hil_{\theta}:=\spann\left\{\cup_{k\in[m]:\,\theta_k=\theta}\hil_k\right\},\ds\ds\ds
\theta\in(0,\pi/2),
\end{align*}
are uniquely determined by $S$ and $Q$, and therefore so are the corresponding projections $P_{\theta}$.

For what follows, only these uniquely determined parameters play a role, 
and therefore all arguments work with any decompositions of the form \eqref{eq:Jordan}. 
\end{rem}

The following two lemmas are straightforward to verify, and hence we omit their proofs.

\begin{lemma}\label{lemma:eple}
Let $S,Q\in\bP(\hil)$ be projections with decompositions as in \eqref{eq:Jordan}. For any 
$\ep\in[0,1)$, the following are equivalent:
\begin{align*}
(1-\ep^2)\norm{Qv}^2\le \norm{SQv}^2,\ds v\in\hil
&\ds\ds\iff\ds\ds
(1-\ep^2)Q\le QSQ\\
&\ds\ds\iff\ds\ds
Q'\le S'\ds\text{and}\ds \sin\theta_k\le\ep,\ds k\in[m].
\end{align*}
\end{lemma}

\begin{lemma}\label{lemma:eport}
Let $S,Q\in\bP(\hil)$ be projections with decompositions as in \eqref{eq:Jordan}. For any 
$\ep\in[0,1)$, the following are equivalent:
\begin{align*}
\norm{SQv}^2\le \ep^2\norm{Qv}^2,\ds v\in\hil
&\ds\ds\iff\ds\ds
QSQ\le\ep^2 Q\\
&\ds\ds\iff\ds\ds
Q'S'=0\ds\text{and}\ds \cos\theta_k\le\ep,\ds k\in[m]\\
&\ds\ds\iff\ds\ds
\overlap(S,Q)\le\ep\\
&\ds\ds\iff\ds\ds
(1-\ep)(Q\vee S)\le Q+S\le (1+\ep)(Q\vee S)\\
&\ds\ds\iff\ds\ds
SQS\le\ep^2 S\\
&\ds\ds\iff\ds\ds
\norm{QSv}^2\le \ep^2\norm{Sv}^2,\ds v\in\hil.
\end{align*}
\end{lemma}

Note that for $\ep=0$, the conditions in Lemma \ref{lemma:eple} are all equivalent to 
$Q\le S$, and the conditions in Lemma \ref{lemma:eport} are all equivalent to 
$Q\perp S$. This motivates the following:

\begin{definition}
For two projections $Q,S\in\bP(\hil)$ and an $\ep\in[0,1)$, 
we say that $Q$ is \ki{$\ep$-dominated} by $S$ (in notation $Q\le_{\ep}S$), if they satisfy any (and hence all)
of the equivalent conditions in Lemma \ref{lemma:eple}, and we say that they are
\ki{$\ep$-orthogonal} (in notation $Q\perp_{\ep}S$), if
they satisfy any (and hence all)
of the equivalent conditions in Lemma \ref{lemma:eport}.
\end{definition} 

\begin{rem}\label{rem:eple-eport}
A further motivation to the above definitions is that the relation 
$Q\le S\iff Q\perp(I-S)$ is preserved as 
$Q\le_{\ep} S\iff Q\perp_{\ep}(I-S)$. Indeed, by Lemmas \ref{lemma:eple}--\ref{lemma:eport},
\begin{align*}
Q\le_{\ep} S
\ds\iff\ds
(1-\ep^2)Q\le QSQ=Q-Q(I-S)Q
\ds\iff\ds
Q(I-S)Q\le\ep^2 Q
\ds\iff\ds
Q\perp_{\ep}(I-S).
\end{align*}
\end{rem}

The following lemma can be obtained from \cite[Lemma 6]{KeLi2016} with 
$\ep(r,t)=\frac{(t-1)}{(r-1)(2t-1)}$. We give a different proof below, 
leading to other possible choices of $\ep(r,t)$, which can sometimes be better (larger) than the above,
although we will not need this in the applications. 

\begin{lemma}\label{lemma:support sum domination}
For any $r\in\bN\setminus\{1\}$ and any $t\in(1,+\infty)$ there exists an $\epsilon(r,t) \in (0,1)$ such that for any finite-dimensional Hilbert space $\hil$, and any projections $Q_1,\ldots, Q_r\in \bP(\hil)$,
\begin{align}\label{eq:support sum domination}
Q_j\perp_{\ep(r,t)}Q_k,\ds\ds\ds j,k\in[r],\,j\ne k\ds\ds\imp\ds\ds
\medvee_{j\in [r]}Q_j \leq t \sum_{j\in [r]}Q_j.
\end{align}
\end{lemma}
\begin{proof}
According to Lemma \ref{lemma:eport}, $\ep(2,t):=1-\frac{1}{t}$ satisfies \eqref{eq:support sum domination}.
We now proceed by induction.
Assume that for every $r=2,\ldots,n$, some functions $\ep(r,\valt)$ have been chosen 
that satisfy \eqref{eq:support sum domination}.
Let $f:\,(1,+\infty)\to(1,+\infty)$ be an arbitrary function such that $f(t)< t$, $t\in(1,+\infty)$.
We will search for a function $\ep(n+1,\valt)$ satisfying $\ep(n+1,t)\le \ep(n,f(t))$ and some additional properties. 
Assume that with such a function, we have $Q_j\perp_{\ep(n+1,t)} Q_k$, $j\ne k$, for some projections
$Q_1,\ldots,Q_{n+1}\in\bP(\hil)$. By definition, 
$Q_jQ_kQ_j\le\ep(n+1,t)^2Q_j\le\ep(n,f(t))^2Q_j$, $j,k\in[n]$, $j\ne k$, whence 
$\vee_{j\in [n]}Q_j\le f(t)\sum_{j\in[n]}Q_j$.
Thus,
\begin{align}\label{eq:support sum domination proof1}
Q_{n+1}\Big(\medvee_{j\in [n]}Q_j\Big) Q_{n+1}
&\le
f(t)\sum_{j\in[n]}\underbrace{Q_{n+1}Q_j Q_{n+1}}_{\le\ep(n+1,t)^2Q_{n+1}}
\le 
nf(t)\ep(n+1,t)^2 Q_{n+1},
\end{align}
where in the second inequality we used the assumption
$Q_j\perp_{\ep(n+1,t)} Q_{n+1}$, $j\in[n]$.
If $nf(t)\ep(n+1,t)^2<1$ then 
\begin{align*}
\medvee_{j\in [n+1]}Q_j 
&=
Q_{n+1}\, \medvee\, \Big(\medvee_{j\in [n]}Q_j\big)\\
&\le 
\frac{1}{1-\ep(n+1,t)\sqrt{nf(t)}}\Big(Q_{n+1} + \Big(\medvee_{j\in [n]}Q_j\Big)\Big)
\\
&\le 
\frac{1}{1-\ep(n+1,t)\sqrt{nf(t)}}
\left(Q_{n+1} + f(t)\left(\sum_{j\in [n]}Q_j\right)\right)
\\
&\le
\frac{f(t)}{1-\ep(n+1,t)\sqrt{nf(t)}}\sum_{j\in [n+1]}Q_j.
\end{align*}
where the first inequality is due to \eqref{eq:support sum domination proof1} according to Lemma \ref{lemma:eport}, and in the third inequality we used that 
$f(t)> 1$. The last expression above is upper bounded by $t\sum_{j\in[n+1]}Q_j$ if and only if 
$\ep(n+1,t)\le (1-f(t)/t)/\sqrt{nf(t)}$. 
Thus, for any choice of $f$ as above, and for any function $g$ defined on $(1,+\infty)$ such that 
\begin{align*}
0<g(t)\le \min\left\{\ep(n,f(t)),\frac{1}{\sqrt{nf(t)}}\bz 1-\frac{f(t)}{t}\jz\right\},
\end{align*}
\eqref{eq:support sum domination} holds for $r=n+1$ with 
$\ep(n+1,t):=g(t)$, $t\in(1,+\infty)$.
\end{proof}

\begin{rem}
Note that the condition 
$\vee_{j\in [r]}Q_j \leq t \sum_{j\in [r]}Q_j$ in \eqref{eq:support sum domination} 
holds trivially when all $Q_j=0$, and otherwise it
is equivalent to 
the smallest non-zero eigenvalue of $\sum_{j\in [r]}Q_j$ being at least as large as $1/t$.
\end{rem}

The following notion of $\ep$-subtraction is a slight variation of that introduced in 
\cite[Definition 3]{KeLi2016}, 
modified so that it fits the notion of $\ep$-orthogonality introduced above:
\begin{definition}
For any two projections $Q,S\in\bP(\hil)$ with decompositions \eqref{eq:Jordan}, and for any 
$\ep\in[0,1)$, let 
\begin{align*}
Q\ominus_{\ep}S:=Q-\medoplus_{k\in[m]:\,\cos\theta_k> \ep}\pr{\phi_k}-Q'S'
=
\left[\medoplus_{k\in[m]:\,\cos\theta_k\le\ep}\pr{\phi_k}\right]\medoplus Q'(I-S').
\end{align*}
\end{definition}
Indeed, the only difference compared to 
\cite[Definition 3]{KeLi2016} is the interchange of strict and non-strict inequalities.

A dual concept can be defined as follows:

\begin{definition}
For any two projections $Q,S\in\bP(\hil)$ with decompositions \eqref{eq:Jordan}, and for any 
$\ep\in[0,1)$, let 
\begin{align*}
Q_{S,\ep}:=\left[\medoplus_{k\in[m]:\,\sin\theta_k\le\ep}\pr{\phi_k}\right]\medoplus Q'S'.
\end{align*}
\end{definition}

Note that $Q\ominus_{\ep}S$ and $Q_{S,\ep}$ are uniquely determined by $S$ and $Q$, and therefore do not depend 
on the particular decomposition of the form \eqref{eq:Jordan}, as 
\begin{align*}
Q\ominus_{\ep}S=Q\bz P_{\hil'}-S'+\sum_{\theta\ge\arccos\ep}P_{\theta}\jz,\ds\ds\ds
Q_{S,\ep}=Q\bz S'+\sum_{\theta\le\arcsin\ep}P_{\theta}\jz,
\end{align*}
where $P_{\hil'}$ is the projection onto $\hil'$.
It is clear from the definitions that 
\begin{align*}
&Q\ominus_{\ep}S\le Q,\ds\ds\ds\ds
Q\ominus_{\ep}S\perp_{\ep}S,\ds\ds\ds\ds\text{and}\ds\ds\ds\ds
Q_{S,\ep}\le Q,\ds\ds\ds\ds 
Q_{S,\ep}\le_{\ep}S.
\end{align*}

We will need the following:

\begin{lemma}\label{lemma:projekcio_vagas}
Let $\rho \in \S(\hil)$ be a density operator and $S:=\rho^0$.
For any $Q\in \mathcal \bP(\hil)$ and any $\ep\in(0,1)$, 
\begin{align}\label{eq:projekcio_vagas}
 \Tr\rho(I-Q)\le \Tr\rho(I-Q_{S,\ep})\leq  \frac{1}{\ep^2} \Tr\rho(I-Q).
\end{align}
\end{lemma}
\begin{proof}
The first inequality in \eqref{eq:projekcio_vagas} is trivial from $Q_{S,\ep}\le Q$. 
Note that 
\begin{align*}
I-Q&=\left[\medoplus_{k\in[m]}\pr{\phi_k^{\perp}}\right]\medoplus (I_{\hil'}-Q'),
\end{align*}
where $\phi_k^{\perp}:=(\sin\theta_k)e_k-(\cos\theta_k)e_k^{\perp}$, $k\in[m]$. Thus,
\begin{align*}
\Tr\rho(I-Q)&=
\Tr\rho S(I-Q)S\\
&=
\sum_{k:\,\sin\theta_k> \ep}\inner{e_k}{\rho e_k}\underbrace{\sin^2\theta_k}_{\ge\ep^2}+
\sum_{k:\,\sin\theta_k\le \ep}\inner{e_k}{\rho e_k}\sin^2\theta_k+
\Tr\rho (S'-S'Q')\\
&\ge
\ep^2\bz
\sum_{k:\,\sin\theta_k> \ep}\inner{e_k}{\rho e_k}+
\sum_{k:\,\sin\theta_k\le \ep}\inner{e_k}{\rho e_k}\sin^2\theta_k+
\Tr\rho (S'-S'Q')
\jz\\
&=\ep^2\Tr\rho S\underbrace{\bz\left[\medoplus_{k:\,\sin\theta_k> \ep}I_{\hil_k}\right]\medoplus
\left[\medoplus_{k:\,\sin\theta_k\le \ep}\pr{\phi_k^{\perp}}\right]
\medoplus (I_{\hil'}-Q'S')
\jz}_{=I-Q_{S,\ep}} S\\
&=
\ep^2\Tr\rho(I-Q_{S,\ep}), 
\end{align*}
proving the second inequality in \eqref{eq:projekcio_vagas}.
\end{proof}

\begin{proposition}\label{prop:disjpoint supports}
Let $\rho_1,\ldots ,\rho_r\in \S(\hil)$ with some $r\in\bN\setminus\{1\}$, and $T_{j,n}\in \mathcal \bP(\hil^{\otimes n})$, 
$j\in[r]$, $n\in\bN$. 
Assume that their supports are pairwise disjoint, i.e., 
\begin{align}\label{eq:disjoint supports}
\rho_j^0\wedge\rho_k^0=0,\ds\ds\ds j,k\in[r],\ds j\ne k.
\end{align}
Then for any $t>1$, any $\ep(r,t)$ as in Lemma \ref{lemma:support sum domination}, and any $0<\ep<\ep(r,t)/2$,
the projections $T_n:=\vee_{j\in[r]}(T_{j,n})_{(\rho_j^{\otimes n})^0,\ep}$, $n\in\bN$, satisfy the following:
\begin{enumerate}
\item\label{modifed projections1}
$\Tr\rho_j^{\otimes n} (I-T_n) \le \frac{1}{\ep^2} \Tr\rho_j^{\otimes n} (I-T_{j,n})$
\ds\ds\ds $j\in [r],\ds n\in\bN$.
\item\label{modifed projections2}
There exists an $N\in\bN$ such that $T_n\le t \sum_{j\in [r]} T_{j,n}$,\ds\ds\ds $n\ge N$.
\end{enumerate}
\end{proposition}
\begin{proof}
Let $t>1$, let $\ep(r,t)$ be as in Lemma \ref{lemma:support sum domination}, and let 
$0<\ep<\ep(r,t)/2$.
Let $R_{j,n}:=(\rho_j^{\otimes n})^0=(\rho_j^0)^{\otimes n}$ be the projection onto the support of 
$\rho_j^{\otimes n}$, and $\tilde T_{j,n}:=(T_{j,n})_{(\rho_j^{\otimes n})^0,\ep}$, so that 
$T_n=\vee_{j\in[r]}\tilde T_{j,n}$. Then
property \ref{modifed projections1} is satisfied as for every $j\in[r]$,
\begin{align*}
\Tr\rho_j^{\otimes n} (I-T_n)
\le
\Tr\rho_j^{\otimes n} (I-\tilde T_{j,n})
\le
 \frac{1}{\ep^2} \Tr\rho_j^{\otimes n} (I-T_{j,n}),
\end{align*}
where the first inequality is trivial, and the second one is immediate from Lemma \ref{lemma:projekcio_vagas}.

By the disjointness assumption \eqref{eq:disjoint supports}, there exists a $\lambda\in(0,1)$ such that 
$R_{j,1}R_{k,1} R_{j,1}\leq \lambda R_{j,1}$, $j,k\in[r]$, $j\ne k$, whence
\begin{align*}
R_{j,n} R_{k,n} R_{j,n}\le\lambda^n R_{j,n},\ds\ds\ds j,k\in[r],\ds j\ne k,\ds n\in\bN,
\end{align*}
which in turn yields
\begin{align}\label{eq:disjoint supports proof1}
\|R_{k,n}R_{j,n}\|^2=\|(R_{k,n}R_{j,n})^*R_{k,n}R_{j,n}\|=
\|R_{j,n}R_{k,n}R_{j,n}\|
\leq \lambda^n.
\end{align}

For any $j\in[r]$ and any $\varphi_j\in\ran \tilde T_{j,n}$,
$\tilde T_{j,n}\le_{\ep} R_{j,n}$ yields
\begin{align*}
\norm{(I-R_{j,n})\varphi_j}^2\le\ep^2\norm{\vfi_j}^2, 
\end{align*}
according to Remark \ref{rem:eple-eport}.
Thus, for any $j\ne k$ and $\varphi_j\in\ran \tilde T_{j,n}$, $\varphi_k\in\ran \tilde T_{k,n}$,
\begin{align*}
|\langle \varphi_j,\varphi_k\rangle|
&=
|\langle R_{j,n}\varphi_j+ (I-R_{j,n})\varphi_j,\varphi_k\rangle| 
\\
&\le
|\langle R_{j,n}\varphi_j, \varphi_k\rangle| + \ep\|\varphi_j\|\,\|\varphi_k\|
\\
&=
|\langle R_{j,n}\varphi_j, R_{k,n}\varphi_k+(I-R_{k,n})\varphi_k\rangle| + \ep\|\varphi_j\|\,\|\varphi_k\|
\\
&\le
|\langle R_{j,n}\varphi_j, R_{k,n}\varphi_k \rangle| +  2\ep \|\varphi_j\|\,\|\varphi_k\|
\\
&\le (\lambda^{\frac{n}{2}}+2\ep)\|\varphi_j\|\,\|\varphi_k\|,
\end{align*}
where the last inequality is due to \eqref{eq:disjoint supports proof1}.
This shows that 
\begin{align*}
\snorm{\tilde{T}_{k,n}\tilde{T}_{j,n}\tilde{T}_{k,n}}^{1/2}
&=
\snorm{\tilde{T}_{j,n}\tilde{T}_{k,n}}=
\max_{\norm{\vfi}=1=\norm{\psi}}\left|\inner{\vfi}{\tilde{T}_{j,n}\tilde{T}_{k,n}\psi}\right|
=
\max_{\norm{\vfi}=1=\norm{\psi}}\left|\inner{\tilde{T}_{j,n}\vfi}{\tilde{T}_{k,n}\psi}\right|\\
&\le
(\lambda^{\frac{n}{2}}+2\ep)\max_{\norm{\vfi}=1=\norm{\psi}}\snorm{\tilde{T}_{j,n}\vfi}\snorm{\tilde{T}_{k,n}\psi}
\le
(\lambda^{\frac{n}{2}}+2\ep),
\end{align*}
whence $\tilde{T}_{k,n}\tilde{T}_{j,n}\tilde{T}_{k,n}\le(\lambda^{\frac{n}{2}}+2\ep)^2I$.
This implies that 
$\tilde{T}_{k,n}\tilde{T}_{j,n}\tilde{T}_{k,n}\le(\lambda^{\frac{n}{2}}+2\ep)^2\tilde{T}_{k,n}$,
which can be further bounded as
\begin{align*}
\tilde{T}_{k,n}\tilde{T}_{j,n}\tilde{T}_{k,n}\le\ep(r,t)^2\tilde{T}_{k,n},
\end{align*}
whenever $n\ge \ceil{2\log(\ep(r,t)-2\ep)/(\log\lambda)}=:N$.
Thus, by Lemma \ref{lemma:support sum domination}, property \ref{modifed projections2} holds. 
\end{proof}

Given $\N,\A\subseteq\S(\hil)$, a 
test sequence $T_n\in\B(\hil^{\otimes n})_{[0,1]}$, $n\in\bN$, is said to achieve a direct exponent pair
$(\rate,\tilde \rate)\in(0,+\infty)^2$, if
\begin{align*}
&\liminf_{n\to+\infty}-\frac{1}{n}\log\sup_{\sigma\in\A}\Tr\sigma^{\otimes n}T_n> \rate,\\
&\liminf_{n\to+\infty}-\frac{1}{n}\log\sup_{\rho\in\N}\Tr\rho^{\otimes n}(I-T_n)\ge \tilde \rate.
\end{align*}
According to \cite[Lemma II.24]{MWSz2020}, any achievable direct exponent pair is also achievable by a sequence of projective tests, i.e., we may always assume that $T_n^2=T_n$, $n\in\bN$. 

\begin{prop}\label{prop:disjoint nullhypo direct exp}
Let $\A\subseteq\S(\hil)$ and $\rho_1,\ldots,\rho_r\in\S(\hil)$ 
with some $r\in\bN\setminus\{1\}$
be such that $\rho_j^0\wedge\rho_k^0=0$
for all $j,k\in[r]$, $j\ne k$. 
Assume that for every $j\in[r]$, a test sequence $T_{j,n}\in\bP(\hil^{\otimes n})$, $n\in\bN$, 
achieves a direct exponent pair
$(\rate_j,\tilde \rate_j)$ for the composite i.i.d.~state discrimination problem $\N_j:=\{\rho_j\}$
vs.~$\A$.  
Then for every small enough $\ep>0$,  the test sequence $T_n:=\vee_{j\in[r]}(T_{j,n})_{(\rho_j^{\otimes n})^0, \ep}$ achieves the exponent pair $(\min_j \rate_j,\min_j\tilde \rate_j)$ for the composite i.i.d.~state discrimination problem $\{\rho_j\}_{j\in[r]}$ vs. $\A$.
\end{prop}
\begin{proof}
Let $t>1$ be fixed, 
$\ep(r,t)$ be as in Lemma \ref{lemma:support sum domination}, and $\ep\in(0,\ep(r,t)/2)$.
Then 
\begin{align*}
\liminf_{n\to+\infty}-\frac{1}{n}\log\sup_{\sigma\in\A}\Tr\sigma^{\otimes n}T_n
&\ge
\liminf_{n\to+\infty}-\frac{1}{n}\log\sup_{\sigma\in\A}\Tr\sigma^{\otimes n}\sum_{j\in[r]}T_{j,n}\\
&\ge
\liminf_{n\to+\infty}-\frac{1}{n}\log\max_{j\in[r]}\sup_{\sigma\in\A}\Tr\sigma^{\otimes n}T_{j,n}\\
&=
\min_{j\in[r]}\liminf_{n\to+\infty}-\frac{1}{n}\log\sup_{\sigma\in\A}\Tr\sigma^{\otimes n}T_{j,n}\\
&>\min_{j\in[r]}\rate_j,
\end{align*}
where the first inequality is due to Proposition \ref{prop:disjpoint supports} \ref{modifed projections2}, and the rest are obvious. Similarly,
\begin{align*}
\liminf_{n\to+\infty}-\frac{1}{n}\log\max_{j\in[r]}\Tr\rho_j^{\otimes n}(I-T_n)
&\ge
\liminf_{n\to+\infty}-\frac{1}{n}\log\max_{j\in[r]}\Tr\rho_j^{\otimes n}(I-T_{j,n})\\
&=
\min_{j\in[r]}\liminf_{n\to+\infty}-\frac{1}{n}\log\Tr\rho_j^{\otimes n}(I-T_{j,n})
\ge\min_{j\in[r]}\tilde \rate_j.
\end{align*}
\end{proof}

\begin{cor}\label{cor:disjoint nullhypo direct exp}
Let $\A\subseteq\S(\hil)$ and $\rho_1,\ldots,\rho_r\in\S(\hil)$ with some $r\in\bN\setminus\{1\}$ be such that $\rho_j^0\wedge\rho_k^0=0$
for all $j,k\in[r]$, $j\ne k$. Then for any $\rate\in(0,+\infty)$,
\begin{align}\label{eq:composite direct exp equality1}
\dli_{\rate}(\{\rho_j\}_{j\in[r]}\|\A)=\min_{j\in[r]}\dli_{\rate}(\rho_j\|\A).
\end{align}
Moreover, if for every $j\in[r]$, a test sequence $T_{j,n}\in\bP(\hil^{\otimes n})$, $n\in\bN$, 
achieves $(\rate,\dli_{\rate}(\{\rho_j\}\|\A))$ for the composite i.i.d.~state discrimination problem 
$\N_j:=\{\rho_j\}$ vs.~$\A$,
then for every small enough $\ep>0$,  the test sequence 
$T_n:=\vee_{j\in[r]}(T_{j,n})_{(\rho_j^{\otimes n})^0, \ep}$, $n\in\bN$, achieves the exponent pair 
$(\rate,\dli_{\rate}(\{\rho_j\}_{j\in[r]}\|\A))$ for the composite i.i.d.~state discrimination problem $\{\rho_j\}_{j\in[r]}$ vs. $\A$.
\end{cor}
\begin{proof}
The inequality $\dli_{\rate}(\{\rho_j\}_{j\in[r]}\|\A)\le
\min_{j\in[r]}\dli_{\rate}(\rho_j\|\A)$ holds trivially, 
even without any assumption on the supports of the $\rho_j$; see, e.g., \cite[Section II.C]{MWSz2020} for details. 
In particular, the equality in \eqref{eq:composite direct exp equality1} holds trivially when 
$\min_{j\in[r]}\dli_{\rate}(\rho_j\|\A)=0$, and therefore we assume 
$\min_{j\in[r]}\dli_{\rate}(\rho_j\|\A)>0$ for the rest. 
In this case, for every $j\in[r]$, the exponent pair
$(\rate, \dli_{\rate}(\rho_j\|\A))$ is achievable for 
the composite i.i.d.~state discrimination problem 
$\N_j:=\{\rho_j\}$ vs.~$\A$, whence, by Proposition \ref{prop:disjoint nullhypo direct exp}, 
$(\rate,\min_{j\in[r]}\dli_{\rate}(\rho_j\|\A))$ is an achievable exponent pair
for the composite i.i.d.~state discrimination problem 
$\{\rho_j\}_{j\in[r]}$ vs.~$\A$.
The assertion about the test sequence achieving it also follows immediately from 
Proposition \ref{prop:disjoint nullhypo direct exp}.
\end{proof}

\begin{prop}\label{prop:disjoint nullhypo direct exp2}
Let $\N\subseteq\S(\hil)$ and $\sigma_1,\ldots,\sigma_m\in\S(\hil)$
with $m\in\bN\setminus\{1\}$
 be such that 
$\sigma_j^0\wedge\sigma_k^0=0$
for all $j,k\in[m]$, $j\ne k$. 
Assume that for every $j\in[m]$, a test sequence $T_{j,n}\in\bP(\hil^{\otimes n})$, $n\in\bN$, 
achieves a direct exponent pair
$(\rate_j,\tilde \rate_j)$ for the composite i.i.d.~state discrimination problem $\N$
vs.~$\A_j:=\{\sigma_j\}$.  
Then for every small enough $\ep>0$, the test sequence 
$T_n:=I-\vee_{j\in[m]}(I-T_{j,n})_{(\sigma_j^{\otimes n})^0, \ep}$ achieves the exponent pair 
$(\min_j \rate_j,\min_j\tilde \rate_j)$ for the composite i.i.d.~state discrimination problem $\N$ vs. $\{\sigma_j\}_{j=1}^m$.
\end{prop}
\begin{proof}
The proof goes the same way as for Proposition \ref{prop:disjoint nullhypo direct exp}, with the 
interchange of the roles of the $\rho$ and the $\sigma$ operators together with $T\leftrightarrow I-T$. We write out the details below for completeness.

Let $t>1$ be fixed, 
$\ep(m,t)$ be as in Lemma \ref{lemma:support sum domination}, and $\ep\in(0,\ep(m,t)/2)$.
For every $j\in[m]$ and every $n\in\bN$, define 
$S_{j,n}:=I-T_{j,n}$, $\tilde S_{j,n}:=(S_{j,n})_{(\sigma_j^{\otimes n})^0,\ep}$, 
$S_n:=\vee_{j\in[m]}\tilde S_{j,n}$, $T_n:=I-S_n$.
By Proposition \ref{prop:disjpoint supports},
\begin{align}
&\Tr\sigma_j^{\otimes n} \underbrace{(I-S_n)}_{=T_n} 
\le 
\frac{1}{\ep^2} \Tr\sigma_j^{\otimes n} \underbrace{(I-S_{j,n})}_{=T_{j,n}},
\ds\ds\ds\ds j\in [m],\ds n\in\bN,\label{eq:disjoint alternative proof1}\\
&I-T_n=S_n\le t \sum_{j\in [r]} S_{j,n}=t \sum_{j\in [r]} (I-T_{j,n}),\ds\ds\ds \text{for every large enough }n.
\label{eq:disjoint alternative proof2}
\end{align}
Thus,
\begin{align*}
\liminf_{n\to+\infty}-\frac{1}{n}\log\max_{j\in[m]}\Tr\sigma_j^{\otimes n}T_n
&\ge
\liminf_{n\to+\infty}-\frac{1}{n}\log\max_{j\in[m]}\Tr\sigma_j^{\otimes n}T_{j,n}\\
&=
\min_{j\in[m]}\liminf_{n\to+\infty}-\frac{1}{n}\log\Tr\sigma_j^{\otimes n}T_{j,n}
>\min_{j\in[m]}\rate_j,
\end{align*}
where the first inequality is due to \eqref{eq:disjoint alternative proof1}, and the 
rest are obvious. Similarly,
\begin{align*}
\liminf_{n\to+\infty}-\frac{1}{n}\log\sup_{\rho\in\N}\Tr\rho^{\otimes n}(I-T_n)
&\ge
\liminf_{n\to+\infty}-\frac{1}{n}\log\sup_{\rho\in\N}\Tr\rho^{\otimes n}\sum_{j\in[m]}(I-T_{j,n})\\
&\ge
\liminf_{n\to+\infty}-\frac{1}{n}\log\max_{j\in[m]}\sup_{\rho\in\N}\Tr\rho^{\otimes n}(I-T_{j,n})\\
&=
\min_{j\in[m]}\liminf_{n\to+\infty}-\frac{1}{n}\log\sup_{\rho\in\N}\Tr\rho^{\otimes n}(I-T_{j,n})\\
&\ge\min_{j\in[m]}\tilde \rate_j,
\end{align*}
where the first inequality is due to Proposition \ref{prop:disjpoint supports} \ref{modifed projections2}, and the rest are obvious. 
\end{proof}

\begin{cor}\label{cor:disjoint nullhypo direct exp2}
Let $\N\subseteq\S(\hil)$ and $\sigma_1,\ldots,\sigma_m\in\S(\hil)$ 
with $m\in\bN\setminus\{1\}$ be such that $\sigma_j^0\wedge\sigma_k^0=0$
for all $j,k\in[m]$, $j\ne k$. Then for any $\rate\in(0,+\infty)$,
\begin{align*}
\dli_{\rate}(\N\|\{\sigma_j\}_{j\in[m]})=\min_{j\in[m]}\dli_{\rate}(\N\|\sigma_j).
\end{align*}
Moreover, if for every $j\in[m]$, a test sequence $T_{j,n}\in\bP(\hil^{\otimes n})$, $n\in\bN$, 
achieves $(\rate,\dli_{\rate}(\N\|\{\sigma_j\}))$ for the composite i.i.d.~state discrimination problem 
$\N$ vs.~$\A_j:=\{\sigma_j\}$,
then for every small enough $\ep>0$, the test sequence 
$T_n:=I-\vee_{j\in[m]}(I-T_{j,n})_{(\sigma_j^{\otimes n})^0, \ep}$, $n\in\bN$, achieves the exponent pair 
$(\rate,\direct_{\rate}(\N\|\{\sigma_j\}_{j\in[m]}))$ for the composite i.i.d.~state discrimination problem 
$\N$ vs. $\{\sigma_j\}_{j\in[m]}$.
\end{cor}
\begin{proof}
The proof goes exactly the same way as for Corollary \ref{cor:disjoint nullhypo direct exp}, and hence we omit it.
\end{proof}

Corollaries \ref{cor:disjoint nullhypo direct exp} and \ref{cor:disjoint nullhypo direct exp2} yield immediately the following:
\begin{cor}
Let $\rho_1,\ldots,\rho_r,\sigma_1,\ldots,\sigma_m\in\S(\hil)$ be such that $\rho_j^0\wedge\rho_k^0=0$
for all $j,k\in[r]$, $j\ne k$, and $\sigma_j^0\wedge\sigma_k^0=0$
for all $j,k\in[m]$, $j\ne k$. 
Then for every $\rate\in(0,+\infty)$,
\begin{align*}
\dli_{\rate}(\{\rho_j\}_{j\in[r]}\|\{\sigma_k\}_{k\in[m]})
=
\min_{j\in[r],k\in[m]}\dli_{\rate}(\rho_j\|\sigma_k).
\end{align*}
\end{cor}
\begin{proof}
We have 
\begin{align*}
\dli_{\rate}(\{\rho_j\}_{j\in[r]}\|\{\sigma_k\}_{k\in[m]})
=
\min_{j\in[r]}\dli_{\rate}(\rho_j\|\{\sigma_k\}_{k\in[m]})
=
\min_{j\in[r],k\in[m]}\dli_{\rate}(\rho_j\|\sigma_k),
\end{align*}
where the first equality follows from
Corollary \ref{cor:disjoint nullhypo direct exp}, and 
the second one from Corollary \ref{cor:disjoint nullhypo direct exp2}.
\end{proof}

\bibliography{bibliography_GM_AM}

\end{document}